\documentclass{lmcs}
\usepackage[utf8]{inputenc}
\pdfoutput=1

\usepackage{lastpage}
\lmcsdoi{19}{1}{16}
\lmcsheading{}{\pageref{LastPage}}{}{}%
{Mar.~15,~2022}{Mar.~06,~2023}{}

\pdfoutput=1

\usepackage{hyperref} 
\usepackage{cleveref,thm-restate}

\newcommand{\renewtheorem}[1]{%
  \expandafter\let\csname #1\endcsname\relax
  \expandafter\let\csname c@#1\endcsname\relax
  \expandafter\let\csname end#1\endcsname\relax
  \newtheorem{#1}%
}

\theoremstyle{plain}

\renewtheorem{thm}{Theorem}[section]
\renewtheorem{cor}[thm]{Corollary}
\renewtheorem{lem}[thm]{Lemma}
\renewtheorem{prop}[thm]{Proposition}
\renewtheorem{asm}[thm]{Assumption}

\theoremstyle{thmC}
\renewtheorem{thmC}[thm]{Theorem}
\renewtheorem{lemC}[thm]{Lemma}

\theoremstyle{definition}

\renewtheorem{rem}[thm]{Remark}
\renewtheorem{rems}[thm]{Remarks}
\renewtheorem{exa}[thm]{Example}
\renewtheorem{exas}[thm]{Examples}
\renewtheorem{defi}[thm]{Definition}
\renewtheorem{conv}[thm]{Convention}
\renewtheorem{conj}[thm]{Conjecture}
\renewtheorem{prob}[thm]{Problem}
\renewtheorem{oprob}[thm]{Open Problem}
\renewtheorem{algo}[thm]{Algorithm}
\renewtheorem{obs}[thm]{Observation}
\renewtheorem{qu}[thm]{Question}
\renewtheorem{fact}[thm]{Fact}
\renewtheorem{pty}[thm]{Property}
\renewtheorem{clm}[thm]{Claim}

\usepackage{amsmath,amsthm}
\usepackage{graphicx}
\usepackage{graphics}
\usepackage{latexsym}
\usepackage{amssymb}
\usepackage{stmaryrd}
\usepackage{ifthen}
\usepackage{array}
\usepackage{mathrsfs}
\usepackage{multirow}
\usepackage{bbm}

\usepackage{pgf}
\usepackage{pgfplots}
\pgfplotsset{compat=1.15}
\usepackage{tikz, mathtools, thm-restate, stmaryrd}
\usetikzlibrary{arrows,calc,automata,backgrounds,decorations,intersections,pgfplots.fillbetween,shapes.misc}
\tikzset{every picture/.style={thick,>=angle 60}}
\tikzset{MDPrand/.style={draw,circle,minimum size=11*1.5,inner sep=0}}
\tikzset{MDPcont/.style={draw,rectangle,minimum size=9*1.5,inner sep=0}}
\tikzset{MDPbad/.style={fill=red}}
\definecolor{myorange}{RGB}{255, 128, 0}

\usetikzlibrary{shapes.multipart,shapes.misc}
\usepgflibrary{decorations.pathreplacing}
\usetikzlibrary{arrows,calc,topaths}
\usetikzlibrary{automata,shapes.arrows, backgrounds, intersections, positioning}
\usetikzlibrary{shadows}

\usetikzlibrary{decorations.pathmorphing}
\usetikzlibrary{decorations.shapes}
\usetikzlibrary{shapes}
\usetikzlibrary{backgrounds}
\usetikzlibrary{fit}
\usetikzlibrary{arrows}
\usetikzlibrary{calc}
\usetikzlibrary{mindmap}
\usetikzlibrary{matrix}

\usepackage{mathtools}
\usepackage{thmtools,thm-restate}
\usepackage{xcolor}

\usepackage{bbding}

\hypersetup{colorlinks,linkcolor=blue,citecolor=blue,urlcolor=blue}

\newcommand{\+}[1]{\mathbb{#1}}

\newcommand{\N}{\+{N}}
\newcommand{\R}{\+{R}}
\newcommand{\x}{\times}

\usepackage{wasysym} 
\newcommand{\rsymbol}{\ocircle}
\newcommand{\zsymbol}{\Box}
\newcommand{\zstates}{\states_\zsymbol}
\newcommand{\rstates}{\states_\rsymbol}
\newcommand{\reachset}{T}

\newcommand{\eqby}[2][=]{\stackrel{\text{{\tiny{#2}}}}{#1}}
\newcommand{\eqdef}{\eqby{def}}
\newcommand{\defeq}{\eqdef}

\newcommand{\eps}{\varepsilon}

\newcommand{\problemx}[3]{
\par\noindent\underline{\sc#1}\par\nobreak\vskip.2\baselineskip
\begingroup\clubpenalty10000\widowpenalty10000
\setbox0\hbox{\bf INPUT:\ }\setbox1\hbox{\bf QUESTION:\ }
\dimen0=\wd0\ifnum\wd1>\dimen0\dimen0=\wd1\fi
\vskip-\parskip\noindent
\hbox to\dimen0{\box0\hfil}\hangindent\dimen0\hangafter1\ignorespaces#2\par
\vskip-\parskip\noindent
\hbox to\dimen0{\box1\hfil}\hangindent\dimen0\hangafter1\ignorespaces#3\par
\endgroup}

\newcommand{\Runs}[2]{{\textit{Runs}_{#1}^{#2}}}
\newcommand{\pRuns}[2]{{\textit{pRuns}_{#1}^{#2}}}
\newcommand{\dist}{\mathcal{D}}

\newcommand{\always}{{\sf G}}
\newcommand{\eventually}{{\sf F}}

\newcommand{\hide}[1]{}

\newcommand{\lrc}[1]{(#1)}

\newcommand{\ignore}[1]{}

\newcommand{\tuple}[1]{\lrc{#1}}

\newcommand{\denotationof}[2]{\llbracket #1\rrbracket^{#2}}

\newcommand{\mdp}{{\mathcal M}}

\newcommand{\mdptuple}{\tuple{\states,\zstates,\rstates,\transition,\probp,r}}

\newcommand{\states}{S}

\newcommand{\state}{s}

\newcommand{\transition}{{\longrightarrow}}

\newcommand{\probp}{P}

\newcommand{\complementof}[1]{\overline{#1}}

\newcommand{\play}{\rho}
\newcommand{\playset}{{\mathfrak R}}

\newcommand{\partialplay}{\rho}

\newcommand{\zstrat}{\sigma}

\newcommand{\zstratset}{\Sigma}

\newcommand{\memory}{{\sf M}}
\newcommand{\memconf}{{\sf m}}

\newcommand{\expectval}{{\mathcal E}}
\newcommand{\probm}{{\mathcal P}}

\newcommand{\formula}{{\varphi}}

\newcommand{\valueof}[2]{{\mathtt{val}_{#1}(#2)}}

\mathchardef\mhyphen="2D 

\newcommand{\liminfmpobj}{{\it MP}_{\liminf\ge 0}}
\newcommand{\liminftpobj}{{\it TP}_{\liminf\ge 0}}
\newcommand{\liminfppobj}{{\it PP}_{\liminf\ge 0}}

\newcommand{\transience}{\mathtt{Transience}}
\newcommand{\reward}{\mathit{r}}

\title[Strategy Complexity of Point, Mean and Total Payoff]
{Strategy Complexity of Point Payoff, Mean Payoff and Total Payoff Objectives in Countable MDPs}

\author[R.~Mayr]{Richard Mayr}
\author[E.~Munday]{Eric Munday}

\address{University of Edinburgh, School of Informatics, LFCS, Edinburgh, UK}


\begin{document}

\newcommand\blfootnote[1]{%
  \begingroup
  \renewcommand\thefootnote{}\footnote{#1}%
  \addtocounter{footnote}{-1}%
  \endgroup
}
\blfootnote{Extended version of results presented at CONCUR 2021.}

\begin{abstract}
  We study countably infinite Markov decision processes (MDPs)
  with real-valued transition rewards.
  Every infinite run induces the following sequences of payoffs:
  1. Point payoff (the sequence of directly seen transition rewards),
  2. Mean payoff (the sequence of the sums of all rewards so far, divided by
  the number of steps), and
  3. Total payoff (the sequence of the sums of all rewards so far).
  For each payoff type, the objective is to maximize the probability
  that the $\liminf$ is non-negative.

  We establish the complete picture of the strategy complexity of
  these objectives, i.e., how much memory is necessary and sufficient for
  $\eps$-optimal (resp.\ optimal) strategies.
  Some cases can be won with memoryless deterministic strategies,
  while others require a step counter, a reward counter, or both.
\end{abstract}

\maketitle

\section{Introduction}\label{sec:intro}
\paragraph{Background.}
Countably infinite Markov decision processes (MDPs) are a standard model for dynamic systems that
exhibit both stochastic and controlled behavior; see, e.g., standard
textbooks~\cite{Puterman:book,MaitraSudderth:DiscreteGambling,DubbinsSavage:2014,raghaven2012} and references therein.
Some fundamental results and proof techniques for countable MDPs were
established in the framework of Gambling Theory~\cite{DubbinsSavage:2014,MaitraSudderth:DiscreteGambling}. See also Ornstein's
seminal paper on stationary strategies~\cite{Ornstein:AMS1969}.
Further applications include control theory~\cite{blondel2000survey,NIPS2004_2569},
operations research
and finance~\cite{nowak2005,Flesch:JOTA2020,bauerle2011finance,schal2002markov},
artificial intelligence and machine
learning~\cite{sutton2018reinforcement,sigaud2013markov},
and formal verification~\cite{KMST2020c,ACMSS2016,BBEKW2010,EWY2010,BBEK:IC2013,
EY:JACM2015, ModCheckHB18,ModCheckPrinciples08}.
The latter works often use countable MDPs to describe unbounded structures in
computational models such as stacks/recursion, counters, queues, etc.

An MDP is a directed graph where states are either random or controlled.
In a random state the next state is chosen according to a fixed probability distribution.
In a controlled state the controller can choose
a distribution over all possible successor states.
By fixing a strategy for the controller (and an initial state), one obtains a probability space
of runs of the MDP\@. The goal of the controller is to optimize the expected value of
some objective function on the runs.
The type of strategy necessary to achieve an $\eps$-optimal (resp.\ optimal)
value for a given objective is called its \emph{strategy complexity}.

\paragraph{Transition rewards and liminf objectives.}
MDPs are given a reward structure by assigning a real-valued
(resp.\ integer or rational) reward to
each transition. Every run then induces an infinite sequence of
seen transition rewards $r_0r_1r_2\dots$.
We consider the $\liminf$
of this sequence, as well as
two other important derived sequences.
\begin{enumerate}
\item
The point payoff considers the
$\liminf$ of the sequence $r_0 r_1r_2 \dots$
directly.
\item
The mean payoff considers the
$\liminf$ of the sequence
$\left\{\frac{1}{n}\sum_{i=0}^{n-1} r_i\right\}_{n \in \N}$, i.e., the mean of
all rewards seen
so far in an expanding prefix of the run.
\item
The total payoff considers the
$\liminf$ of the sequence
$\left\{\sum_{i=0}^{n-1} r_i\right\}_{n \in \N}$, i.e., the sum of all rewards seen
so far.
\end{enumerate}
For each of the three cases above, the $\liminf$ threshold objective is
to maximize the probability that the $\liminf$ of the respective type of sequence
is $\ge 0$.

\paragraph{Our contribution.}
We establish the strategy complexity of all the
$\liminf$ threshold objectives above for \emph{countably infinite} MDPs.
(For the simpler case of finite MDPs, see the paragraph on related work below.)
We show the amount and type of memory
that is necessary and sufficient for $\eps$-optimal strategies
(and optimal strategies, where they exist).

Classes of strategies are defined via the amount and type of memory
used, and whether they are randomized or deterministic.
Some canonical types of memory for strategies are the following:
No memory (also called memoryless or positional),
finite memory, a step counter (i.e., a discrete clock),
a reward counter (i.e., a variable that records the sum of all transition
rewards seen so far)
and general infinite memory.
Strategies using only a step counter are also called \emph{Markov strategies}~\cite{Puterman:book}.
The reward counter has the same type as the transition rewards in the MDP, i.e.,
integers, rationals or reals.
Moreover, there can be combinations of these, e.g., a step counter plus some
finite general purpose memory.
Other types of memory are possible, e.g., an unbounded stack or a queue, but they are less
common in the literature.

To establish an upper bound $X$ on the strategy complexity of an objective
in countable MDPs, it suffices to prove that there always exist good
($\eps$-optimal, resp.\ optimal) strategies in some class of strategies $X$.
Lower bounds on the strategy complexity of an objective
can only be established in the sense of proving that good
strategies for the objective do not exist in some classes $Y$, $Z$, etc.
See \Cref{putermanexample} for an example.

\begin{figure}[htbp]
\begin{center}
    \begin{tikzpicture}

    \node[draw, minimum height=0.7cm, minimum width=0.7cm] (S1) at (1,0) {$s_{1}$};
    \node[draw, minimum height=0.7cm, minimum width=0.7cm] (S2) at (2.5,0) {$s_{2}$};
    \node[draw, minimum height=0.7cm, minimum width=0.7cm] (S3) at (4,0) {$s_{3}$};
    \node[draw, minimum height=0.7cm, minimum width=0.7cm] (S4) at (5.5,0) {$s_{4}$};

    \node[draw, minimum height=0.7cm, minimum width=0.7cm] (SK) at (7,0) {$s_{k}$};
    \node (SI1) at (8.5,0) {};


    \draw[->,>=latex] (S1) -- (S2) node[above, midway]{\small $-1$};
    \draw[->,>=latex] (S2) -- (S3) node[above, midway]{\small $-1$};
    \draw[->,>=latex] (S3) -- (S4) node[above, midway]{\small $-1$};
    \draw[->,>=latex, dotted, thick] (S4) -- (SK) node[above, midway]{};
    \draw[dotted, thick] (SK) -- (SI1);

	\path (S1) edge[->,>=latex, loop below] node[below] {\small $-1$} (S1);
	\path (S2) edge[->,>=latex, loop below] node[below] {\small $-\dfrac{1}{2}$} (S2);
	\path (S3) edge[->,>=latex, loop below] node[below] {\small $-\dfrac{1}{3}$} (S3);
	\path (S4) edge[->,>=latex, loop below] node[below] {\small $-\dfrac{1}{4}$} (S4);
	\path (SK) edge[->,>=latex, loop below] node[below] {\small $-\dfrac{1}{k}$} (SK);

    \end{tikzpicture}
    \caption{Adapted from~\cite[Example 8.10.2]{Puterman:book}. While there is
      no optimal MD (memoryless deterministic) strategy, the following
      strategy is optimal for lim inf/lim sup mean payoff:
      Loop $\exp(\exp(k))$ many
      times
      in state $s_k$ for all $k$.
      In this particular example, this can be implemented with either
      just a step counter or just a reward counter, but
      in general both are needed; cf.~\Cref{table:allresults}.}%
    \label{putermanexample}
\end{center}
    \end{figure}
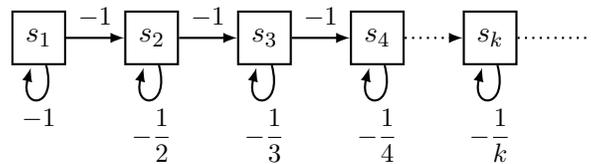

Note that different classes of strategies are not always comparable,
for several reasons.
First, different types of memory may be incomparable. E.g.,
a step counter uses infinite memory, but it is updated in a very particular
way, and thus it does not subsume a finite general purpose memory.
Second, randomized strategies are more general than deterministic ones if they
use the same memory, but not if the memory is different.
E.g., randomized positional strategies are incomparable to deterministic
strategies with finite memory (or a step counter).
Since strategy classes are not always comparable, there
there can be cases with several incomparable upper/lower bounds.

Moreover, there is no weakest type of infinite memory with restricted use.
Hence, upper and lower bounds on the strategy complexity of an objective
can be only be tight \emph{relative} to the
considered alternative strategy classes, e.g., the canonical classes mentioned above.

Our results are summarized in \Cref{table:allresults}.
By Rand(X) (resp.\ Det(X)) we denote the classes of randomized
(resp.\ deterministic) strategies that use memory of size/type X.
SC denotes a step counter, RC denotes a reward counter and F denotes arbitrary finite memory.
Positional/memoryless means that no memory is used.
The simplest type are memoryless deterministic (MD) strategies.
The results depend on the type of objective (point, total, or
mean payoff) and on whether the MDP is finitely or infinitely
branching. For our objectives, the strategy complexities  of $\eps$-optimal
and optimal strategies (where they exist) coincide, but the proofs are different.

For clarity of presentation,
our counterexamples use large transition rewards and high degrees of
branching.
However, the results can be strengthened such that the lower bounds hold even for just binary
branching MDPs with rational transition probabilities and transition rewards in $\{-1,0,1\}$;
cf.~\Cref{sec:strengthening}.

\begin{table*}[hbtp]
{\small
\centering
\begin{tabular}{|l||c|c|c|}\hline
      & Point payoff & Mean payoff & Total payoff \\ \hline
$\eps$-opt., fin.\ br. & Det(Positional)~\ref{finpointpayoff}  & Det(SC+RC)~\ref{mpepsupper},~\ref{infinitesummary},~\ref{mpstepepslower}  & Det(RC)~\ref{fintpepsupper},~\ref{infinitesummarytp}  \\ \hline
opt., fin.\ br.     & Det(Positional)~\ref{finoptupperpp}  & Det(SC+RC)~\ref{infoptuppermp},~\ref{almostsummary},~\ref{mpstepoptlower}  & Det(RC)~\ref{finoptupper},~\ref{almostsummarytp}   \\ \hline
$\eps$-opt., inf.\ br.   & Det(SC)~\ref{infpointpayoff},~\ref{infbranchsteplower} & Det(SC+RC)~\ref{mpepsupper},~\ref{infinitesummary},~\ref{mpstepepslower}   & Det(SC+RC)~\ref{inftpepsupper},~\ref{infinitesummarytp},~\ref{infbranchsteplowertp}  \\ \hline
opt., inf.\ br.   & Det(SC)~\ref{infoptupperpp},~\ref{infbranchsteplower} & Det(SC+RC)~\ref{infoptuppermp},~\ref{almostsummary},~\ref{mpstepoptlower}    & Det(SC+RC)~\ref{infoptuppertp},~\ref{almostsummarytp},~\ref{infbranchsteplowertp}  \\ \hline
\end{tabular}
}
\caption{Combined upper and lower bounds on the strategy complexity of $\eps$-optimal (resp.\ optimal)
strategies for point, mean and total payoff objectives in
finitely branching and infinitely branching MDPs.
All strategies are deterministic and randomization does not help.
For each result, we list the numbers of the theorems that show the
upper and lower bounds on the strategy complexity.
The lower bounds hold wrt.\ the canonical strategies mentioned above.
Explicit lower bounds are listed in the tables in the following sections.
The lower bounds hold even for integer transition
rewards. The upper
bounds hold even for real-valued transition rewards.
}\label{table:allresults}
\end{table*}

Some complex new proof techniques are developed to show these results.
E.g., the examples showing the lower bound in cases where both a step counter
and a reward counter are required use a finely tuned tradeoff between different
risks that can be managed with both counters, but not with just one
counter plus arbitrary finite memory.
The strategies showing the upper bounds need to take into account
convergence effects, e.g., the sequence of point rewards
$-1/2, -1/3, -1/4, \dots$ \emph{does} satisfy $\liminf$ $\ge 0$, i.e.,
one cannot assume that rewards are integers.

\paragraph{Related work.}
Mean payoff objectives for \emph{finite} MDPs have been widely
studied; cf.~survey in~\cite{chatterjee2011games}.
There exist optimal MD strategies for $\liminf$ mean payoff
(which are also optimal for $\limsup$ mean payoff since the transition rewards are bounded),
and the associated computational problems
can be
solved in polynomial time~\cite{chatterjee2011games,Puterman:book}.
Similarly, see~\cite{CDH:ICALP2009} for a survey on
$\limsup$ and $\liminf$ point payoff objectives in finite stochastic games and MDPs,
where there also exist optimal MD strategies, and the more recent paper
by Flesch, Predtetchinski and Sudderth~\cite{FPS:2018} on simplifying optimal strategies.

All this does \emph{not} carry over to countably infinite MDPs.
Optimal strategies need not exist
(not even for much simpler objectives),
($\eps$-)optimal strategies can require
infinite memory, and computational problems are not defined in general, since
a countable MDP need not be finitely presented~\cite{Puterman:book,MaitraSudderth:DiscreteGambling,DubbinsSavage:2014,raghaven2012,KMSW2017}.
Moreover, attainment for $\liminf$ mean payoff need not coincide with attainment
for $\limsup$ mean payoff, even for very simple examples.
E.g., consider the acyclic infinite graph with
transitions $s_n \rightarrow s_{n+1}$ for all $n \in \N$
with reward $(-1)^n2^n$ in the $n$-th step, which yields a $\liminf$ mean payoff
of $-\infty$ and a $\limsup$ mean payoff of $+\infty$.

Mean payoff objectives for countably infinite MDPs have been considered in~\cite[Section 8.10]{Puterman:book}, e.g.,~\cite[Example 8.10.2]{Puterman:book}
(adapted in \Cref{putermanexample})
shows that there are no optimal MD (memoryless deterministic)
strategies for $\liminf$/$\limsup$ mean payoff.~\cite[Counterexample 1.3]{Ross:1983} shows that there are not even
$\eps$-optimal memoryless randomized strategies for $\liminf$/$\limsup$ mean payoff.
(We show much stronger lower/upper bounds; cf.~\Cref{table:allresults}.)

Sudderth~\cite{Sudderth:2020} considered an objective on countable MDPs that is
related to our point payoff threshold objective. However, instead of maximizing the
probability that the $\liminf$/$\limsup$ is non-negative, it asks to maximize the
\emph{expectation} of the $\liminf$/$\limsup$ point payoffs, which is a different
problem (e.g., it can tolerate a high probability of a negative $\liminf$/$\limsup$
if the remaining cases have a huge positive $\liminf$/$\limsup$).
Hill \& Pestien~\cite{Hill-Pestien:1987} showed the existence of good
randomized Markov strategies for the $\limsup$ of the \emph{expected}
average reward up-to step $n$ for growing $n$, and for the \emph{expected}
$\liminf$ of the point payoffs.

\section{Preliminaries}\label{sec:prelim}
\paragraph{Markov decision processes.}
A \emph{probability distribution} over a countable set $S$ is a function
$f:\states\to[0,1]$ with $\sum_{\state\in \states}f(\state)=1$.
We write
$\dist(\states)$ for the set of all probability distributions over $\states$.
A \emph{Markov decision process} (MDP) $\mdp=\mdptuple$ consists of
a countable set~$\states$ of \emph{states},
which is partitioned into a set~$\zstates$ of \emph{controlled states}
and  a set~$\rstates$ of \emph{random states},
a  \emph{transition relation} $\transition\subseteq\states\x\states$,
and a  \emph{probability function}~$\probp:\rstates \to \dist(\states)$.
We  write $\state\transition{}\state'$ if $\tuple{\state,\state'}\in \transition$,
and  refer to~$s'$ as a \emph{successor} of~$s$.
We assume that every state has at least one successor.
The probability function~$P$  assigns to each random state~$\state\in \rstates$
a probability distribution~$P(\state)$ over its (non-empty) set of successor states.
A \emph{sink in $\mdp$} is a subset $T \subseteq \states$ closed under the $\transition$ relation,
that is,  $\state \in \reachset$ and  $\state\transition\state'$ implies that $\state'\in T$.

An MDP is \emph{acyclic} if the underlying directed graph~$(S,\transition)$ is acyclic, i.e.,
there is no directed cycle.
It is  \emph{finitely branching}
if every state has finitely many successors
and \emph{infinitely branching} otherwise.
An MDP without controlled states
($\zstates=\emptyset$) is called a \emph{Markov chain}.

In order to specify our point/mean/total payoff objectives (see below), we
define a function $\reward: \states \times \states \to \R$ that assigns numeric
rewards to transitions.

\paragraph{Strategies and Probability Measures.}
A \emph{run}~$\play$ is an  infinite sequence of states and transitions
$\state_0e_0\state_1e_1\cdots$
such that $e_i = (\state_i, \state_{i+1}) \in \transition$
for all~$i\in \mathbb{N}$.
Let $\Runs{\mdp}{\state_0}$
be the set of all runs from $\state_0$ in the MDP $\mdp$.
A \emph{partial run} is a finite prefix of a run,
$\pRuns{\mdp}{\state_0}$ is the set of all partial runs
from $\state_0$ and $\pRuns{\mdp}{}$ the set of partial runs from any state.

We write~$\play_s(i)\eqdef\state_i$ for the $i$-th state along~$\play$
and $\play_e(i)\eqdef e_i$ for the $i$-th transition along~$\play$.
We sometimes write runs as $\state_0\state_1\cdots$, leaving the transitions implicit.
We say that a (partial) run $\play$ \emph{visits} $\state$ if
$\state=\play_s(i)$ for some $i$, and that~$\play$ starts in~$s$ if $\state=\play_s(0)$.

A \emph{strategy}
is a function $\zstrat:\pRuns{\mdp}{}\!\cdot\!\zstates \to \dist(S)$ that
assigns to partial runs $\partialplay\state$, where $\state \in \zstates$,
a distribution over the successors~$\{\state'\in \states\mid \state \transition{} \state'\}$.
The set of all strategies  in $\mdp$ is denoted by $\zstratset_\mdp$
(we omit the subscript and write~$\zstratset$ if $\mdp$ is clear from the context).
A (partial) run~$\state_0e_0\state_1e_1\cdots$ is consistent with a strategy~$\zstrat$
if for all~$i$
either $\state_i \in \zstates$ and $\zstrat(\state_0e_0\state_1e_1\cdots\state_i)(\state_{i+1})>0$,
or
$\state_i \in \rstates$ and $\probp(\state_i)(\state_{i+1})>0$.
%

An MDP $\mdp=\mdptuple$, an initial state $\state_0\in \states$, and a strategy~$\zstrat$
induce a probability space in which the outcomes are runs starting in $\state_0$
and with measure $\probm_{\mdp,\state_0,\zstrat}$
defined as follows.
It is first defined on \emph{cylinders} $s_0 e_0 s_1 e_1 \ldots
s_n \Runs{\mdp}{s_n}$:
if $s_0 e_0 s_1 e_1 \ldots s_n$
is not a partial run consistent with~$\zstrat$ then
$\probm_{\mdp,\state_0,\zstrat}(s_0 e_0 s_1 e_1 \ldots s_n \Runs{\mdp}{s_n}) \eqdef 0$.
Otherwise, $\probm_{\mdp,\state_0,\zstrat}(s_0 e_0 s_1 e_1 \ldots
s_n \Runs{\mdp}{s_n}) \eqdef \prod_{i=0}^{n-1} \bar{\zstrat}(s_0 e_0
s_1 \ldots s_i)(s_{i+1})$, where $\bar{\zstrat}$ is the map that
extends~$\zstrat$ by $\bar{\zstrat}(w s) = \probp(s)$
for all partial runs $w s \in \pRuns{\mdp}{}\!\cdot\!\rstates$.
By Carath\'eodory's theorem~\cite{billingsley-1995-probability},
this extends uniquely to a probability measure~$\probm_{\mdp,\state_0,\zstrat}$ on
the Borel $\sigma$-algebra $\?F$ of subsets of~$\Runs{\mdp}{s_0}$.
Elements of $\?F$, i.e., measurable sets of runs, are called \emph{events} or \emph{objectives} here.
For $X\in\?F$ we will write $\complementof{X}\eqdef \Runs{\mdp}{s_0}\setminus X\in \?F$ for its complement
and $\expectval_{\mdp,\state_0,\zstrat}$
for the expectation wrt.~$\probm_{\mdp,\state_0,\zstrat}$.
We drop the indices if possible without
ambiguity.

\paragraph{Objectives.}
We consider objectives that are determined by a predicate on infinite runs.
We assume familiarity with the syntax and semantics of the temporal
logic LTL~\cite{CGP:book}.
Formulas are interpreted on the structure $(\states,\transition)$.
We use $\denotationof{\formula}{\state}$ to denote the set of runs starting from
$\state$ that satisfy the LTL formula $\formula$,
which is a measurable set~\cite{Vardi:probabilistic}.
We also write $\denotationof{\formula}{}$ for $\bigcup_{s \in S} \denotationof{\formula}{s}$.
Where it does not cause confusion we will
identify $\varphi$ and $\denotationof{\formula}{}$
and just write
$\probm_{\mdp,\state,\zstrat}(\formula)$
instead of
$\probm_{\mdp,\state,\zstrat}(\denotationof\formula\state)$.
The reachability objective of eventually visiting a set of states $X$ can be expressed by
$\denotationof{\eventually X}{} \eqdef \{\play\,|\, \exists
i.\, \play_s(i) \in X\}$.
Reaching $X$ within at most $k$ steps is expressed by
$\denotationof{\eventually^{\le k} X}{} \eqdef \{\play\,|\, \exists
i \le k.\, \play_s(i) \in X\}$.
The definitions for eventually visiting certain transitions are analogous.
The operator $\always$ (always) is defined as $\neg\eventually\neg$.
So the safety objective of avoiding $X$ is expressed by
$\always \neg X$.

We consider the following objectives.
\begin{itemize}
\item
The $\liminfppobj$ objective is to maximize the
probability that the $\liminf$ of the \emph{point} payoffs (the immediate
transition rewards) is $\ge 0$, i.e.,
$\liminfppobj \defeq
\{\rho \mid \liminf_{n\in \N} \reward(\rho_e(n)) \ge 0\}$.
\item
The $\liminfmpobj$ objective is to maximize the
probability that the $\liminf$ of the \emph{mean} payoff
is $\ge 0$, i.e.,
$\liminfmpobj\defeq
\{\rho \mid \liminf_{n\in \N}\frac{1}{n}\sum_{j=0}^{n-1}
\reward(\rho_e(j)) \ge 0\}$.
\item
The $\liminftpobj$ objective is to maximize the
probability that the $\liminf$ of the \emph{total} payoff (the sum of the
transition rewards seen so far) is $\ge 0$, i.e.,
$\liminftpobj \defeq
\{\rho \mid \liminf_{n\in \N} \sum_{j=0}^{n-1}\reward(\rho_e(j)) \ge 0\}$.
\end{itemize}
An objective $\formula$ is called \emph{shift invariant} in $\mdp$
if for every run
$\rho'\rho$ in $\mdp$ with some finite prefix $\rho'$ we have
$\rho'\rho \in \denotationof{\formula}{} \Leftrightarrow \rho \in \denotationof{\formula}{}$.
An objective is called a \emph{shift invariant objective} if it is shift invariant in every MDP\@.
$\liminfppobj$ and $\liminfmpobj$ are shift invariant objectives,
but $\liminftpobj$ is not.
Also $\liminfppobj$ is more general than co-B\"uchi.
(The special case of integer transition rewards coincides with co-B\"uchi,
since rewards $\le -1$ and accepting states can be encoded into each other.)

\paragraph{Strategy Classes.}
Strategies are in general  \emph{randomized} (R) in the sense that they take values in $\dist(\states)$.
A strategy~$\zstrat$ is \emph{deterministic} (D) if $\zstrat(\rho)$ is a Dirac
distribution for all $\rho$.
General strategies can be \emph{history dependent} (H), while others are
restricted by the size or type of memory they use, see below.
We consider certain classes of strategies:
\begin{itemize}
\item
A strategy $\zstrat$ is  \emph{memoryless}~(M) (also called \emph{positional})
if it can be implemented with a memory  of size~$1$.
We may view
M-strategies as functions $\zstrat: \zstates \to \dist(\states)$.
\item
A strategy~$\zstrat$ is \emph{finite memory}~(F) if
there exists a finite memory~$\memory$ implementing~$\zstrat$.
Hence FR stands for finite memory randomized.
\item
A \emph{step counter (SC) strategy} bases decisions only on
the current state and the number of steps taken so far, i.e., it uses
an unbounded integer counter that gets incremented by $1$ in every step.
Such strategies are also called \emph{Markov strategies}~\cite{Puterman:book}.
\item
\emph{$k$-bit Markov strategies} use $k$ extra bits of general purpose memory
in addition to a step counter~\cite{KMST2020c}.
\item
A \emph{reward counter (RC) strategy} uses infinite memory, but only in
the form of a counter that always contains the sum of all transition rewards
seen so far.
\item
A \emph{step counter + reward counter strategy} uses both a step counter and a reward counter.
\end{itemize}
Step counters and reward counters are very restricted forms of
memory, since the memory update is not directly under the control of the
player. These counters merely record an aspect of the partial run.

\paragraph{Memory and strategies.}
Here we give a formal definition how strategies use memory.
Let $\memory$ be a countable set of memory modes, and let $\tau: \memory\times \states \to \dist(\memory\times \states)$ be a function that meets the
following two conditions: for all modes $\memconf \in \memory$,
\begin{itemize}
	\item for all controlled states~$\state\in \zstates$,
	the distribution $\tau(\memconf,\state)$ is   over
	$\memory \times \{\state'\mid \state \transition{} \state'\}$.
	\item for all random states~$\state \in \rstates$,
        and $\state' \in \states$, we have $\sum_{\memconf'\in \memory} \tau(\memconf,\state)(\memconf',\state')=P(\state)(\state')$.
\end{itemize}

\noindent
The function~$\tau$ together with an initial memory mode~$\memconf_0$
induce a strategy~$\zstrat_{\tau}$
as follows.
Consider the Markov chain with the set~$\memory \times \states$ of states
and the  probability function~$\tau$.
A sequence $\rho=s_0 \cdots s_i$ corresponds to a set
$H(\rho)=\{(\memconf_0,s_0) \cdots (\memconf_i,s_i) \mid \memconf_0,\ldots, \memconf_i\in \memory\}$
of runs in this Markov chain.
Each $\rho s \in \state_0 \states^{*} \zstates$
induces a
probability distribution~$\mu_{\rho \state}\in \dist(\memory)$,
the  probability of   being in  state~$(\memconf,s)$
conditioned on having  taken some partial run
from~$H(\rho s)$.
We define~$\zstrat_{\tau}$ such that
$\zstrat_{\tau}(\rho \state)(\state')=\sum_{\memconf,\memconf'\in \memory} \mu_{\rho \state}(\memconf) \tau(\memconf,\state)(\memconf',\state') $
for all $\rho \state\in \states^{*} \zstates$ and all $\state' \in \states$.

We say that a strategy $\zstrat$ can be \emph{implemented} with
memory~$\memory$ if there exist~$\memconf_0 \in \memory$
and $\tau$ such that  $\zstrat_{\tau}=\zstrat$.

\paragraph{Optimal and $\eps$-optimal Strategies.}
Given an objective~$\formula$, the value of state~$s$ in an MDP~$\mdp$, denoted by
$\valueof{\mdp,\formula}{s}$, is the supremum probability of achieving~$\formula$.
 Formally, $\valueof{\mdp,\formula}{s} \eqdef\sup_{\sigma \in \Sigma} \probm_{\mdp,\state,\zstrat}(\formula)$ where $\Sigma$ is the set of all strategies.
For $\eps\ge 0$ and state~$s\in\states$, we say that a strategy is \emph{$\eps$-optimal} from $s$
if $\probm_{\mdp,\state,\zstrat}(\formula) \geq \valueof{\mdp,\formula}{s} -\eps$.
A $0$-optimal strategy is called \emph{optimal}.
An optimal strategy is \emph{almost-surely winning} if $\valueof{\mdp,\formula}{s} = 1$.
Considering an MD strategy as a function $\zstrat: \zstates \to \states$ and $\eps\ge 0$, $\zstrat$ is \emph{uniformly} $\eps$-optimal  (resp.~uniformly optimal) if it is $\eps$-optimal (resp.~optimal) from \emph{every} $s\in S$.

\bigskip
\paragraph{MDP variants.}
In order to show our results, we will sometimes use derived MDPs.
Given an MDP $\mdp$, we define three different MDPs $S(\mdp)$, $R(\mdp)$ and $A(\mdp)$.
These new MDPs will be used in order to reduce objectives to $\liminfppobj$
in a simpler setting with the step and/or reward counter encoded into the states.

\begin{defi}\label{def:encodestep}
Let $\mdp$ be an MDP with a given initial state $s_{0}$.
We construct the MDP $S(\mdp) \eqdef (S', \zstates', \rstates', \longrightarrow_{S(\mdp)}, P')$
that encodes the step counter into the states as follows:
\begin{itemize}
\item The state space of $S(\mdp)$ is
$S' \eqdef \{ (s,n) \mid s \in S \text{ and } n \in \N \}$.
Note that $S'$ is countable.
We write $s_{0}'$ for the initial state $(s_{0},0)$.
\item $\zstates' \eqdef \{ (s,n) \in S' \mid s \in \zstates \}$
and $\rstates' \eqdef S' \setminus \zstates'$.
\item The set of transitions in $S(\mdp)$ is
\[ \longrightarrow_{S(\mdp)} \eqdef
\{
\left( (s,n),(s',n+1) \right) \mid (s,n),(s',n+1) \in S',
 s \longrightarrow_{\mdp} s'
\}.
\]
\item $P': \rstates' \to \mathcal{D}(S')$ is defined such that
\[
P'(s,n)(s',m) \eqdef
    \begin{cases}
    P(s)(s') & \text{ if } (s,n) \longrightarrow_{S(\mdp)} (s',m) \\
    0 & \text{ otherwise }
    \end{cases}
\]

\item The reward $r((s,n) \longrightarrow_{S(\mdp)} (s',n+1)) \eqdef r(s \longrightarrow_{\mdp} s')$.
\end{itemize}

\noindent
By labeling the state with the path length from $s_0$, we effectively encode a
step counter into the MDP $S(\mdp)$.
\end{defi}

\begin{lem}\label{steptopoint}
Let $\mdp$ be an MDP with initial state $s_{0}$. Then given an MD strategy $\sigma'$ in $S(\mdp)$ attaining $c \in [0,1]$ for
$\liminfppobj$ from $(s_{0},0)$, there exists a strategy
$\sigma$ attaining $c$ for $\liminfppobj$ in $\mdp$ from $s_{0}$ which uses the same memory as $\sigma'$ plus a step counter.
\end{lem}

\begin{proof}
Let $\sigma'$ be an MD strategy in $S(\mdp)$ attaining $c \in [0,1]$ for
$\liminfppobj$ from $(s_{0},0)$.
We define a strategy $\sigma$ on $\mdp$ from $s_0$ that uses the same memory
as $\sigma'$ plus a step counter. Then $\sigma$ plays on $\mdp$ exactly like
$\sigma'$ plays on $S(\mdp)$, keeping the step counter in its memory instead
of in the state.
I.e., at a given state $s$ and step counter value $n$, $\sigma$ plays exactly as $\sigma'$
plays in state $(s,n)$
By our construction of $S(\mdp)$ and the definition of $\sigma$,
the sequences of point rewards seen by $\sigma'$ in runs on $S(\mdp)$
coincide with the sequences of point rewards seen by $\sigma$ in runs in $\mdp$.
Hence we obtain
$\probm_{S(\mdp),(s_0,0),\sigma'}(\liminfppobj)
= \probm_{\mdp,s_0,\sigma}(\liminfppobj)$
\end{proof}

\begin{defi}\label{def:encodereward}
Let $\mdp$ be an MDP\@. From a given initial state $s_{0}$,
the reward level in each state $s \in S$ can be any of the
countably many values $r_{1}, r_{2}, \dots$
corresponding to the rewards accumulated along all the possible paths
leading to $s$ from $s_{0}$.
We construct the MDP $R(\mdp) \eqdef (S', \zstates', \rstates', \longrightarrow_{R(\mdp)}, P')$
that encodes the reward counter into the state as follows:
\begin{itemize}
\item
The state space of $R(\mdp)$ is
$S' \eqdef \{ (s,r) \mid s \in S, r \in \mathbb{R} \text{ is a reward level attainable at } s \}$.
Note that $S'$ is countable.
We write $s_{0}'$ for the initial state $(s_{0},0)$.
\item
$\zstates' \eqdef \{ (s,r) \in S' \mid s \in \zstates \}$
and $\rstates' \eqdef S' \setminus \zstates'$.
\item
The set of transitions in $R(\mdp)$ is
\begin{align*}
\longrightarrow_{R(\mdp)} \eqdef
\{ &
\left( (s,r),(s',r') \right) \mid (s,r),(s',r') \in S', \\
& s \longrightarrow s' \text{ in } \mdp \text{ and }  r' \eqdef
r+r(s \to s')
\}.
\end{align*}

\item $P': \rstates' \to \mathcal{D}(S')$ is defined such that
\[
P'(s,r)(s',r') \eqdef
    \begin{cases}
    P(s)(s') & \text{ if } (s,r) \longrightarrow_{R(\mdp)} (s',r') \\
    0 & \text{ otherwise }
    \end{cases}
\]

\item The reward for taking transition $ (s,r) \longrightarrow (s',r')$ is $r'$.
\end{itemize}
\end{defi}

\noindent
By labeling transitions in $R(\mdp)$ with the state encoded total reward of the
target state, we ensure that the
point rewards in $R(\mdp)$ correspond exactly to the total rewards in $\mdp$.

\begin{restatable}{lem}{lemmatotaltopoint}\label{totaltopoint}
Let $\mdp$ be an MDP with initial state $s_{0}$. Then given an MD
(resp.\ Markov) strategy $\sigma'$ in $R(\mdp)$ attaining $c \in [0,1]$ for
$\liminfppobj$ from $(s_{0},0)$, there exists a strategy
$\sigma$ attaining $c$ for $\liminftpobj$ in $\mdp$ from $s_{0}$ which uses the same memory as $\sigma'$ plus a reward counter.
\end{restatable}
\begin{proof}
Let $\sigma'$ be an MD (resp.\ Markov) strategy in $R(\mdp)$ attaining $c \in [0,1]$ for
$\liminfppobj$ from $(s_{0},0)$.
We define a strategy $\sigma$ on $\mdp$ from $s_0$ that uses the same memory
as $\sigma'$ plus a reward counter. Then $\sigma$ plays on $\mdp$ exactly like
$\sigma'$ plays on $R(\mdp)$, keeping the reward counter in its memory instead
of in the state.
I.e., at a given state $s$ (and step counter value $m$, in case $\sigma'$ was a
Markov strategy) and reward level $r$, $\sigma$ plays exactly as $\sigma'$
plays in state $(s,r)$ (and step counter value $m$, in case $\sigma'$ was a
Markov strategy).
By our construction of $R(\mdp)$ and the definition of $\sigma$,
the sequences of point rewards seen by $\sigma'$ in runs on $R(\mdp)$
coincide with the sequences of total rewards seen by $\sigma$ in runs in $\mdp$.
Hence we obtain
$\probm_{R(\mdp),(s_0,0),\sigma'}(\liminfppobj)
= \probm_{\mdp,s_0,\sigma}(\liminftpobj)$
as required.
\end{proof}

\begin{defi}\label{def:encodeam}
Given an MDP $\mdp$ with initial state $s_{0}$, we define the new MDP
$A(\mdp)$ that encodes the mean payoffs of partial runs of $\mdp$ into the seen transition
rewards in $A(\mdp)$.
To this end, the states of $A(\mdp)$ encode both the step counter and the
total reward of the run so far. However, the transition rewards in $A(\mdp)$
reflect the \emph{mean} payoff, i.e., the total reward divided by the number
of steps.

We construct $A(\mdp) \eqdef (S', \zstates', \rstates', \longrightarrow_{A(\mdp)}, P')$ as follows:
\begin{itemize}
\item
The state space of $A(\mdp)$ is
\[
S' \eqdef \{(s,n,r) \mid s \in S, n \in \mathbb{N}
\text{ and } r \in \!\mathbb{R}\, \text{ is a reward level attainable at $s$ at step
$n$}\}
\]
Note that $S'$ is countable.
We write $s_{0}'$ for the initial state $(s_{0},0,0)$ of $A(\mdp)$.
\item
$ \zstates' \eqdef
\{(s,n,r) \in S' \mid s \in \zstates \}$ and
$\rstates' \eqdef S' \setminus \zstates'$.
\item
The set of transitions in $A(\mdp)$ is
\begin{align*}
\longrightarrow_{A(\mdp)} \eqdef
\{ &
\left( (s,n,r),(s',n+1,r') \right) \mid \\
& (s,n,r),(s',n+1,r') \in S',\\
& s \longrightarrow s' \text{ in } \mdp \text{ and }
r'= r+r(s \rightarrow s')
\}.
\end{align*}
\item $P': \rstates' \to \mathcal{D}(S')$ is defined such that
\[
P'(s,n,r)(s',n',r') \eqdef
    \begin{cases}
    P(s)(s') & \text{if } (s,n,r) \!\rightarrow_{A(\mdp)}\! (s',n',r') \\
    0 & \!\!\text{otherwise}
    \end{cases}
\]

\item The reward for taking transition $(s,n,r) \longrightarrow (s',n',r')$ is $r' / n'$,
i.e., the transition reward is the \emph{mean} payoff of the partial run so far.
\end{itemize}
\end{defi}


\begin{lem}\label{meantopoint}
Let $\mdp$ be an MDP with initial state $s_{0}$.
Then given an MD strategy $\sigma'$ in $A(\mdp)$ attaining $c \in [0,1]$ for
$\liminfppobj$ from $(s_{0},0,0)$, there exists a strategy $\sigma$ attaining
$c$ for $\liminfmpobj$ in $\mdp$ from $s_{0}$ which uses just a reward counter
and a step counter.
\end{lem}
\begin{proof}
The proof is very similar to that of \Cref{totaltopoint}.
\end{proof}

\bigskip
\paragraph{Sums and products.}


In our proofs we will use the following basic properties of sums and products
(see, e.g.,~\cite{Bromwich:1955}).

\begin{prop}\label{prop:product-sum}
Given an infinite sequence of real numbers $a_n$ with $0 \le a_n < 1$, we
have
\[
\prod_{n=1}^\infty (1-a_n) > 0 \quad\Leftrightarrow\quad \sum_{n=1}^\infty a_n
< \infty.
\]
and the ``$\Rightarrow$'' implication holds even for the weaker assumption $0 \le a_n \le 1$.
\end{prop}
\begin{proof}
If $a_n=1$ for any $n$ then the ``$\Rightarrow$'' implication is vacuously
true, but the ``$\Leftarrow$'' implication does not hold in general.
In the following we assume $0 \le a_n < 1$.

In the case where $a_n$ does not converge to zero, the property is trivial.
In the case where $a_n \rightarrow 0$, it is shown
by taking the logarithm of the product and using the limit comparison test as follows.

Taking the logarithm of the product gives the series
\[
\sum_{n=1}^{\infty} \ln(1 - a_n)
\]
whose convergence (to a finite number $\le 0$) is equivalent to the positivity of the product.
It is also equivalent to the convergence (to a number $\ge 0$) of its negation
$\sum_{n=1}^{\infty} -\ln(1 - a_n)$.
But observe that (by L'H\^{o}pital's rule)
\[
\lim_{x \rightarrow 0} \frac{-\ln(1-x)}{x} = 1.
\]
Since $a_n \rightarrow 0$ we have
\[
\lim_{n \rightarrow \infty} \frac{-\ln(1-a_n)}{a_n} = 1.
\]
By the limit comparison test, the series
$\sum_{n=1}^{\infty} -\ln(1 - a_n)$ converges if and only if the series
$\sum_{n=1}^\infty a_n$ converges.
\end{proof}

\begin{prop}\label{prop:tail-product}
Given an infinite sequence of real numbers $a_n$ with $0 \le a_n \le 1$,
\[
\prod_{n=1}^\infty a_n > 0 \quad\Rightarrow\quad \forall\eps>0\,\exists N.\, \prod_{n=N}^\infty a_n \ge (1-\eps).
\]
\end{prop}
\begin{proof}
  If $a_n=0$ for any $n$ then the property is vacuously true.
  In the following we assume $a_n >0$.
  Since $\prod_{n=1}^\infty a_n > 0$, by taking the logarithm we obtain
  $\sum_{n=1}^\infty \ln(a_n) > -\infty$.
  Thus for every $\delta>0$ there exists an $N$
  s.t.\ $\sum_{n=N}^\infty \ln(a_n) \ge -\delta$.
  By exponentiation we obtain
  $\prod_{n=N}^\infty a_n \ge \exp(-\delta)$.
  By picking $\delta = -\ln(1-\eps)$ the result follows.
\end{proof}

\section{Point Payoff}\label{sec:pointpayoff}
\begin{table*}[hbtp]
\centering
\begin{tabular}{|ll||c|c|}
\hline
\multicolumn{2}{|l||}{Point Payoff}                                        & $\eps$-optimal & Optimal \\ \hline
\multicolumn{1}{|l|}{\multirow{2}{*}{Finitely branching}}   & Upper Bound & Det(Positional)~\ref{finpointpayoff} & Det(Positional)~\ref{finoptupperpp} \\ \cline{2-4}
\multicolumn{1}{|l|}{}                                      & Lower Bound & n/a  & n/a\\ \hline
\multicolumn{1}{|l|}{\multirow{2}{*}{Infinitely branching}} & Upper Bound & Det(SC)~\ref{infpointpayoff} & Det(SC)~\ref{infoptupperpp} \\ \cline{2-4}
\multicolumn{1}{|l|}{}                                      & Lower Bond  & $\neg$Rand(F)~\ref{infbranchsteplower} & $\neg$Rand(F)~\ref{infbranchsteplower} \\ \hline
\end{tabular}
\caption{Strategy complexity of $\eps$-optimal/optimal
strategies for the point payoff objective in
infinitely/finitely branching MDPs.
}%
\label{table:pointpayoff}
\end{table*}

\subsection{Upper Bounds}
In this section we show that
for \emph{finitely branching} MDPs, there exist $\eps$-optimal MD
strategies for $\liminfppobj$.
Whereas for \emph{infinitely branching} MDPs, a step counter suffices in
order to achieve $\liminfppobj$ $\eps$-optimally.

These two very technical results will form the basis of our upper bound analysis of $\liminfmpobj$ and $\liminftpobj$ in later sections.

\begin{lemC}[{\cite[Lemma 23]{KMST2020c}}]%
\label{acyclicsafety}
For every acyclic MDP with a safety objective and every $\eps > 0$,
there exists an MD strategy that is uniformly $\eps$-optimal.
\end{lemC}

\begin{thmC}[{\cite[Theorem 7]{KMST:Transient-arxiv}}]%
\label{epsilontooptimal}
Let $\mdp=\mdptuple$ be a countable MDP, and let $\formula$ be an event that is shift invariant in~$\mdp$.
Suppose for every $s \in S$ there exist $\eps$-optimal MD strategies for~$\formula$.
Then:
\begin{enumerate}
\item There exist uniform $\eps$-optimal MD strategies for~$\formula$.
\item There exists a single MD strategy that is optimal from every state that has an optimal strategy.
\end{enumerate}
\end{thmC}

\subsubsection{Finitely Branching Case}\label{subsec:upper-fb}

In order to prove the main result of this section, we use the following result
on the $\transience$ objective, which is the set of runs that do not visit any state infinitely often.
Given an MDP $\mdp=\mdptuple$,
$\transience \eqdef \bigwedge_{s \in S} \eventually \always \neg s$.

\begin{thmC}[{\cite[Theorem 8]{KMST:Transient-arxiv}}]%
\label{epstransience}
In every countable MDP there exist uniform $\eps$-optimal MD strategies for
$\transience$.
\end{thmC}

\begin{restatable}{lem}{lemfbavoid}\label{lem:fbavoid}
Given a finitely branching countable MDP $\mdp$, a subset $T \subseteq \to$ of
the transitions and a state $\state$, we have
\[
\valueof{\mdp,\neg\eventually T}{\state} < 1
\ \Rightarrow\ \exists k \in \N.\,
\valueof{\mdp,\neg\eventually^{\le k} T}{\state} < 1
\]
i.e., if it is impossible to completely avoid $T$ then
there is a bounded threshold $k$ and a fixed nonzero
chance of seeing $T$ within $\le k$ steps, regardless of the strategy.
\end{restatable}
\begin{proof}
If suffices to show that
$\forall k \in \N.\, \valueof{\mdp,\neg\eventually^{\le k} T}{\state} =1$
implies $\valueof{\mdp,\neg\eventually T}{\state} = 1$.
Since $\mdp$ is finitely branching, the state $\state$ has only finitely many
successors $\{\state_1,\dots,\state_n\}$.

Consider the case where $\state$ is a controlled state.
If we had the property $\forall {1 \le i\le n}\,\exists k_i \in \N.\,
\valueof{\mdp,\neg\eventually^{\le k_i} T}{\state_{i}} < 1$
then we would have
$\valueof{\mdp,\neg\eventually^{\le k} T}{\state} < 1$
for $k=(\max_{1 \le i \le n} k_i)+1$
which contradicts our assumption.
Thus there must exist an $i \in \{1,\dots,n\}$ with
$\forall k \in \N.\, \valueof{\mdp,\neg\eventually^{\le k} T}{\state_i} =1$.
We define a strategy $\zstrat$ that chooses the successor state $s_i$ when in
state $\state$.

Similarly, if $\state$ is a random state, we must have
$\forall k \in \N.\, \valueof{\mdp,\neg\eventually^{\le k} T}{\state_i} =1$
for all its successors $s_i$.

By using our constructed strategy $\zstrat$, we obtain
$\probm_{\mdp,\state,\zstrat}(\neg\eventually T)=1$ and this implies
$\valueof{\mdp,\neg\eventually T}{\state}=1$ as required.
\end{proof}

\begin{thm}\label{finpointpayoff}
Consider a finitely branching MDP $\mdp =\mdptuple$ with initial state $s_{0}$ and a $\liminfppobj$ objective. Then there exist $\eps$-optimal MD strategies.
\end{thm}
\begin{proof}
Let $\eps >0$.
We begin by partitioning the state space into two sets, $S_{\text{safe}}$ and
$S \setminus S_{\text{safe}}$.
The set $S_{\text{safe}}$ is the subset of states which is surely winning for
the safety objective of only using transitions with non-negative rewards
(i.e., never using transitions with negative rewards at all).
Since $\mdp$ is finitely branching, there exists a uniformly optimal MD
strategy $\sigma_{\text{safe}}$ for this safety objective~\cite{Puterman:book,KMSW2017}.

We construct a new MDP $\mdp'$ by modifying $\mdp$. We create a gadget
$G_{\text{safe}}$ composed of a sequence of new controlled states
$x_{0}, x_{1}, x_{2}, \dots$ where all  transitions $x_{i} \to x_{i+1}$
have reward $0$. Hence any run entering $G_{\text{safe}}$ is winning for $\liminfppobj$.
We insert $G_{\text{safe}}$ into $\mdp$ by replacing all incoming transitions
to $S_{\text{safe}}$ with transitions that lead to $x_{0}$. The idea behind
this construction is that when playing in $\mdp$, once you hit a state in
$S_{\text{safe}}$, you can win surely by playing an optimal MD strategy for
safety. So we replace $S_{\text{safe}}$ with the surely winning gadget
$G_{\text{safe}}$.
Thus
\begin{equation}\label{eq:valuesmmprime}
\valueof{\mdp,\liminfppobj}{s_0} = \valueof{\mdp',\liminfppobj}{s_0}
\end{equation}
and if an $\eps$-optimal MD strategy exists in $\mdp$,
then there exists a corresponding one in $\mdp'$, and vice-versa.

We now consider a general (not necessarily MD) $\eps$-optimal strategy $\sigma$
for $\liminfppobj$ from $s_{0}$ on $\mdp'$, i.e.,
\begin{equation}\label{eq:fbsigma-eps-opt}
\probm_{\mdp',s_0,\sigma}(\liminfppobj) \ge \valueof{\mdp',\liminfppobj}{s_0} - \eps.
\end{equation}
Define the safety objective $\text{Safety}_{i}$ which is the objective of
never seeing any point rewards $< -2^{-i}$.
This then allows us to characterize $\liminfppobj$ in terms of safety objectives.
\begin{equation}\label{eq:fbliminfppissafety}
\liminfppobj = \bigcap_{i \in \mathbb{N}} \eventually(\text{Safety}_{i}).
\end{equation}

Now we define the safety objective $\text{Safety}_{i}^{k} \eqdef \eventually^{\leq k}( \text{Safety}_{i} )$ to attain $\text{Safety}_{i}$ within at most $k$ steps. This allows us to write
\begin{equation}\label{eq:fbsafetyisunion}
\eventually(\text{Safety}_{i}) = \bigcup_{k \in \mathbb{N}} \text{Safety}_{i}^{k}.
\end{equation}
By continuity of measures from above we get
\[
0 = \mathcal{P}_{\mdp',s_0,\sigma} \left( \eventually(\text{Safety}_{i}) \cap \bigcap_{k \in \N} \overline{\text{Safety}^{k}_{i}} \right)
  = \lim_{k \to \infty} \mathcal{P}_{\mdp',s_0,\sigma} \left( \eventually(\text{Safety}_{i}) \cap \overline{\text{Safety}^{k}_{i}}
\right).
\]
Hence for every $i \in \mathbb{N}$ and
$\eps_{i} \defeq \eps \cdot 2^{-i}$
there exists $n_{i}$ such that
\begin{equation}\label{eq:fbni}
\mathcal{P}_{\mdp',s_0,\sigma} \left( \eventually(\text{Safety}_{i}) \cap
\overline{\text{Safety}^{n_{i}}_{i}} \right) \leq \eps_{i}.
\end{equation}

Now we can show the following claim.
\begin{restatable}{clm}{claimfblosetwoeps}\label{claim:fblose2eps}
\[
\probm_{\mdp',s_0,\sigma} \left( \bigcap_{i \in \N} \text{Safety}^{n_{i}}_{i}  \right)
\geq
\valueof{\mdp',\liminfppobj}{s_0} - 2\eps.
\]
\end{restatable}
\begin{proof}
\begin{align*}
& \mathcal{P}_{\mdp',\state_0,\sigma} \left( \bigcap_{i \in \N} \text{Safety}^{n_{i}}_{i} \right)
\\
& \geq
\mathcal{P}_{\mdp',\state_0,\sigma} \left(
\bigcap_{k \in \N}  \eventually(\text{Safety}_{k}) \cap \bigcap_{i \in \N} \text{Safety}^{n_{i}}_{i}
\right) \\
& =  \mathcal{P}_{\mdp',\state_0,\sigma} \left( \left( \bigcap_{k \in \N}  \eventually(\text{Safety}_{k}) \cap \bigcap_{i \in \N} \text{Safety}^{n_{i}}_{i} \right)
\cup \left(\overline{\bigcap_{k \in \N}  \eventually(\text{Safety}_{k})} \cap \bigcap_{k \in \N}  \eventually(\text{Safety}_{k})\right)
\right) \\
& = \mathcal{P}_{\mdp',\state_0,\sigma} \left(  \bigcap_{k \in \N}  \eventually(\text{Safety}_{k}) \cap
\left(\bigcap_{i \in \N} \text{Safety}^{n_{i}}_{i} \cup \overline{\bigcap_{k \in \N}  \eventually(\text{Safety}_{k})} \right)
\right) \\
& = 1 - \mathcal{P}_{\mdp',\state_0,\sigma} \left(  \overline{ \bigcap_{k \in \N}  \eventually(\text{Safety}_{k})} \cup
\left(\overline{\bigcap_{i \in \N} \text{Safety}^{n_{i}}_{i}} \cap \bigcap_{k \in \N}  \eventually(\text{Safety}_{k})\right)
\right) \\
& \geq 1 - \mathcal{P}_{\mdp',\state_0,\sigma} \left(  \overline{ \bigcap_{k \in \N}  \eventually(\text{Safety}_{k})} \right)
 -
\mathcal{P}_{\mdp',\state_0,\sigma} \left(\overline{\bigcap_{i \in \N} \text{Safety}^{n_{i}}_{i}} \cap \bigcap_{k \in \N}  \eventually(\text{Safety}_{k})\right) \\
& = \mathcal{P}_{\mdp',\state_0,\sigma} \left(\liminfppobj\right)
 -
\mathcal{P}_{\mdp',\state_0,\sigma} \left(\overline{\bigcap_{i \in \N} \text{Safety}^{n_{i}}_{i}} \cap \bigcap_{k \in \N}  \eventually(\text{Safety}_{k})\right)
 \hspace{2cm} \text{by~\eqref{eq:fbliminfppissafety}}\\
& \geq \valueof{\mdp',\liminfppobj}{s_0} - \eps -
\mathcal{P}_{\mdp',\state_0,\sigma} \left(\bigcup_{i \in \N} \overline{\text{Safety}^{n_{i}}_{i}} \cap \bigcap_{k \in \N}  \eventually(\text{Safety}_{k})\right)
 \hspace{1.57cm} \text{by~\eqref{eq:fbsigma-eps-opt}}\\
& \geq \valueof{\mdp',\liminfppobj}{s_0} - \eps
 - \sum_{i \in \N} \mathcal{P}_{\mdp',\state_0,\sigma} \left( \overline{\text{Safety}^{n_{i}}_{i}} \cap \bigcap_{k \in \N} \eventually(\text{Safety}_{k})
\right) \\
& \geq \valueof{\mdp',\liminfppobj}{s_0} - \eps - \sum_{i \in \N} \eps_{i}
 \hspace{7.45cm} \text{by~\eqref{eq:fbni}}\\
& = \valueof{\mdp',\liminfppobj}{s_0} - 2 \eps  \qedhere
\end{align*}
\end{proof}

Since $\mdp'$ does not have an implicit step counter, we use the following
construction to approximate one.
We define the distance $d(s)$ from $s_{0}$ to a state $s$ as the length of the shortest path from $s_{0}$ to $s$.
Let $\text{Bubble}_{n}(s_{0}) \defeq \{s \in S \mid d(s) \leq n\}$
be those states that can be reached within $n$ steps from $s_{0}$.
Since $\mdp'$ is finitely branching, $\text{Bubble}_{n}(s_{0})$ is finite for
every fixed $n$.
Let
\[
\text{Bad}_{i} \defeq \{  t \in \longrightarrow_{\mdp'} \mid t =
s \longrightarrow_{\mdp'} s',
s \notin \text{Bubble}_{n_{i}}(s_{0}) \text{ and } r(t) < -2^{-i}
\}
\]
be the set of transitions originating outside $\text{Bubble}_{n_{i}}(s_{0})$ whose reward is too negative.
Thus a run from $s_0$ that satisfies $\text{Safety}_{i}^{n_i}$ cannot
use any transition in $\text{Bad}_{i}$, since (by definition of
$\text{Bubble}_{n_{i}}(s_{0})$) they would come after the $n_i$-th step.

Now we create a new state $\perp$ whose only outgoing transition is a self
loop with reward $-1$.
We transform $\mdp'$ into $\mdp''$ by re-directing all transitions in
$\text{Bad}_{i}$ to the new target state $\perp$ for every $i$.
Notice that any run visiting $\perp$ must be losing for $\liminfppobj$ due to
the negative reward on the self loop, but it must also be losing for $\transience$ because of the self loop.

We now show that the change from $\mdp'$ to $\mdp''$
has decreased the value of $s_0$ for $\liminfppobj$ by at most $2\eps$, i.e.,
\begin{equation}\label{eq:fblose2eps}
\valueof{\mdp'',\liminfppobj}{s_0} \ge \valueof{\mdp',\liminfppobj}{s_0} - 2\eps.
\end{equation}
\Cref{eq:fblose2eps} follows from the following steps.
\begin{align*}
\valueof{\mdp'',\liminfppobj}{s_0}
& \ge \probm_{\mdp'',s_0,\sigma} \left( \bigcap_{i \in \N} \text{Safety}^{n_{i}}_{i}  \right) \\
&
= \probm_{\mdp',s_0,\sigma} \left( \bigcap_{i \in \N} \text{Safety}^{n_{i}}_{i}  \right)
& \ \text{by def.\ of $\mdp''$}\\
& \ge \valueof{\mdp',\liminfppobj}{s_0} - 2\eps & \ \text{by \Cref{claim:fblose2eps}}
\end{align*}

In the next step
we argue that under \emph{every} strategy
$\sigma''$ from $s_0$ in $\mdp''$ the attainment for $\liminfppobj$ and
$\transience$ coincide, i.e.,
\begin{restatable}{clm}{eqliminfpptransience}\label{eqliminfpptransience}
\[
\forall \sigma''.\, \probm_{\mdp'',s_0,\sigma''}(\liminfppobj) = \probm_{\mdp'',s_0,\sigma''}(\transience).
\]
\end{restatable}

\begin{proof}
First we show that
\begin{equation}\label{eq:transience-part-liminf}
\transience \subseteq \liminfppobj \quad \text{in $\mdp''$}.
\end{equation}
Let $\rho \in \transience$ be a transient run. Then $\rho$ can never visit the
state $\perp$. Moreover, $\rho$ must eventually leave every finite set
forever. In particular $\rho$ must satisfy $\eventually \always (\neg \text{Bubble}_{n_i}(s_0))$
for every $i$, since $\text{Bubble}_{n_i}(s_0)$ is finite, because $\mdp''$ is
finitely branching. Thus $\rho$ must either fall into $G_{\text{safe}}$,
in which case it satisfies $\liminfppobj$, or for every $i$,
$\rho$ must eventually leave $\text{Bubble}_{n_i}(s_0)$ forever.
By definition of $\text{Bubble}_{n_i}(s_0)$ and $\mdp''$, the run
$\rho$ must eventually stop seeing rewards $< -2^{-i}$ for every $i$.
In this case $\rho$ also satisfies $\liminfppobj$. Thus~\eqref{eq:transience-part-liminf}.

Secondly, we show that
\begin{equation}\label{eq:liminf-part-transience}
\forall \sigma''.\ \probm_{\mdp'',s_0,\sigma''}(\liminfppobj \cap \overline{\transience})=0.
\end{equation}
i.e., except for a null-set, $\liminfppobj$ implies $\transience$ in $\mdp''$.

Let $\sigma''$ be an arbitrary strategy from $s_0$ in $\mdp''$ and
$\playset$ be the set of all runs induced by it.
For every $s \in S$, let $\playset_{s} \defeq \{\rho \in \playset \mid \rho \text{ satisfies } \always \eventually (s) \}$
be the set of runs seeing state $s$ infinitely often.
In particular, any run $\rho \in \playset_{s}$ is not transient.
Indeed, $\overline{\transience} = \bigcup_{s \in S} \playset_{s}$.
We want to show that for every state $s \in S$ and strategy $\sigma''$
\begin{equation}\label{eq:liminf-part-transience-s}
\probm_{\mdp'',s_0,\sigma''}(\liminfppobj \cap \playset_{s}) = 0.
\end{equation}
Since all runs seeing a state in $G_{\text{safe}}$ are transient, every
$\playset_{s}$ with $s \in G_{\text{safe}}$ must be empty. Similarly, every run
seeing $\perp$ is losing for $\liminfppobj$ by construction.
Hence we have~\eqref{eq:liminf-part-transience-s}
for any state $s$ where $s = \perp$ or $s \in G_{\text{safe}}$.

Now consider $\playset_{s}$ where $s$ is neither in $G_{\text{safe}}$ nor
$\perp$.
Let $T_{\it neg} \eqdef \{t \in \longrightarrow\ \mid\ r(t) < 0\}$
be the subset of transitions with negative rewards in $\mdp''$.

We now show that $\valueof{\mdp'',\neg\eventually T_{\it neg}}{s} < 1$
by assuming the opposite and deriving a contradiction.
Assume that $\valueof{\mdp'',\neg\eventually T_{\it neg}}{s} = 1$.
The objective $\neg\eventually T_{\it neg}$ is a safety objective.
Thus, since $\mdp''$ is finitely branching, there exists a strategy
from $s$ that surely avoids $T_{\it neg}$ (always pick an optimal
move)~\cite{Puterman:book,KMSW2017}.
(This does not hold in infinitely branching MDPs where optimal moves might
not exist.)
However, by construction of $\mdp''$, this implies that
$s \in G_{\text{safe}}$. Contradiction.
Thus $\valueof{\mdp'',\neg\eventually T_{\it neg}}{s} < 1$.

Since $\mdp''$ is finitely branching, we can apply \Cref{lem:fbavoid}
and obtain that there exists a threshold $k_s$ such that
$\valueof{\mdp'',\neg\eventually^{\le k_s} T_{\it neg}}{s} < 1$.
Therefore $\delta_s \eqdef 1 - \valueof{\mdp'',\neg\eventually^{\le k_s}
T_{\it neg}}{s} >0$.
Thus, under every strategy, upon visiting $s$ there is a chance $\ge \delta_s$
of seeing a transition in $T_{\it neg}$ within the next $\le k_s$ steps.
Moreover, the subset $T^s_{\it neg} \subseteq T_{\it neg}$
of transitions that can be reached
in $\le k_s$ steps from $s$ is finite, since $\mdp''$ is finitely branching.
The finiteness of $T^s_{\it neg}$ implies that
the maximum of the rewards in $T^s_{\it neg}$ exists and is still negative, i.e.,
$\ell_s \eqdef \max\{r(t)\ \mid\ t \in T^s_{\it neg}\} < 0$.
(This would not be true for an infinite set, since the $\sup$ over an
infinite set of negative numbers could be zero.)
Let $T_{\le \ell} \eqdef \{t \in \longrightarrow\ \mid\ r(t) \le \ell_s\}$
be the subset of transitions with rewards $\le \ell_s$ in $\mdp''$.

Thus, under \emph{every} strategy, upon visiting $s$ there is a chance $\ge \delta_s$
of seeing a transition in $T_{\le \ell}$ within the next $\le k_s$ steps.

For every state $\state \in \states$,
let $\playset_{\state}^{i} \defeq \{ \rho \in \playset \mid \rho \text{ visits $\state$
at least $i$ times} \}$, so we get $\playset_{\state} = \bigcap_{i \in \N} \playset_{\state}^{i}$.
We obtain
\begin{align*}
& \sup_{\sigma''}\probm_{\mdp'',s_0,\sigma''}(\liminfppobj \cap \playset_{s}) \\
& \le \sup_{\sigma''}\probm_{\mdp'',s_0,\sigma''}(\eventually\always\neg
T_{\le \ell} \cap \playset_{s}) & \text{set inclusion}\\
&
= \sup_{\sigma''}\lim_{n \to \infty}\probm_{\mdp'',s_0,\sigma''}(\eventually^{\le
n}\always\neg T_{\le \ell} \cap \playset_{s}) &  \text{continuity of measures}\\
& \le \sup_{\sigma'''}\probm_{\mdp'',s,\sigma'''}(\always\neg
T_{\le \ell} \cap \playset_{s}) & \text{$s$ visited after $>n$ steps}\\
& = \sup_{\sigma'''} \probm_{\mdp'',s,\sigma'''}(\always\neg
T_{\le \ell} \cap \bigcap_{i \in \N} \playset_{s}^{i})
& \text{def.\ of $\playset_{s}^{i}$} \\
& = \sup_{\sigma'''} \lim_{i \to \infty}\probm_{\mdp'',s,\sigma'''}(\always\neg
T_{\le \ell} \cap \playset_{s}^{i})
&  \text{continuity of measures} \\
& \le \lim_{i \to \infty}(1-\delta_s)^i = 0 & \text{by def.\ of $\playset_{s}^{i}$ and $\delta_s$}
\end{align*}
and thus~\eqref{eq:liminf-part-transience-s}.

From this we obtain
$\probm_{\mdp'',s_0,\sigma''}(\liminfppobj \cap \overline{\transience})=
\probm_{\mdp'',s_0,\sigma''}(\liminfppobj \cap \bigcup_{s \in S} \playset_{s}) \le
\sum_{s \in S} \probm_{\mdp'',s_0,\sigma''}(\liminfppobj \cap \playset_{s})=0$
and thus~\eqref{eq:liminf-part-transience}.

From~\eqref{eq:transience-part-liminf}
and~\eqref{eq:liminf-part-transience}
we obtain that for every $\sigma''$ we have
\begin{align*}
& \probm_{\mdp'',s_0,\sigma''}(\liminfppobj) \\
& = \probm_{\mdp'',s_0,\sigma''}(\liminfppobj \cap \transience) + \probm_{\mdp'',s_0,\sigma''}(\liminfppobj \cap \overline{\transience}) \\
& = \probm_{\mdp'',s_0,\sigma''}(\transience) + 0 \\
& = \probm_{\mdp'',s_0,\sigma''}(\transience)
\end{align*}

and thus \Cref{eqliminfpptransience}.
\end{proof}

By \Cref{epstransience}, there exists a uniformly $\eps$-optimal MD
strategy $\widehat{\sigma}$ from $s_0$ for $\transience$ in $\mdp''$, i.e.,
\begin{equation}\label{eq:transience-eps}
\probm_{\mdp'',s_0,\hat{\sigma}}(\transience) \ge \valueof{\mdp'',\transience}{s_0}
- \eps.
\end{equation}
We construct an MD strategy $\sigma^{*}$ in $\mdp$ which plays
like $\sigma_{\text{safe}}$ in $S_{\text{safe}}$ and plays like $\widehat{\sigma}$ everywhere else.
\begin{align*}
\probm_{\mdp,\state_0,\zstrat^{*}}(\liminfppobj)
&=  \probm_{\mdp',\state_0,\hat{\zstrat}}(\liminfppobj) & \text{def.\ of $\zstrat^{*}$ and $\sigma_{\text{safe}}$}\\
& \ge \probm_{\mdp'',\state_0,\hat{\zstrat}}(\liminfppobj) & \text{new losing sink in $\mdp''$}\\
& = \probm_{\mdp'',\state_0,\hat{\zstrat}}(\transience)  & \text{by \Cref{eqliminfpptransience}}\\
& \ge \valueof{\mdp'',\transience}{s_0} - \eps  & \text{by~\eqref{eq:transience-eps}}\\
& = \valueof{\mdp'',\liminfppobj}{s_0} - \eps   & \text{by \Cref{eqliminfpptransience}}\\
& \ge \valueof{\mdp',\liminfppobj}{s_0} - 2\eps -\eps & \text{by~\eqref{eq:fblose2eps}}\\
& = \valueof{\mdp,\liminfppobj}{s_0} - 3\eps & \text{by~\eqref{eq:valuesmmprime}}
\end{align*}
Hence $\sigma^{*}$ is a $3\eps$-optimal MD strategy for $\liminfppobj$ from
$s_0$ in $\mdp$. Since $\eps$ can be chosen arbitrarily small, the result follows.
\end{proof}

\begin{cor}\label{finoptupperpp}
Given a finitely branching MDP $\mdp$ and initial state $s_{0}$, optimal strategies, where they exist, can be chosen MD for $\liminfppobj$.
\end{cor}
\begin{proof}
Since $\liminfppobj$ is shift invariant, the result follows
from \Cref{finpointpayoff} and \Cref{epsilontooptimal}.
\end{proof}

\begin{rem}\label{rem:pptransience}
The proof of \Cref{finpointpayoff} (and thus also \Cref{finoptupperpp})
uses the result about the $\transience$ objective from \Cref{epstransience}.
This is unavoidable, since, at least for finitely branching MDPs,
\Cref{finpointpayoff} also conversely implies \Cref{epstransience} as follows.

Consider a finitely branching MDP $\mdp =\mdptuple$ with initial state $s_{0}$.
Then we can define a reward structure on $\mdp$ such that $\transience$
and $\liminfppobj$ coincide.
For each transition $t$ from state $s$ to state $s'$ let
$n(t) \eqdef \min \{n \mid s \in \text{Bubble}_{n}(s_{0})\}$
and $r(t) \eqdef -1/n(t)$.
Since $\mdp$ is finitely branching, $\text{Bubble}_{n}(s_{0})$
is finite for every $n$. Transient runs eventually leave every finite set forever.
Thus for all runs starting in $s_0$ we have
$\transience \subseteq \liminfppobj$, because $\lim_{n\to\infty} -1/n =0$.
For the reverse inclusion, consider a non-transient run from $s_0$.
This run visits some state $s$ infinitely often.
Since $\mdp$ is finitely branching, by the Pigeonhole principle, it also
visits some transition $t$ from $s$ infinitely often.
So it infinitely often sees some reward $r(t) < 0$ and thus does not satisfy
$\liminfppobj$. I.e., $\overline{\transience} \subseteq
\overline{\liminfppobj}$.
Now, since $\transience = \liminfppobj$,
\Cref{finpointpayoff} implies \Cref{epstransience} (for the finitely branching case).

However, the connection between $\transience$ and $\liminfppobj$
only holds for finitely branching MDPs.
In infinitely branching MDPs, the result about $\transience$ (\Cref{epstransience})
still holds, but the result for $\liminfppobj$ is different, as shown in
\Cref{infpointpayoff} in the
next section.
\end{rem}

\subsubsection{Infinitely Branching Case}\label{subsec:upper-ib}

In this section we consider infinitely branching MDPs. In this setting, $\eps$-optimal strategies
for $\liminfppobj$ require more memory than in the finitely branching case.
In the following theorem we show
how to obtain $\eps$-optimal deterministic Markov strategies
for $\liminfppobj$.
We do this by deriving $\eps$-optimal MD strategies in $S(\mdp)$ via
a reduction to a safety objective.

\begin{thm}\label{infpointpayoff}
Consider an MDP $\mdp$ with initial state $s_{0}$ and a $\liminfppobj$
objective. For every $\eps >0$ there exist
\begin{itemize}
\item $\eps$-optimal MD strategies in $S(\mdp)$.
\item $\eps$-optimal deterministic Markov strategies in $\mdp$.
\end{itemize}
\end{thm}

\begin{proof}
Let $\eps >0$.
We work in $S(\mdp)$ by encoding the step counter into the states of $\mdp$.
Thus $S(\mdp)$ is an acyclic MDP with implicit step counter and corresponding
initial state $s_0' = (s_0,0)$.

We consider a general (not necessarily MD) $\eps$-optimal strategy $\sigma$
for $\liminfppobj$
from $s_0'$ on $S(\mdp)$, i.e.,
\begin{equation}\label{eq:sigma-eps-opt}
\probm_{S(\mdp),s_0',\sigma}(\liminfppobj) \ge \valueof{S(\mdp),\liminfppobj}{s_0'} - \eps.
\end{equation}
Define the safety objective $\text{Safety}_{i}$ which is the objective of
never seeing any point reward $< -2^{-i}$.
This then allows us to characterize $\liminfppobj$ in terms of safety objectives.
\begin{equation}\label{eq:liminfppissafety}
\liminfppobj = \bigcap_{i \in \mathbb{N}} \eventually(\text{Safety}_{i})
\end{equation}

Now we define the safety objective $\text{Safety}_{i}^{k} \eqdef \eventually^{\leq k}( \text{Safety}_{i} )$ to attain $\text{Safety}_{i}$ within at most $k$ steps. This allows us to write
\begin{equation}\label{eq:safetyisunion}
\eventually(\text{Safety}_{i}) = \bigcup_{k \in \mathbb{N}} \text{Safety}_{i}^{k}.
\end{equation}
By continuity of measures from above we get
\begin{align*}
0 & = \mathcal{P}_{S(\mdp),s_0',\sigma} \left( \eventually(\text{Safety}_{i}) \cap \bigcap_{k \in \N} \overline{\text{Safety}^{k}_{i}}\right)\\
  & = \lim_{k \to \infty} \mathcal{P}_{S(\mdp),s_0',\sigma} \left( \eventually(\text{Safety}_{i}) \cap \overline{\text{Safety}^{k}_{i}}\right).
\end{align*}
Hence for every $i \in \mathbb{N}$ and
$\eps_{i} \defeq \eps \cdot 2^{-i}$
there exists $n_{i}$ such that
\begin{equation}\label{eq:ni}
\mathcal{P}_{S(\mdp),s_0',\sigma} \left( \eventually(\text{Safety}_{i}) \cap
\overline{\text{Safety}^{n_{i}}_{i}} \right) \leq \eps_{i}.
\end{equation}

\begin{restatable}{clm}{claimlosetwoeps}\label{claim:lose2eps}
\[
\probm_{S(\mdp),s_0',\sigma} \left( \bigcap_{i \in \N} \text{Safety}^{n_{i}}_{i}  \right)
\geq
\valueof{S(\mdp),\liminfppobj}{s_0'} - 2 \eps.
\]
\end{restatable}
\begin{proof}
The proof is almost identical to the proof of \Cref{claim:fblose2eps}.
Instead of $\mdp'$ with initial state $\state_0$
we have $S(\mdp)$ with initial state $\state_0'$,
and instead of equations~\eqref{eq:fbliminfppissafety},~\eqref{eq:fbsigma-eps-opt} and~\eqref{eq:fbni} we use the corresponding
equations~\eqref{eq:liminfppissafety},~\eqref{eq:sigma-eps-opt}
and~\eqref{eq:ni}, respectively.
\end{proof}

Let $\formula \eqdef \bigcap_{i \in \N} \text{Safety}^{n_{i}}_{i} \subseteq \liminfppobj$.
It follows from \Cref{claim:lose2eps} that
\begin{equation}\label{eq:lose2eps}
 \valueof{S(\mdp),\formula}{s_0'} \ge
 \valueof{S(\mdp),\liminfppobj}{s_0'} - 2\eps.
\end{equation}
The objective $\formula$
is a safety objective on $S(\mdp)$. Therefore, since $S(\mdp)$ is acyclic,
we can apply \Cref{acyclicsafety} to obtain a uniformly $\eps$-optimal MD
strategy $\sigma'$ for $\formula$. Thus
\begin{align*}
& \probm_{S(\mdp),s_0',\sigma'}(\liminfppobj) \\
& \ge \probm_{S(\mdp),s_0',\sigma'}(\formula)   & \ \text{set inclusion} \\
& \ge \valueof{S(\mdp),\formula}{s_0'} - \eps  & \ \text{$\sigma'$ is $\eps$-opt.}\\
& \ge \valueof{S(\mdp),\liminfppobj}{s_0'} - 3\eps. & \ \text{by~\eqref{eq:lose2eps}}
\end{align*}
Thus $\sigma'$ is a $3\eps$-optimal MD strategy for $\liminfppobj$ in $S(\mdp)$.

By \Cref{steptopoint} this then yields a $3\eps$-optimal Markov strategy for
$\liminfppobj$ from $s_{0}$ in $\mdp$, since runs in $\mdp$ and $S(\mdp)$
coincide wrt.\ $\liminfppobj$.
\end{proof}

The following corollary allows us to re-purpose \Cref{infpointpayoff} to obtain an upper bound for optimal strategies.

\begin{cor}\label{infoptupperpp}
Given an MDP $\mdp$ and initial state $s_{0}$, optimal strategies, where they exist,
can be chosen with just a step counter for $\liminfppobj$.
\end{cor}

\begin{proof}
We work in $S(\mdp)$ and we
apply \Cref{infpointpayoff} to obtain $\eps$-optimal MD strategies from every
state of $S(\mdp)$.
Since $\liminfppobj$ is a shift invariant objective, \Cref{epsilontooptimal} yields an MD
strategy that is optimal from every state of $S(\mdp)$ that has an optimal
strategy. By \Cref{steptopoint} we can translate
this MD strategy on $S(\mdp)$ back to a Markov strategy in $\mdp$,
which is optimal for $\liminfppobj$ from $s_{0}$ (provided that $s_0$ admits
any optimal strategy at all).
\end{proof}

\subsection{Lower Bounds}

We have just showed that finitely branching $\eps$-optimal $\liminfppobj$ can be achieved MD, as a result no lower bound exists within the scope of our memory considerations.
For the infinitely branching case, we provide a counterexample in which $\liminfppobj$ and co-B\"{u}chi coincide.

\begin{figure}[htbp]
\begin{center}
\scalebox{1.2}{
\begin{tikzpicture}[>=latex',shorten >=1pt,node distance=1.9cm,on grid,auto,
roundnode/.style={circle, draw,minimum size=1.5mm},
squarenode/.style={rectangle, draw,minimum size=2mm},
diamonddnode/.style={diamond, draw,minimum size=2mm}]

\node [squarenode,initial,initial text={}] (s) at(0,0) [draw]{$s$};

\node[roundnode] (r1)  [below right=1.4cm and 1.5cm of s] {$r_1$};
\node[roundnode] (r3)  [draw=none,right=1.3 of r1] {$\cdots$};
\node[roundnode] (r4)  [right=1.3cm of r3] {$r_i$};
\node[roundnode] (r5)  [draw=none,right=1.3cm of r4] {$\cdots$};

\node [squarenode,double,inner sep = 4pt] (t)  [below=1.6cm of r3] {$ t $};

\draw [->] (s) -- ++(1.5,0) -- (r1);
\node[roundnode] (d)  [draw=none,right=2.8 of s] {$\cdots$};
\draw[-] (s) -- (d);
\draw [->] (d) -- ++(1.3,0) -- (r4);
\node[roundnode] (dd)  [draw=none,right=2.8 of d] {$\cdots$};
\draw[-] (d) -- (dd);

\path[->] (r1) edge node [midway,left] {$\frac{1}{2}$} node[near start, right] {$-1$} (t);
\path[->] (r4) edge node [midway,right=.2cm] {$\frac{1}{2^{i}}$} node[near start, left] {$-1$} (t);
\path[->] (r1) edge  node[pos=0.3,below] {$\frac{1}{2}$} (s);
\path[->] (r4) edge [bend left=1] node [pos=0.2,above] {$1-\frac{1}{2^{i}}$} (s);

\draw [->] (t) -- node[midway, above] {$+1$} ++(-2.8,0) -- (s);
\end{tikzpicture}
}
\caption{
This infinitely branching MDP is adapted from~\cite[Figure 3]{KMSW2017} and
augmented with a reward structure.
(A very similar example has been described in~\cite[Example 2]{Sudderth:2020}.)
All of the edges carry reward $0$ except the edges entering $t$ that carry reward $-1$ and the edge from $t$ to $s$ carries reward $+1$.
As a result, entering $t$ necessarily brings the total reward down to $-1$ before resetting it to $0$.
We use a reduction to the co-B\"{u}chi objective
(i.e., visiting $t$ only finitely often) to show that infinite memory is required for almost-sure as well as $\eps$-optimal strategies for $\liminftpobj$ as well as $\liminfppobj$.
}%
\label{infinitebranchtp}
\end{center}
\end{figure}
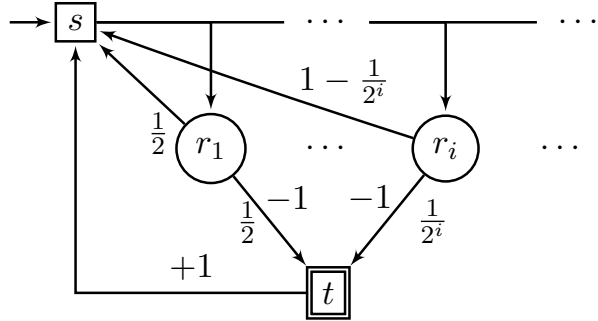

\begin{thm}\label{infbranchsteplower}
There exists an infinitely branching MDP $\mdp$ as in \Cref{infinitebranchtp} with reward implicit in the state and initial state $s$ such that
\begin{itemize}
\item every FR strategy $\sigma$ is such that $\probm_{\mdp, s, \sigma} (\liminfppobj) = 0$
\item there exists an HD strategy $\sigma$ such that $\probm_{\mdp, s, \sigma} (\liminfppobj) = 1$.
\end{itemize}
Hence, optimal (and even almost-surely winning) strategies and $\eps$-optimal
strategies for $\liminfppobj$ require infinite memory
beyond a reward counter.
\end{thm}
\begin{proof}
This follows directly from~\cite[Theorem 4]{KMSW2017} and the observation that in \Cref{infinitebranchtp}, $\liminfppobj$ and co-B\"{u}chi objectives coincide.
\end{proof}

Consequently, when the MDP $\mdp$ is infinitely branching and has the reward
counter implicit in the state, $\liminfppobj$ requires at least a step
counter.

Note that $\liminftpobj$ also coincides with co-B\"{u}chi in the MDP $\mdp$
of \Cref{infinitebranchtp}, hence we restate this theorem in terms of $\liminftpobj$ later in \Cref{infbranchsteplowertp}.

\section{Mean Payoff}\label{sec:meanpayoff}
\begin{table*}[hbtp]
\centering
\begin{tabular}{|ll||c|c|}
\hline
\multicolumn{2}{|l||}{Mean Payoff}                                        & $\eps$-optimal & Optimal \\ \hline
\multicolumn{1}{|l|}{\multirow{2}{*}{Finitely branching}}   & Upper Bound & Det(SC+RC)~\ref{mpepsupper} &  Det(SC+RC)~\ref{infoptuppermp}  \\ \cline{2-4}
\multicolumn{1}{|l|}{}                                      & Lower Bound & \begin{tabular}[c]{@{}l@{}} $\neg$Rand(F+SC)~\ref{infinitesummary}  and\\ $\neg$Rand(F+RC)~\ref{mpstepepslower} \end{tabular} &    \begin{tabular}[c]{@{}l@{}} $\neg$Rand(F+SC)~\ref{almostsummary}  and\\ $\neg$Rand(F+RC)~\ref{mpstepoptlower} \end{tabular} \\ \hline
\multicolumn{1}{|l|}{\multirow{2}{*}{Infinitely branching}} & Upper Bound & Det(SC+RC)~\ref{mpepsupper}  &  Det(SC+RC)~\ref{infoptuppermp} \\ \cline{2-4}
\multicolumn{1}{|l|}{}                                      & Lower Bond  &  \begin{tabular}[c]{@{}l@{}} $\neg$Rand(F+SC)~\ref{infinitesummary} and\\ $\neg$Rand(F+RC)~\ref{mpstepepslower} \end{tabular}  &    \begin{tabular}[c]{@{}l@{}} $\neg$Rand(F+SC)~\ref{almostsummary}  and\\ $\neg$Rand(F+RC)~\ref{mpstepoptlower} \end{tabular} \\ \hline
\end{tabular}
\caption{Strategy complexity of $\eps$-optimal/optimal
strategies for the mean payoff objective in
infinitely/finitely branching MDPs.
}\label{table:meanpayoff}
\end{table*}

\subsection{Upper Bounds}

In order to tackle the upper bounds for the mean payoff objective $\liminfmpobj$, we
work with the acyclic MDP $A(\mdp)$ which encodes both the step counter and the
average reward into the state. Once the average reward is encoded into the state,
the point payoff coincides with the mean payoff. We use this observation to reduce
$\liminfmpobj$ to $\liminfppobj$ and obtain our upper bounds from the corresponding point payoff results.

\begin{cor}\label{mpepsupper}
Given an MDP $\mdp$ and initial state $s_{0}$, there exist $\eps$-optimal strategies $\sigma$ for $\liminfmpobj$ which use just a step counter and a reward counter.
\end{cor}
\begin{proof}
We consider the encoded system $A(\mdp)$ in which both step counter and reward counter are implicit in the state. Recall that the partial mean payoffs in $\mdp$ correspond exactly to point rewards in $A(\mdp)$. Since $A(\mdp)$ has an encoded step counter, \Cref{infpointpayoff} gives us $\eps$-optimal MD strategies for $\liminfppobj$ in $A(\mdp)$. \Cref{meantopoint} allows us to translate these strategies back to $\mdp$ with a memory overhead of just a reward counter and a step counter as required.
\end{proof}

\begin{cor}\label{infoptuppermp}
Given an MDP $\mdp$ and initial state $s_{0}$, optimal strategies, where they exist,
can be chosen with just a reward counter and a step counter for $\liminfmpobj$
\end{cor}

\begin{proof}
We place ourselves in $A(\mdp)$
and apply \Cref{infpointpayoff} to obtain  $\eps$-optimal MD strategies from
every state of $A(\mdp)$. Since $\liminfmpobj$ is a shift invariant objective, \Cref{epsilontooptimal} yields a single MD
strategy that is optimal from every state of $A(\mdp)$ that has an optimal
strategy.
By \Cref{meantopoint} we can translate this MD strategy on $A(\mdp)$ back to
a strategy on $\mdp$ with a step counter and a reward counter.
Provided that $s_{0}$ admits any optimal strategy at all, we obtain
an optimal strategy for $\liminfmpobj$ from $s_{0}$ that uses only a step counter and a reward counter.
\end{proof}

\subsection{Lower Bounds}\label{sec:mplowerbounds}

In this section we will show that $\liminfmpobj$ always requires at least a step counter and a reward counter,
whether it be for $\eps$-optimal or optimal strategies.
We introduce in \Cref{infinitegadget} a building block for an MDP which will form the foundation for all of our lower bound results for $\liminfmpobj$.
Many of these results also hold for $\liminftpobj$, so we will restate them in \Cref{sec:totalpayoff} in due course.

\subsubsection{$\eps$-optimal Strategies}

We construct an acyclic MDP $\mdp$ in which the step counter is implicit in the state as follows.



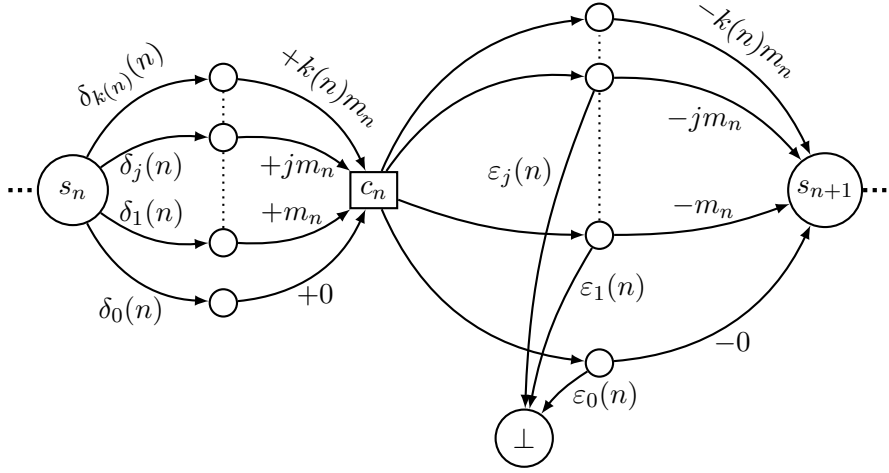
\begin{figure*}
\begin{center}
    \begin{tikzpicture}

    \node[draw,circle, inner sep=5pt] (R) at (0,0) {$s_{n}$};
    \node[draw, minimum height=0.4cm, minimum width=0.4cm] (N) at (4,0) {$c_{n}$};
    \node[draw,circle,inner sep=2pt] (S) at (10,0) {$s_{n+1}$};

    \node[] (Left) at (-1,0) {};
    \node[] (Right) at (11,0) {};
    \draw [dotted, ultra thick] (Left) -- (R);
    \draw [dotted, ultra thick] (S) -- (Right);

    \node[draw,circle] (M) at (7,-0.6){};
    \node[draw,circle] (U) at (7,1.5){};

    \node[draw, circle] at (7,-2.3) (B) {};
    \node[draw, circle] at (7,2.3) (E) {};

    \node[draw, circle] at (6,-3.3) (Deadbottom) {$\perp$};

    %
    %
    %
    %
    %

    \node[draw,circle] (MidBotL) at (2,-0.7){};
    \node[draw,circle] (MidTopL) at (2,0.7){};

    \node[draw, circle] (BotL) at (2,-1.5) {};
    \node[draw, circle] (TopL) at (2,1.5) {};

    \draw [dotted, thick] (MidBotL) -- (MidTopL);
    \draw [dotted, thick] (MidTopL) -- (TopL);

    %
    %
    %
    %
    \coordinate[shift={(0mm,2.5mm)}] (Mshift) at (M);


    \coordinate[shift={(0mm,2.5mm)}] (Ushift) at (U);
    \coordinate[shift={(0mm,-2.4mm)}] (Eshift) at (E);

    \draw [dotted, thick] (Ushift) -- (Eshift);


    \draw[->,>=latex] (R) edge[bend left] node[sloped, above, midway]{$\delta_{k(n)}(n)$} (TopL)
                    (TopL) edge[bend left] node[above, sloped, midway]{$+k(n)m_{n}$} (N)
                    (R) edge[bend left=20] node[below, midway]{$\delta_{j}(n)$} (MidTopL)
                    (MidTopL) edge[bend left=20] node[below, midway]{$+jm_{n}$} (N)
                    (R) edge[bend right=20] node[above, midway]{$\delta_{1}(n) $} (MidBotL)
                    (MidBotL) edge[bend right=20] node[above, midway, xshift=-2pt, yshift=1pt]{$+m_{n}$} (N)
                    (R) edge[bend right] node[below, midway, xshift=-2pt, yshift=-2pt]{$\delta_{0}(n)$} (BotL)
                    (BotL) edge [bend right] node[below, midway]{$+0$} (N);


    \draw[->,>=latex] (N) edge[bend right=10] (M)
                      (N) edge[bend left]  (E)
                      (E) edge[bend left] node[sloped, above, midway]{$-k(n)m_{n}$} (S)
                      (N) edge[bend right] (B)
                      (B) edge[bend right] node[below, midway]{$-0$} (S)
                      (B) edge[bend right=10] node[right, midway, yshift=-2pt]{$\eps_{0}(n)$} (Deadbottom)
                      (M) edge[bend right=10] node[right, near start]{$\varepsilon_{1}(n)$} (Deadbottom)
                      (M) edge[bend right=10] node[above, midway]{$-m_{n}$} (S)
                      (N) edge[bend left] (U)
                      (U) edge[bend left=25] node[below, midway, xshift=-5pt]{$-jm_{n}$} (S)
                      (U) edge[bend right=10] node[left, near start]{$\varepsilon_{j}(n)$} (Deadbottom);

    \draw [dotted, thick] (M) -- (U);

    \end{tikzpicture}
    \caption{A typical building block with $k(n)+1$ choices, first random then controlled. The number of choices $k(n) + 1$ grows unboundedly with $n$. This is the $n$-th building block of the MDP in \Cref{chain}.
    The $\delta_{i}(n)$ and $\eps_{i}(n)$ are probabilities depending on $n$ and the $\pm i m_{n}$ are transition rewards.
    We index the successor states of $s_{n}$ and $c_{n}$ from $0$ to $k(n)$ to
    match the indexing of the $\delta$'s and $\eps$'s such that the bottom
    state is indexed with $0$ and the top state with $k(n)$.
    }%
    \label{infinitegadget}
\end{center}
    \end{figure*}



\begin{figure}[t]
\begin{center}
    \begin{tikzpicture}

    \node[draw, minimum height=0.5cm, minimum width=0.5cm] (I1) at (-0.5,0) {$s_{0}$};
    \node (I2) at (7.5,0) {};
    \node (I3) at (7.5,-1.5) {};

    \node[draw, minimum height=0.7cm, minimum width=0.7cm, fill=black] (B1) at (0.5,0) {B1};
    \node[draw, minimum height=0.7cm, minimum width=0.7cm, fill=black] (B2) at (2,0) {B2};
    \node[draw, minimum height=0.7cm, minimum width=0.7cm, fill=black] (B3) at (3.5,0) {B3};
    \node[draw, minimum height=0.7cm, minimum width=0.7cm, fill=black] (B4) at (5,0) {B4};
    \node[draw, minimum height=0.7cm, minimum width=0.7cm, fill=black] (B5) at (6.5,0) {B5};

    \node[draw, minimum height=0.5cm, minimum width=0.5cm] (W1) at (0.5,-1.5) {};
    \node[draw, minimum height=0.5cm, minimum width=0.5cm] (W2) at (2,-1.5) {};
    \node[draw, minimum height=0.5cm, minimum width=0.5cm] (W3) at (3.5,-1.5) {};
    \node[draw, minimum height=0.5cm, minimum width=0.5cm] (W4) at (5,-1.5) {};
    \node[draw, minimum height=0.5cm, minimum width=0.5cm] (W5) at (6.5,-1.5) {};

    \node[draw, minimum height=0.5cm, minimum width=0.5cm] (W6) at (0.5,-3) {};
    \node[draw, minimum height=0.5cm, minimum width=0.5cm] (W7) at (2,-3) {};
    \node[draw, minimum height=0.5cm, minimum width=0.5cm] (W8) at (3.5,-3) {};
    \node[draw, minimum height=0.5cm, minimum width=0.5cm] (W9) at (5,-3) {};
    \node[draw, minimum height=0.5cm, minimum width=0.5cm] (W10) at (6.5,-3) {};




    \draw[->, >=latex] (I1) -- (B1);
    \draw[dotted, thick] (I2) -- (B5);
    \draw[dotted, thick] (I3) -- (W5);
    \draw[->, >=latex] (I1) to[bend right=20] (W6);

    \draw[->,>=latex] (W1) -- (B2) node[sloped, above, midway]{$+0$};

    \draw[->,>=latex] (W2) -- (B3) node[sloped, above, midway]{$+1$};

    \draw[->,>=latex] (W3) -- (B4) node[sloped, above, midway]{$+2$};

    \draw[->,>=latex] (W4) -- (B5) node[sloped, above, midway]{$+3$};

    \draw[->,>=latex] (B1) -- (B2);
    \draw[->,>=latex] (B2) -- (B3);
    \draw[->,>=latex] (B3) -- (B4);
    \draw[->,>=latex] (B4) -- (B5);

    \draw[->, >=latex, dotted, thick] (W6) -- (W1);
    \draw[->, >=latex, dotted, thick] (W7) -- (W2);
    \draw[->, >=latex, dotted, thick] (W8) -- (W3);
    \draw[->, >=latex, dotted, thick] (W9) -- (W4);
    \draw[->, >=latex, dotted, thick] (W10) -- (W5);



    \draw[->,>=latex] (W1) -- (W7) node[above=0.25cm, midway]{$-1$};
    \draw[->,>=latex] (W2) -- (W8) node[above=0.25cm, midway]{$-1$};
    \draw[->,>=latex] (W3) -- (W9) node[above=0.25cm, midway]{$-1$};
    \draw[->,>=latex] (W4) -- (W10) node[above=0.25cm, midway]{$-1$};

\draw [decorate, decoration={brace,amplitude=5pt},xshift=6pt,yshift=0pt]
(6.4,-1.8) -- (6.4,-2.7) node [black,midway,right, xshift=5pt] {$4$ steps};

    \end{tikzpicture}
    \caption{The buildings blocks from \Cref{infinitegadget} represented by
      black boxes are chained together ($n$ increases as you go to the
      right). The chain of white boxes allows to skip arbitrarily long
      prefixes while preserving path length. The positive rewards from the
      white states to the black boxes reimburse the lost reward accumulated
      until then. The $-1$ rewards between white states ensure that skipping
      gadgets forever is losing.
      \vspace{-5mm}
    }%
    \label{chain}
\end{center}
    \end{figure}
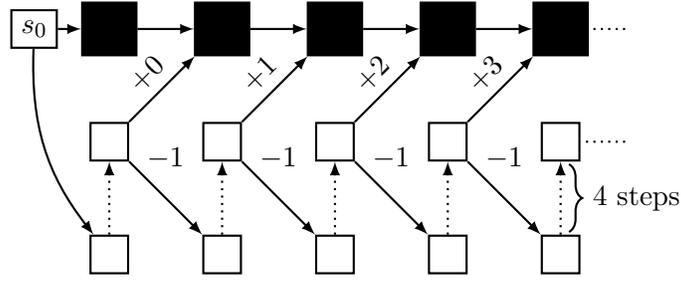


The system consists of a sequence of gadgets.  \Cref{infinitegadget} depicts a typical building block in this system. The system consists of these gadgets chained together as illustrated in  \Cref{chain}, starting with $n$ sufficiently high at $n=N^{*}$. In the controlled choice, there is a small chance in all but the top choice of falling into a $\bot$ state. These $\bot$ states are abbreviations for an infinite chain of states with $-1$ reward on the transitions and are thus losing. The intuition behind the construction is that there is a random transition with branching degree $k(n)+1$. Then, the only way to win, in the controlled states, is to play the $i$-th choice if one arrived from the $i$-th choice. Thus intuitively, to remember what this choice was, one requires at least $k(n)+1$ memory modes. That is to say, the one and only way to win is to mimic, and mimicry requires memory. We will present two similar versions of this MDP, initially we present a version in which the step counter is implicit in the state, then later we will present a version in which the reward counter is  implicit in the state instead.

\begin{rem}
$\mdp$ is acyclic, finitely branching and for every state $s \in S, \exists n_{s} \in \N$ such that every path from $s_{0}$ to $s$ has length $n_{s}$. That is to say the step counter is implicit in the state.
\end{rem}
Additionally, the number of transitions in each gadget grows unboundedly with $n$ according to the function $k(n)$. Consequently, we will show that the number of memory modes required to play correctly grows above every finite bound. This will imply that no finite amount of memory suffices for $\varepsilon$-optimal strategies.

\noindent
\textbf{Notation}:
All logarithms are assumed to be in base $e$.
\begin{align*}
& \log_{1}n \eqdef \log n, \quad \log_{i+1}n \eqdef \log ( \log_{i}n) \\
& \delta_{0}(n)\eqdef\dfrac{1}{\log n}, \quad \delta_{i}(n) \eqdef \dfrac{1}{\log_{i+1}n},
\quad \delta_{k(n)}(n) \eqdef 1 - \sum_{j=0}^{k(n)-1}\delta_{j}(n)\\
& \eps_{0}(n) \eqdef \dfrac{1}{n \log n}, \ \eps_{i+1}(n) \eqdef \dfrac{\eps_{i}(n)}{\log_{i+2}n},
\ \text{i.e. } \eps_{i}(n) = \dfrac{1}{n \cdot \log n \cdot \log_{2}n \cdots \log_{i+1}n},
\eps_{k(n)}(n) \eqdef 0\\
& \text{Tower}(0) \eqdef e^{0} = 1, \quad \text{Tower}(i+1) \eqdef e^{\text{Tower}(i)},
\quad N_{i} \eqdef \text{Tower}(i)
\end{align*}

\begin{lem}\label{convdiv}
The family of series
$
\sum_{n > N_{j}} \delta_{j}(n)  \cdot \eps_{i}(n)
$
is divergent for all $i,j \in \N$, $i < j$.

\noindent
Additionally, the related family of series
$\sum_{n > N_{i}} \delta_{i}(n) \cdot \eps_{i}(n)
$
is convergent for all $i \in \N$.

\end{lem}
\begin{proof}
These are direct consequences of Cauchy's Condensation Test.
\end{proof}

\begin{defi}\label{def:kn}
We define $k(n)$, the rate at which the number of transitions grows. We define $k(n)$ in terms of fast growing functions $g, \text{Tower}$ and $h$ defined for $i \ge 1$ as follows:

\[
g(i) \eqdef \text{min} \left\{ N : \left( \sum_{n>N} \delta_{i-1}(n) \eps_{i-1}(n) \right)  \le 2^{-i}  \right\}, \quad h(1) \eqdef 2
\]

\[
h(i+1) \eqdef \left\lceil \max \left\{ g(i+1), \text{Tower}(i+2),  \min \left\{ m+1 \in \N :\sum^{m}_{n=h(i)} \eps_{i-1}(n) \geq 1 \right\} \right\} \right\rceil.
\]

Note that function $g$ is well defined by \Cref{convdiv}, and
$h(i+1)$ is well defined since for all $i$, $\sum^{\infty}_{n=h(i)} \eps_{i-1}(n)$ diverges to infinity.
$k(n)$ is a slow growing unbounded step function defined in terms of $h$ as
$k(n) \eqdef h^{-1}(n)$.
The Tower function features in the definition to ensure that the transition probabilities are always well defined. $g$ and $h$ are used to smooth the proofs of \Cref{infwin} and \Cref{claim:divergence} respectively.
\emph{Notation:} $N^{*} \eqdef \min \{ n \in \N : k(n) = 1 \}$. This is intuitively the first natural number for which the construction is well defined.

The reward $m_{n}$ which appears in the $n$-th gadget is defined such that it
outweighs any possible reward accumulated up to that
point in previous gadgets. As such we define
$m_{n} \eqdef 2k(n) \sum_{i=N^{*}}^{n-1} m_{i}$,
with $m_{N^{*}} \eqdef 1$ and where $k(n)$ is the
branching degree.
\end{defi}

To simplify the notation,
the state $s_{0}$ in our theorem statements refers to
$s_{N^{*}}$.

\begin{restatable}{lem}{lemwelldefined}\label{lem:welldefined}
For $k(n) \geq 1$, the transition probabilities in the gadgets are well defined.
\end{restatable}

\begin{proof}
Recall that Tower$(i)$ is $i$ repeated exponentials. Thus, $\log($Tower$(i))=$Tower$(i-1)$.

When checking whether probabilities in a given gadget are well defined,
first we choose a gadget. The choice of gadget gives us a branching degree
$k(n)+1$ which in turn lower bounds the value of $n$ in that gadget.
So for a branching degree of $k(n)+1$, we have $n$ lower bounded
by Tower$(k(n)+1)$ by definition of $k(n)$.

We need to show that $\sum_{i=0}^{k(n)-1} \delta_{i}(n) \leq 1$.
Indeed, we have that:

\[
\sum_{i=0}^{k(n)-1} \delta_{i}(n)
\leq
\sum_{i=0}^{k(n)-1} \dfrac{1}{\text{log}_{i+1}(\text{Tower}(k(n)+1))}\\
=
\sum_{i=1}^{k(n)} \dfrac{1}{\text{Tower}(i)}
<
\sum_{i=1}^{k(n)} \dfrac{1}{e^{i}}
<
\sum_{i=1}^{k(n)} \dfrac{1}{2^{i}}
<
1.
\]
Hence, for $k(n) \geq 1$, the transition probabilities are well defined, i.e.\ $\delta_{0}(n), \delta_{1}(n), \ldots ,\delta_{k(n)}(n)$ do indeed sum to 1.
\end{proof}




\begin{restatable}{lem}{lemmainfwin}\label{infwin}
For every $\eps > 0$, there exists a strategy $\zstrat_{\eps}$ with  $\probm_{\mdp, s_{0}, \zstrat_{\eps}}(\liminfmpobj) \geq 1 - \eps$ that cannot fail unless it hits a $\perp$ state. Formally, $\probm_{\mdp, s_{0}, \zstrat_{\eps}}(\liminfmpobj \wedge \always(\neg \perp)) = \probm_{\mdp, s_{0}, \zstrat_{\eps}}(\always( \neg \perp)) \geq 1- \eps$. So in particular, $\valueof{\mdp,\liminfmpobj}{s_{0}} = 1$.
\end{restatable}

\ignore{
\begin{proof}[Outline of the proof.]
We define a strategy $\zstrat$ which in $c_{n}$ always mimics the choice in $s_{n}$. Playing according to $\zstrat$, the only way to lose is by dropping into the $\perp$ state. This is because by mimicking, the player finishes each gadget with a reward of $0$. From $s_{0}$, the probability of surviving while playing in all the gadgets is
\[
\prod_{n \geq N^{*}} \left( 1 - \sum_{j=0}^{k(n)-1} \delta_{j}(n) \cdot \eps_{j}(n) \right) > 0.
\]
Hence the player has a non zero chance of winning when playing $\zstrat$.

When playing with the ability to skip gadgets, as illustrated in \Cref{chain}, all runs not visiting a $\perp$ state are winning since the total reward never dips below $0$.
We then consider the strategy $\zstrat_{\eps}$ which plays like $\zstrat$
after skipping forwards by sufficiently many gadgets
(starting at $n \gg N^{*}$). Its
probability of satisfying $\liminfmpobj$
corresponds to a tail of the
above product, which can be made arbitrarily close to $1$ (and thus $\geq 1-\eps$)
by \Cref{prop:tail-product}.
Thus the strategies $\zstrat_{\eps}$ for arbitrarily small $\eps >0$
witness that $\valueof{\mdp,\liminfmpobj}{s_{0}} = 1$.
\end{proof}
}

\begin{proof}
We define a strategy $\zstrat$ which in $c_{n}$ always mimics the choice in $s_{n}$. We first prove that playing this way gives us a positive chance of winning.
Then we show that there are strategies $\zstrat_\eps$ that attain $1- \eps$ from $s_{0}$ without hitting a $\perp$ state.
This implies in particular that $\valueof{\mdp,\liminfmpobj}{s_{0}} = 1$.

Playing according to $\zstrat$, the only way to lose is by dropping into the $\perp$ state. This is because by mimicking, the player finishes each gadget with a reward of $0$.
In the $n$-th gadget, the chance of reaching the $\perp$ state is
$\sum_{j=0}^{k(n)-1}\delta_{j}(n) \cdot \varepsilon_{j}(n)$.
Thus, the probability of surviving while playing in all the gadgets is
\[
\prod_{n \geq N^{*}} \left( 1 - \sum_{j=0}^{k(n)-1} \delta_{j}(n) \cdot \eps_{j}(n) \right).
\]
However, by \Cref{prop:product-sum}, this product is strictly greater than $0$ if and only if the sum
\[
\sum_{n \geq N^{*}} \left( \sum_{i=0}^{k(n)-1} \delta_{i}(n) \eps_{i}(n) \right)
\]
is finite. With some rearranging exploiting the definition of $k(n)$ we see that this is indeed the case:

\begin{align*}
& \sum_{n \geq N^{*}} \left( \sum_{i=0}^{k(n)-1} \delta_{i}(n) \eps_{i}(n) \right) \\
\leq & \sum_{i \geq 1} \left( \sum_{n=g(i)}^{\infty} \delta_{i-1}(n) \eps_{i-1}(n) \right) &\text{by definition of $k(n)$}\\
\leq & \sum_{i \geq 1} 2^{-i}  &\text{by definition of $g(n)$}\\
\leq & 1
\end{align*}
Hence the player has a non zero chance of winning.

When playing with the ability to skip gadgets, as illustrated in \Cref{chain},
all runs not visiting a $\perp$ state are winning since the total reward never
dips below $0$.
Hence $\probm_{\mdp, s_{0}, \zstrat_{\eps}}(\liminfmpobj \wedge \neg \perp) =
\probm_{\mdp, s_{0}, \zstrat_{\eps}}( \neg \perp )$.
Thus the idea is to skip an arbitrarily long prefix of gadgets to push
the chance of winning $\eps$ close to $1$ by pushing the
chance of visiting a $\perp$ state $\eps$ close to $0$.
From the $N$-th state, for $N \geq N^{*}$, the chance of winning is
\[
\prod_{n \geq N} \left( 1 - \sum_{j=0}^{k(n)-1} \delta_{j}(n) \cdot \eps_{j}(n) \right) > 0
\]
By \Cref{prop:tail-product} this can be made arbitrarily close to $1$ by choosing $N$
sufficiently large.

Let
$N_\eps \eqdef \min
\left\{
    N \in \N \mid
        \prod_{n \geq N} \left(
            1 - \sum_{j=0}^{k(n)-1} \delta_{j}(n) \cdot \eps_{j}(n)
            \right)
        \geq 1-\eps
\right\}$.
Now define the strategy $\zstrat_{\eps}$ to be the strategy that plays like $\zstrat$ after skipping forwards by $N_\eps$ gadgets. Thus, by definition $\zstrat_\eps$ attains $1-\eps$ for all $\eps > 0$.

Thus, by playing $\zstrat_\eps$ for an arbitrarily small $\eps$ the chance of winning must be arbitrarily close to 1. Hence, $\valueof{\mdp,\liminfmpobj}{s_{0}} = 1$.
\end{proof}


\begin{lem}\label{claim:divergence}
$\sum_{n= k^{-1}(2)}^{\infty} \dfrac{1}{2} \delta_{j(n)}(n) \varepsilon_{i(n)}(n)$
diverges for all $i(n),j(n) \in \{0,1, \ldots , k(n)-1\}$ with $i(n) < j(n)$.
\end{lem}

\begin{proof}
This result is not immediate, because the indexing functions $i(n)$ and $j(n)$
may grow with $k(n)$ as $n$ increases.

Under the assumption that $i(n) < j(n)$ we have that
\[\delta_{j(n)}(n) \eps_{i(n)}(n) \geq \delta_{j(n)}(n) \eps_{j(n)-1}(n) \geq \delta_{k(n)-1}(n) \eps_{k(n)-2}(n) = \eps_{k(n)-1}(n).\]
Thus it suffices to show that $\sum_{n=k^{-1}(2)}^{\infty} \eps_{k(n)-1}(n)$ diverges:

\begin{align*}
\sum_{n=k^{-1}(2)}^{\infty} \eps_{k(n)-1}(n) & = \sum^{\infty}_{a=2} \quad \sum_{n= k^{-1}(a)}^{k^{-1}(a+1)-1} \eps_{a-1}(n) & \text{splitting the sum up}\\
& = \sum^{\infty}_{a=2} \quad \sum_{n= h(a)}^{h(a+1)-1} \eps_{a-1}(n) & \text{$k(n) = h^{-1}(n)$}\\
& \geq \sum^{\infty}_{a=2} 1 & \text{definition of $h(n)$}
\end{align*}

Note that the definition of $h(i)$ says exactly that a block of the form $\sum_{n= h(a)}^{h(a+1)-1} \eps_{a-1}(n)$ is at least $1$.
Hence $\sum_{n= k^{-1}(2)}^{\infty} \dfrac{1}{2} \delta_{j(n)}(n) \varepsilon_{i(n)}(n)$ diverges as required.
\end{proof}

\begin{lem}\label{alglose}
For any sequence $\{\alpha_{n}\}$, where $\alpha_{n} \in [0,1]$ for all $n$, and any functions $i(n), j(n) : \mathbb{N} \to \N $ with $i(n), j(n) \in \{0, 1, \ldots , k(n)-1\}, i(n) < j(n)$ for all $n$, the following sum diverges:

\begin{equation}\label{ijlosingsum}
\sum_{n=k^{-1}(2)}^{\infty} \Big( \delta_{j(n)}(n)(\alpha_{n} \varepsilon_{j(n)}(n) + (1- \alpha_{n}) \varepsilon_{i(n)}(n)) + \delta_{i(n)}(n)(\alpha_{n} + (1- \alpha_{n}) \varepsilon_{i(n)}(n)) \Big).
\end{equation}
\end{lem}

\begin{proof}
We can narrow our focus by noticing that
\begin{align*}
& \sum_{n=k^{-1}(2)}^{\infty} \Big( \delta_{j(n)}(n)(\alpha_{n} \varepsilon_{j(n)}(n) + (1- \alpha_{n}) \varepsilon_{i(n)}(n)) + \delta_{i(n)}(n)(\alpha_{n} + (1- \alpha_{n}) \varepsilon_{i(n)}(n)) \Big) \\
& = \sum_{n=k^{-1}(2)}^{\infty} \alpha_{n} \delta_{j(n)}(n) \varepsilon_{j(n)}(n) + (1- \alpha_{n}) \delta_{i(n)}\varepsilon_{i(n)}(n) \qquad \text{Convergent by def.\ of $\delta_{i}(n), \varepsilon_{i}(n)$} \\
& + \sum_{n=k^{-1}(2)}^{\infty}(1- \alpha_{n}) \delta_{j(n)} \varepsilon_{i(n)}(n) + \alpha_{n} \delta_{i(n)}(n)
\end{align*}

Hence the divergence of~\eqref{ijlosingsum} depends only on the divergence of
\[\sum_{n=k^{-1}(2)}^{\infty}(1- \alpha_{n}) \delta_{j(n)} \varepsilon_{i(n)}(n) + \alpha_{n} \delta_{i(n)}(n).\]
No matter how the sequence $\{ \alpha_{n} \}$ behaves, for every $n$ we have that either $\alpha_{n} \geq 1/2$ or $1- \alpha_{n} \geq 1/2$. Hence for every $n$ it is the case that
\begin{align*}
(1-\alpha_{n}) \delta_{j(n)}(n) \varepsilon_{i(n)}(n) +
  \alpha_{n}\delta_{i(n)}(n) \,\geq\, & \dfrac{1}{2} \delta_{j(n)}(n) \varepsilon_{i(n)}(n) \\
\text{or } & \\
\,\geq\, & \dfrac{1}{2} \delta_{i(n)}(n)
\end{align*}

Define the function $f$ as follows:
 \[
    f(n)=\left\{
                \begin{array}{ll}
                \dfrac{1}{2} \delta_{i(n)}(n) \text{ if $\alpha_{n} \geq 1/2$} \\
                  \\
                  \dfrac{1}{2} \delta_{j(n)}(n) \varepsilon_{i(n)}(n) \text{ otherwise}
                \end{array}
              \right.
  \]

Hence no matter how $\{ \alpha_{n} \}$ behaves, we have that
\begin{align*}
\sum_{n= k^{-1}(2)}^{\infty} \Big( \delta_{j(n)}(n)(\alpha_{n} \eps_{j(n)}(n) + (1- \alpha_{n}) \eps_{i(n)}(n)) + \delta_{i(n)}(n)(\alpha_{n} & + (1- \alpha_{n}) \eps_{i(n)}(n)) \Big) \\
& \geq \sum_{n= k^{-1}(2)}^{\infty} f(n).
\end{align*}
We know that both
$\sum_{n= k^{-1}(2)}^{\infty} \dfrac{1}{2} \delta_{j(n)}(n) \varepsilon_{i(n)}(n)$
and
$\sum_{n= k^{-1}(2)}^{\infty} \dfrac{1}{2} \delta_{i(n)}(n)$
diverge for all $i(n),j(n) \in \{0,1, \ldots, k(n)-1\}$,
$i(n) < j(n)$,
as shown in \Cref{claim:divergence}.

Thus
$\sum_{n= k^{-1}(2)}^{\infty} \dfrac{1}{2} \delta_{j(n)}(n) \varepsilon_{i(n)}(n)$
and
$\sum_{n=k^{-1}(2)}^{\infty} \dfrac{1}{2} \delta_{i(n)}(n)$ must also diverge no matter how $i(n)$ and $j(n)$ behave. As a result  it must be the case that
$\sum_{n=k^{-1}(2)}^{\infty} f(n)$ diverges.
Hence~\eqref{ijlosingsum} must be divergent as desired as $i(n)$ and $j(n)$ vary for $n \geq k^{-1}(2)$.
\end{proof}

\begin{restatable}{lem}{lemmainflose}\label{inflose}
For any FR strategy $\zstrat$, almost surely either the mean payoff dips below $-1$ infinitely often, or the run hits a $\perp$ state, i.e.\ $\probm_{\mdp, \zstrat, s_{0}}(\liminfmpobj)=0$.
\end{restatable}

\begin{proof}[Outline of the proof.]
Let $\zstrat$ be some FR strategy with $k$ memory modes.
We prove a \emph{lower bound} $e_n$ on the probability of a local error
(reaching a $\perp$ state, or seeing a mean payoff $\leq -1$)
in the current $n$-th gadget. This lower bound $e_n$ holds regardless
of events in past gadgets, regardless of the memory mode of $\zstrat$
upon entering the $n$-th gadget, and cannot be improved by
$\zstrat$ randomizing its memory updates.

The main idea is that,
once $k(n) > k+1$
(which holds for $n \geq N'$ sufficiently large)
by the Pigeonhole Principle there will always be
a memory mode confusing at least two different branches $i(n),j(n) \neq k(n)$
of the previous random choice at state $s_n$.
This confusion yields a probability $\geq e_n$
of reaching a $\perp$ state or seeing a mean payoff $\leq -1$,
regardless of events in past gadgets and regardless
of the memory upon entering the $n$-th gadget.
We show that $\sum_{n \geq N'} e_n$ is a \emph{divergent} series.
Thus, by \Cref{prop:product-sum}, $\prod_{n \geq N'} (1 - e_n)=0$.
Hence, $\probm_{\mdp, \zstrat, s_{0}}(\liminfmpobj) \leq \prod_{n \geq
N'} (1 - e_n) = 0$.
\end{proof}

\begin{proof}[Full proof]
Let $\sigma$ be some FR strategy with $k$ memory modes.
Our MDP consists of a linear sequence of gadgets (\Cref{infinitegadget}) and
is in particular acyclic.
The $n$-th gadget is entered at state $s_n$ and takes 4 steps.
Locally in the $n$-th gadget there are 3 possible scenarios:
\begin{enumerate}
\item
The random transition picks some branch $i$ at $s_n$ and the strategy then
picks a branch $j > i$ at $c_n$.

By the definition of the payoffs (multiples of $m_n$; cf.~\Cref{def:kn}),
this means that we see a mean payoff $\le -1$, regardless of events in past
gadgets.
This is because the numbers $m_n$ grow so quickly
with $n$ that even the combined maximal possible rewards of all past gadgets
are so small in comparison that they do not matter for the outcome
in the $n$-th gadget, i.e.,
rewards from past gadgets cannot help to avoid seeing a mean payoff $\le -1$
in the above scenario.
\item
We reach the losing sink $\bot$ (and thus will keep seeing a mean payoff $\le -1$
forever). This happens with probability $\eps_j(n)$ if the strategy picks
some branch $j$ at $c_n$, regardless of past events.
\item
All other cases.
\end{enumerate}
As explained above, due to the definition of the rewards (\Cref{def:kn}),
events in past gadgets do not make the difference between (1),(2),(3) in the
current gadget.
It just depends on the choices of the strategy $\sigma$ in the current gadget.

Let ${\it Bad}_n$ be the event of seeing either of the two unfavorable outcomes (1) or (2) in the $n$-th
gadget. Let $p_n$ be the probability of ${\it Bad}_n$ under strategy $\sigma$.
Since $\sigma$ has memory, the probabilities $p_n$ are not necessarily
independent.
However, we show \emph{lower bounds} $e_n \le p_n$ that hold universally for every
FR strategy $\sigma$ with $\le k$ memory modes and every $n$ such that
$k(n) > k+1$.
The lower bound $e_n$ will hold regardless of the memory mode of $\sigma$ upon
entering the $n$-th gadget.

{\bf\noindent Memory updates.}
First we show that $\sigma$ randomizing its memory update after observing
the random transition from state $s_n$ does \emph{not} help to reduce the
probability of event ${\it Bad}_n$.
I.e., we show that without restriction $\sigma$ can update its memory
deterministically after observing the transition from state $s_n$.

Once in the controlled state $c_n$, the strategy $\sigma$ can base its choice only on the
current state (always $c_n$ in the $n$-th gadget)
and on the current memory mode.
Thus, in state $c_n$, in each memory mode $\memconf$, the strategy has
to pick a distribution $\mathcal{D}^{c_n}_{\memconf}$
over the available transitions from $c_n$.
By the finiteness
of the number of memory modes of $\sigma$
(just $\le k$ by our assumption above),
for each possible
reward level $x$ (obtained in the step from the preceding random transition
from $s_n$)
there is a best memory mode $\memconf(x)$ such that $\mathcal{D}^{c_n}_{\memconf(x)}$ is optimal
(in the sense of minimizing the probability of event ${\it Bad}_n$)
for that particular reward level $x$.
(In case of a tie, just use an arbitrary tie break, e.g.,
some pre-defined linear order on the memory modes.)

Therefore, upon witnessing a reward level $x$ in the random transition from
state $s_n$, the strategy $\sigma$ can minimize the probability of event
${\it Bad}_n$ by \emph{deterministically} setting its memory to $\memconf(x)$.
Thus randomizing its memory update does not help to reduce the probability
of ${\it Bad}_n$, and we may assume without restriction that $\sigma$ updates
its memory deterministically.

(Note that the above argument only works because it is local to the current
gadget where we have a finite number of decisions (here just one),
we have a finite number of memory modes, and a one-dimensional criterion for
local optimality (minimizing the probability of event ${\it Bad}_n$).
We do \emph{not} claim that randomized memory updates are useless for every strategy
in every MDP and every objective.)

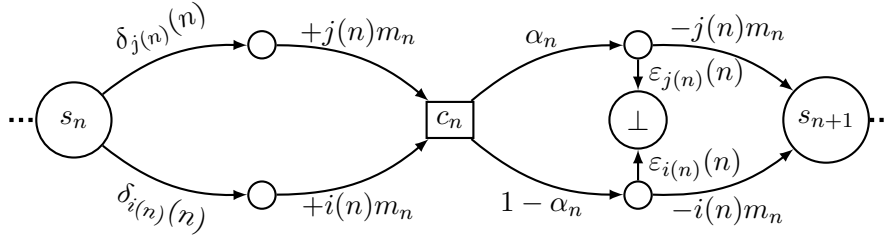
\begin{figure}
\begin{center}
    \begin{tikzpicture}

    \node[draw,circle, minimum height=1cm] (R) at (0,0) {$s_{n}$};
    \node[draw, minimum height=0.4cm, minimum width=0.4cm] (N) at (5,0) {$c_{n}$};
    \node[draw,circle] (S) at (10,0) {$s_{n+1}$};
    \node (M) at (6,0){};

    \node[] (Left) at (-1,0) {};
    \node[] (Right) at (11,0) {};
    \draw [dotted, ultra thick] (Left) -- (R);
    \draw [dotted, ultra thick] (S) -- (Right);

    \node[draw, circle] (Top) at (2.5,1) {};
    \node[draw, circle] (Bot) at (2.5,-1) {};

    \node[draw, circle] at (7.5,-1) (B) {};
    \node[draw, circle] (E) at (7.5,1) {};
    \node[draw, circle] at (7.5,0) (Deadup) {$\perp$};

    \draw[->,>=latex] (R) edge[bend left=20] node[above, midway, sloped]{$\delta_{j(n)}(n)$} (Top)
                      (Top) edge[bend left=20] node[above, midway]{$+j(n)m_{n}$} (N)
                      (R) edge[bend right=20] node[below, midway, sloped]{$\delta_{i(n)}(n)$} (Bot)
                      (Bot) edge[bend right=20] node[below, midway]{$+i(n)m_{n}$} (N);

    \draw[->,>=latex] (N) edge[bend left=20] node[above, midway]{$ \alpha_{n}$} (E)
                      (E) edge[bend left=20] node[above, midway]{$-j(n)m_{n}$} (S)
                      (N) edge[bend right=20] node[below, midway]{$1- \alpha_{n}$} (B)
                      (B) edge[bend right=20] node[below, midway]{$-i(n)m_{n}$} (S)
                      (E) edge node[right, midway]{$\varepsilon_{j(n)}(n)$} (Deadup)
                      (B) edge node[right, midway]{$\varepsilon_{i(n)}(n)$} (Deadup);

    \end{tikzpicture}
    \caption{When transitions $i(n)$ and $j(n)$ are confused in the player's memory, the player's choice is at least as bad as the reduced play in this simplified gadget.}%
    \label{ijcase}
    \end{center}
    \end{figure}

\begin{clm}\label{claim:confusion-simple}
Assume that the transitions $i(n)$ and $j(n)$
(with $i(n) < j(n)$) leading to state $c_n$ are confused in the memory of the
strategy. Then we can assume without restriction that the strategy
only plays transitions $i(n)$ and $j(n)$ with nonzero probability
from state $c_n$, since every other behavior yields a higher probability of
the event ${\it Bad}_n$ (cf.~\Cref{ijcase}).
\end{clm}
\begin{proof}
When confusing transitions $i(n)$ and $j(n)$ with $i(n) < j(n)$,
the player's choice of transition from $c_n$ can be broken down into 5 distinct cases.
The player can choose transition $x(n)$ as follows.
\begin{enumerate}
\item $x(n) = i(n)$
\item $x(n) = j(n)$
\item $x(n) > j(n)$
\item $x(n) < i(n)$
\item $i(n) < x(n) < j(n)$
\end{enumerate}

\noindent Case 1 leads to a probability of ${\it Bad}_n$ of $\delta_{j(n)}(n) \eps_{i(n)}(n) + \delta_{i(n)}(n) \eps_{i(n)}(n)$.

\noindent Case 2 leads to a probability of ${\it Bad}_n$ of $\delta_{j(n)}(n) \eps_{j(n)}(n) + \delta_{i(n)}(n)$.

\noindent Case 3 leads to a mean payoff $\le -1$ (and thus ${\it Bad}_n$) with probability $1$.
This is the worst possible case.

\noindent Case 4 leads to a probability of ${\it Bad}_n$ of
$\delta_{j(n)}(n) \eps_{x(n)}(n) + \delta_{i(n)}(n) \eps_{x(n)}(n) >
\delta_{j(n)}(n) \eps_{i(n)}(n) + \delta_{i(n)}(n) \eps_{i(n)}(n)$, i.e.,
this is worse than Case 1.

\noindent Case 5 leads to a probability of ${\it Bad}_n$ of
$\delta_{j(n)}(n) \eps_{x(n)}(n) + \delta_{i(n)}(n) > \delta_{j(n)}(n) \eps_{j(n)}(n) + \delta_{i(n)}(n)$,
i.e., this is worse than Case 2.

Hence, without restriction we can assume that
only cases 1 and 2 will get played with positive probability,
that is to say that in state $c_n$ the strategy will only randomize over
the outgoing transitions $i(n)$ and $j(n)$.
\end{proof}
{\bf\noindent The lower bounds $e_n$.}
Now we consider an FR strategy $\sigma$ that without restriction updates its memory
\emph{deterministically} after each random choice (from state $s_n$) in the $n$-th
gadget. It can still randomize its actions, however.

Let $N'$ be the minimal number such that for all $n \ge N'$ we
have $k(n) > k+1$. In particular, this implies $N' \ge k^{-1}(2)$,
and thus we can apply \Cref{alglose} later.

Once $n \ge N'$, then by the Pigeonhole Principle there will always be a
memory mode confusing at least two different transitions $i(n),j(n) \neq k(n)$
from state $s_n$ to $c_n$.
Note that this holds regardless of the memory mode of $\sigma$ upon entering
the $n$-th gadget.
(The strategy might confuse many other scenarios, but just one confused pair
$i(n),j(n) \neq k(n)$ is enough for our lower bound.)
Without loss of generality, let $j(n)$ be larger of the two confused transitions, i.e., $i(n) < j(n)$.
Let $i(n)$ and $j(n)$ be two functions taking values in $\{0, 1, \ldots, k(n)-1\}$ where $i(n) < j(n)$ for all $n$.

Confusing two transitions $i(n)$ and $j(n)$ from $s_n$ to $c_n$
(where without restriction $i(n) < j(n)$), the strategy is in the same
memory mode afterwards. However, it can still randomize its choices in state
$c_n$.
To prove our lower bound on the probability of ${\it Bad}_n$,
it suffices to consider the case where the strategy
only randomizes over the outgoing transitions $i(n)$ and $j(n)$ from state $c_n$.
This is because, by \Cref{claim:confusion-simple}, every other
behavior would perform even worse, in the sense of yielding a higher
probability of ${\it Bad}_n$.

That is to say that the strategy picks the higher $j(n)$-th branch with some probability $\alpha_n$
and the lower $i(n)$-th branch with probability $1-\alpha_n$.
(We leave the probabilities $\alpha_n$ unspecified here. Using \Cref{alglose},
we'll show that our result holds regardless of their values.)

The local chance of
the event ${\it Bad}_n$  is then lower bounded by
\[e_n \eqdef \delta_{j(n)}(n)(\alpha_{n} \varepsilon_{j(n)}(n) + (1- \alpha_{n}) \varepsilon_{i(n)}(n)) + \delta_{i(n)}(n)(\alpha_{n} + (1- \alpha_{n}) \varepsilon_{i(n)}(n)).
\]
The term above just expresses a case distinction.
In the first scenario, the random transition chooses the $j(n)$-th branch
(with probability $\delta_{j(n)}(n)$) and then the strategy
chooses the $j(n)$-th branch with probability $\alpha_n$
and the lower $i(n)$-th branch with probability $1-\alpha_n$, and you obtain the
respective chances of reaching the sink $\bot$.
In the second scenario, the random transition chooses the $i(n)$-th branch
(with probability $\delta_{i(n)}(n)$). If the strategy then
chooses the higher $j(n)$-th branch (with probability $\alpha_n$) then we have
outcome (1), yielding a mean payoff $\le -1$.
If the strategy chooses the $i(n)$-th branch (with probability $1-\alpha_n$)
then we still have a chance of $\varepsilon_{i(n)}(n)$
of reaching the sink.

Since, as shown above, randomized memory updates do
not help to reduce the probability of ${\it Bad}_n$, the lower bound
$e_n$ for deterministic updates carries over to the general case.
Thus, even for general randomized FR strategies $\sigma$ with $k$ memory
modes, the probability of event ${\it Bad}_n$ in the $n$-th gadget
(for $n \ge N'$) is lower bounded by $e_n$,
regardless of the memory mode $\memconf$ upon entering the
gadget and regardless of events in past gadgets.
We write $\sigma[\memconf]$ for the strategy $\sigma$ in memory mode
$\memconf$ and obtain
\begin{equation}\label{eq:local-bad}
\forall n \ge N'.\ \forall \memconf.\ \probm_{\mathcal{M}, \sigma[m], s_{n}}({\it Bad}_n) \ge e_n
\end{equation}

{\bf\noindent The final step.}
Let ${\it Bad} \eqdef \cup_n {\it Bad}_n$.

Since $i(n), j(n) \neq k(n)$ and $N' \ge k^{-1}(2)$,
we apply \Cref{alglose} to conclude that the
series $\sum_{n=N'}^{\infty} e_n =
\sum_{n=N'}^{\infty} \delta_{j(n)}(n)(\alpha_{n} \varepsilon_{j(n)}(n) +
(1- \alpha_{n}) \varepsilon_{i(n)}(n)) + \delta_{i(n)}(n)(\alpha_{n} +
(1- \alpha_{n}) \varepsilon_{i(n)}(n))$
is divergent, regardless of the behavior of $i(n), j(n)$ or the sequence
$\{\alpha_{n}\}$.

Finally, we obtain
\begin{align*}
& \probm_{\mathcal{M}, \sigma, s_{0}}(\liminfmpobj) \\
& \le \probm_{\mathcal{M}, \sigma, s_{0}}(\eventually\always \neg{\it Bad})
& \text{set inclusion}\\
& = \probm_{\mathcal{M}, \sigma, s_{0}}\left(\bigcup_l\eventually^{\le
l}\always \neg{\it Bad}\right) & \text{def.\ of $\eventually$}\\
& = \lim_{l \to \infty} \probm_{\mathcal{M}, \sigma, s_{0}}(\eventually^{\le l}\always \neg{\it Bad}) & \text{continuity of measures}\\
& \le \lim_{l \to \infty} \probm_{\mathcal{M}, \sigma, s_{0}}\left(\bigcap_{n \ge l/4} \neg{\it Bad}_n\right) & \text{4 steps per gadget}\\
& \le \lim_{4N' \le l \to \infty} \prod_{n \ge l/4 \ge
N'} (\max_\memconf\,\probm_{\mdp, \sigma[\memconf], s_{n}}(\neg{\it Bad}_n))
&\begin{tabular}{l}
linear sequence of gadgets,        \\
finite memory and past events      \\
do not help to avoid ${\it Bad}_n$
\end{tabular} \\
& \le \lim_{4N' \le l \to \infty} \prod_{n \ge l/4 \ge N'} (1-e_n) & \text{by~\eqref{eq:local-bad}}\\
& = \lim_{4N' \le l \to \infty} 0 & \text{divergence of $\sum_{n = N'}^\infty e_n$ and \Cref{prop:product-sum}}\\
& = 0    &  \qedhere
\end{align*}
\end{proof}

\Cref{infwin} and \Cref{inflose} yield the following theorem.

\begin{thm}\label{infinitesummary}
There exists a countable, finitely branching and acyclic MDP $\mdp$ whose step counter is implicit in the state for which
$\valueof{\mdp,\liminfmpobj}{s_{0}} = 1$ and any FR strategy $\zstrat$ is such that
$\probm_{\mdp, s_{0}, \zstrat}(\liminfmpobj)=0$.
In particular, there are no $\eps$-optimal $k$-bit Markov strategies
for any $k \in \N$ and any $\eps < 1$ for
$\liminfmpobj$ in
countable MDPs.
\end{thm}
\begin{proof}
Proved by \Cref{infwin} and \Cref{inflose}.
\end{proof}

\bigskip

\Cref{infinitesummary} shows that even if the step counter is implicit in the
state, infinite memory is still required.
We now adapt our construction from \Cref{infinitegadget} such that
instead the current total reward is implicit in the state, in order to show that
a reward counter plus arbitrary finite memory
does not suffice for ($\eps$-)optimal strategies for $\liminfmpobj$ either.


\begin{figure*}
\begin{center}
\begin{tikzpicture}

\node[draw, circle, minimum width=1cm] (S1) at (0,0) {$s_{n}$};
\node[draw, circle, minimum width=1cm] (S2) at (12,0) {$s_{n+1}$};
\node[draw, minimum width=0.5cm, minimum height=0.5cm] (C) at (6,0) {$c_{n}$};

\node[draw, circle] (P1) at (2,2) {};
\node[draw, circle] (P2) at (2,0.67) {};
\node[draw, circle] (P3) at (2,-0.67) {};
\node[draw, circle] (P4) at (2,-2) {};

\node[draw, circle] (P5) at (4,2) {};
\node[draw, circle] (P6) at (4,0.67) {};
\node[draw, circle] (P7) at (4,-0.67) {};
\node[draw, circle] (P8) at (4,-2) {};

\node[draw, circle] (Q1) at (9,2) {};
\node[draw, circle] (Q2) at (9,0.67) {};
\node[draw, circle] (Q3) at (9,-0.67) {};
\node[draw, circle] (Q4) at (9,-2) {};

\node[draw, circle, minimum width=0.5cm] (Dead) at (7,-3) {$\perp$};

\node[] (I1) at (2,1.8) {};
\node[] (I2) at (2,0.87) {};
\node[] (I3) at (2,0.47) {};
\node[] (I4) at (2,-0.47) {};

\node[] (I5) at (4,1.8) {};
\node[] (I6) at (4,0.87) {};
\node[] (I7) at (4,0.47) {};
\node[] (I8) at (4,-0.47) {};

\node[] (I9) at (9,1.8) {};
\node[] (I10) at (9,0.87) {};
\node[] (I11) at (9,0.47) {};
\node[] (I12) at (9,-0.47) {};

\draw[dotted, thick] (P1) -- (P5);
\draw[dotted, thick] (P2) -- (P6);
\draw[dotted, thick] (P3) -- (P7);
\draw[dotted, thick] (P4) -- (P8);

\draw[dotted, thick] (I1) -- (I2);
\draw[dotted, thick] (I3) -- (I4);
\draw[dotted, thick] (I5) -- (I6);
\draw[dotted, thick] (I7) -- (I8);

\draw[->,>=latex]
(S1) edge node[above, midway, sloped]{\scriptsize $\delta_{k(n)}(n)$} (P1)
(S1) edge node[above, midway, sloped]{\scriptsize $\delta_{i}(n)$} (P2)
(S1) edge node[above, midway, sloped]{\scriptsize $\delta_{1}(n)$} (P3)
(S1) edge node[below, midway, sloped]{\scriptsize $\delta_{0}(n)$} (P4);

\draw[->,>=latex]
(P5) edge (C)
(P6) edge (C)
(P7) edge (C)
(P8) edge (C);

\draw [decorate, decoration={brace,amplitude=8pt},xshift=0pt,yshift=0pt]
(2.2,2) -- (3.8,2) node [black,midway,above, yshift=5pt] {\scriptsize $n m_{n}^{k(n)}$ steps};
\draw [decorate, decoration={brace,amplitude=8pt},xshift=0pt,yshift=0pt]
(2.2,0.67) -- (3.8,0.67) node [black,midway,above, yshift=5pt] {\scriptsize $nm_{n}^{i}$ steps};
\draw [decorate, decoration={brace,amplitude=8pt},xshift=0pt,yshift=0pt]
(2.2,-0.67) -- (3.8,-0.67) node [black,midway,above, yshift=5pt] {\scriptsize $nm_{n}$ steps};
\draw [decorate, decoration={brace,amplitude=8pt},xshift=0pt,yshift=0pt]
(2.2,-2) -- (3.8,-2) node [black,midway,above, yshift=5pt] {\scriptsize $n$ steps};

\draw[->,>=latex]
(C) edge node [midway, above, sloped]{\scriptsize $-m_{n}^{k(n)}$} (Q1)
(C) edge node [midway, above, sloped]{\scriptsize $-m_{n}^{i}$} (Q2)
(C) edge node [midway, above, sloped]{\scriptsize $-m_{n}$} (Q3)
(C) edge node [midway, above, sloped, pos=0.4]{\scriptsize $-1$} (Q4);

\draw[->,>=latex]
(Q1) edge node [midway, above, sloped]{\scriptsize $+m_{n}^{k(n)}$} (S2)
(Q2) edge node [midway, above, sloped]{\scriptsize $+m_{n}^{i}$} (S2)
(Q3) edge node [midway, above, sloped]{\scriptsize $+m_{n}$} (S2)
(Q4) edge node [midway, above, sloped, pos=0.6]{\scriptsize $+1$} (S2);

\draw[->,>=latex]
(Q2) edge node [near end, left]{\scriptsize $\eps_{i}(n)$} (Dead)
(Q3) edge node [pos=0.65, right]{\scriptsize $\eps_{1}(n)$} (Dead)
(Q4) edge node [midway, below]{\scriptsize $\eps_{0}(n)$} (Dead);

\draw[dotted, thick] (I9) -- (I10);
\draw[dotted, thick] (I11) -- (I12);

\end{tikzpicture}
\caption{All transition rewards are $0$ unless specified. Recall that $\sum \delta_{i}(n) \cdot \eps_{i}(n)$ is convergent and $\sum \delta_{j}(n) \cdot \eps_{i}(n)$ is divergent for all $i,j$ with $j > i$. The negative reward incurred before falling into the $\perp$ state is reimbursed. We do not show it in the figure for readability.
In the state before $s_{n+1}$, if the correct transition was chosen, the mean payoff is $-1/n$. If the incorrect transition was chosen, then either the mean payoff is $<-m_{n}/n$, or the risk of falling into $\perp$ is too high.
}%
\label{stepcounter}
\end{center}
\end{figure*}
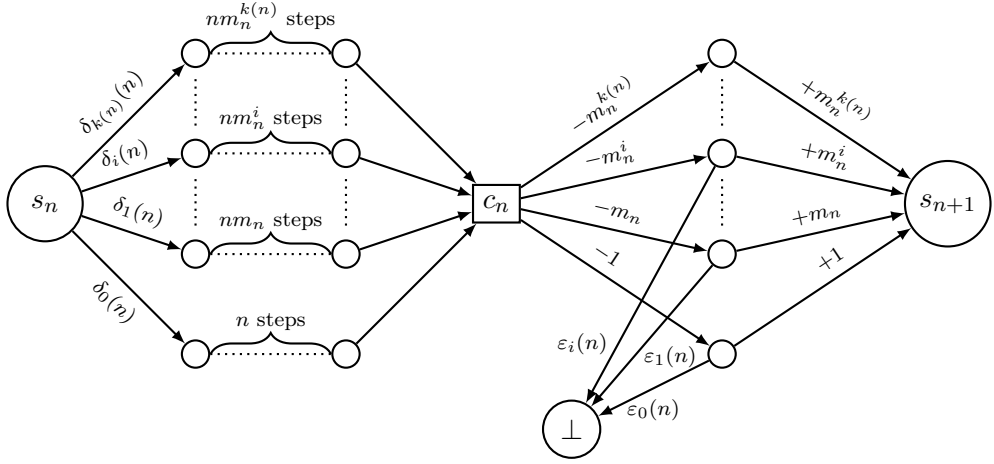

We use the example from \Cref{stepcounter}. It is very similar to \Cref{infinitegadget}, but differs in the following ways.
\begin{itemize}
\item The current total reward level is implicit in each state.
\item The step counter is no longer implicit in the state.
\item In the random choice, instead of changing the reward levels in each choice, it is the path length that differs.
\item The definition of $m_{n}$ is different, it is now $m_{n} \defeq \sum_{i = N^{*}}^{n-1} m_{i}^{k(n)}$ with $m_{N^{*}} \defeq 1$.
\end{itemize}

\noindent
We construct a finitely branching acyclic MDP $\mdp_{\text{RI}}$ (Reward Implicit) which has the total reward implicit in the state. We do so by chaining together the gadgets from \Cref{stepcounter} as is shown in \Cref{chain}.

In order to convince ourselves that the history of the play in past gadgets does not affect the outcome of the current gadget, we do a brief analysis of the path length and total reward involved in a run going through the $n$th gadget.
Consider the scenario where the play took the $i$-th random choice. In this case, the path length is upper bounded by $\left( \sum_{i=N^{*}}^{n-1} 4 + m_{i}^{k(i)} \right)
+ 4 + m_{n}^{i} \leq 2m_{n}^{i}$.
In the case where the player chooses the $j$th controlled choice with $j \geq i$, this gives us an average reward of \[-\frac{m_{n}^{j}}{2nm_{n}^{i}}.\] This is $<-1$ when $j>i$ and converges to $0$ with $-\frac{1}{n}$ when $j=i$. The choices $j<i$ are losing due to the risk of falling into the losing sink as described previously.

Hence, the analysis within each gadget still reduces to mimicking the random choice in the controlled choice. This allows us to simply reuse the results from the step counter encoded case in order to obtain symmetrical results for the reward counter encoded case.

\begin{lem}\label{liminfmpstepval1}
 $\valueof{\mdp_{\text{\emph{RI}}}, \liminfmpobj}{(s_{0},0)} = 1$.
\end{lem}

\begin{proof}
We define a strategy $\sigma$ which, in $c_{n}$ always mimics the random choice in $s_{n}$.
Playing according to $\sigma$, the only way to lose is by dropping into the bottom state. This is because by mimicking, the mean payoff in each gadget is lower bounded by $-1/n$.
The rest of the proof is identical to \Cref{infwin}.
\end{proof}

\begin{lem}\label{liminfmpstepval0}
Any FR strategy $\sigma$ in $\mdp_{\text{\emph{RI}}}$ is such that $\probm_{\mdp_{\text{\emph{RI}}}, s_{0}, \sigma}(\liminfmpobj)=0$.
\end{lem}

\begin{proof}
When playing with finitely many memory modes, there are two ways for a run in $\mdp_{\text{RI}}$ to lose. Either it falls into a losing sink, or it never falls into a sink but its mean payoff is $<-1$. The proof that either of these occurs with probability $1$ is the same as in
\Cref{inflose}.
\end{proof}

\begin{thm}\label{mpstepepslower}
There exists a countable, finitely branching, acyclic MDP $\mdp_{\text{\emph{RI}}}$ with initial state $(s_{0},0)$ with the total reward implicit in the state such that
\begin{itemize}
\item $\valueof{\mdp_{\text{\emph{RI}}}, \liminfmpobj}{(s_{0},0)} = 1$,
\item for all FR strategies $\sigma$, we have $\probm_{\mdp_{\text{\emph{RI}}}, (s_{0},0), \sigma}(\liminfmpobj)=0$.
\end{itemize}
\end{thm}

\begin{proof}
This follows from \Cref{liminfmpstepval1} and \Cref{liminfmpstepval0}.
\end{proof}

\begin{rem}\label{rem:glue}
The MDPs from \Cref{infinitegadget}
and \Cref{stepcounter}
show that good strategies for $\liminfmpobj$ require ``at least''
a reward counter and a step counter, respectively.
There does, of course, exist a \emph{single MDP}
where good strategies for $\liminfmpobj$
require at least both a step
counter and a reward counter. We construct such an MDP by `gluing' the two
different MDPs together via an initial random state which points to each with
probability $1/2$.
\end{rem}


\subsubsection{Optimal Strategies}


\begin{figure}[t]
\begin{center}
\begin{tikzpicture}

\node[draw, minimum height=0.5cm, minimum width=0.5cm] (I1) at (-0.5,0) {\scriptsize $s_{0}$};
\node[] (I2) at (-0.5,-4.5) {};
\node[] (I3) at (7,0) {};
\node[] (I4) at (7,-4.5) {};

\node[] (HI1) at (-0.5,-3) {};
\node[] (HI2) at (7,-1.5) {};
\node[] (HI3) at (2,-6) {};
\node[] (HI4) at (4,-6) {};
\node[] (HI5) at (6,-6) {};

\draw [decorate,decoration={brace,mirror,amplitude=10pt},xshift=-4pt,yshift=0pt] (-2,0.75) -- (-2,-0.75) node [black,midway,xshift=-27pt] {\footnotesize Row 1};

\draw [decorate,decoration={brace,mirror,amplitude=10pt},xshift=-4pt,yshift=0pt] (-2,-3.75) -- (-2,-5.25) node [black,midway,xshift=-27pt] {\footnotesize Row 2};

\node[draw, circle, minimum width=1cm] (W1) at (1,0) {\scriptsize $i$};
\node[draw, circle, minimum width=1cm] (W2) at (3,0) {\scriptsize $i+1$};
\node[draw, circle, minimum width=1cm] (W3) at (5,0) {\scriptsize $i+2$};

\node[draw, circle, minimum width=1cm] (W4) at (1,-4.5) {\scriptsize $i+1$};
\node[draw, circle, minimum width=1cm] (W5) at (3,-4.5) {\scriptsize $i+2$};
\node[draw, circle, minimum width=1cm] (W6) at (5,-4.5) {\scriptsize $i+3$};

\node[draw, circle, minimum width=1cm] (R1) at (1,-1.5) {\scriptsize $r_{i, 1}$};
\node[draw, circle, minimum width=1cm] (R2) at (3,-1.5) {\scriptsize $r_{i+1, 1}$};
\node[draw, circle, minimum width=1cm] (R3) at (5,-1.5) {\scriptsize $r_{i+2, 1}$};

\node[draw, circle, fill=black] (B1) at (2,-1.5) {};
\node[draw, circle, fill=black] (B2) at (2,-2.25) {};
\node[draw, circle, fill=black] (B3) at (2,-3) {};

\node[draw, circle, fill=black] (B4) at (4,-1.5) {};
\node[draw, circle, fill=black] (B5) at (4,-2.25) {};
\node[draw, circle, fill=black] (B6) at (4,-3) {};

\draw[->, >=latex, dotted, thick] (I1) -- (W1);
\draw[->, >=latex, dotted, thick] (I2) -- (W4);
\draw[->, >=latex, dotted, thick] (W3) -- (I3);
\draw[->, >=latex, dotted, thick] (HI1) -- (W4);
\draw[->, >=latex, dotted, thick] (R3) -- (HI2);
\draw[->, >=latex, dotted, thick] (W6) -- (I4);
\draw[->, >=latex, dotted, thick] (W4) -- (HI3);
\draw[->, >=latex, dotted, thick] (W5) -- (HI4);
\draw[->, >=latex, dotted, thick] (W6) -- (HI5);

\draw[->,>=latex]
(W1) edge node[right, midway]{\scriptsize $-m_{i+2}$} (R1)
(R1) edge (B1)
(B1) edge (B2)
(B2) edge (B3)
(B4) edge (B5)
(B5) edge (B6)
(B3) edge node[above, midway, sloped]{\scriptsize $+m_{i+2}$} (W5)
(W2) edge node[right, midway]{\scriptsize $-m_{i+3}$} (R2)
(R2) edge (B4)
(B6) edge node[above, midway, sloped]{\scriptsize $+m_{i+3}$}(W6)
(W3) edge node[right, midway]{\scriptsize $-m_{i+4}$} (R3);
\draw[->,>=latex]
(W1) edge (W2)
(W2) edge (W3)
(W4) edge (W5)
(W5) edge (W6);

\end{tikzpicture}
\caption{Each row represents a copy of the MDP depicted in \Cref{chain}.
  Each white circle labeled with a number $i$
  represents the correspondingly numbered gadget (like in \Cref{infinitegadget})
  from that MDP\@. Now, instead
  of the bottom states in each gadget leading to an infinite losing chain,
  they lead to a restart state $r_{i,j}$ which leads to a fresh copy of the
  MDP (in the next row).
  Each restart incurs a penalty guaranteeing that the mean payoff dips
  below $-1$ before refunding it and continuing on in the next copy of the
  MDP\@. The states $r_{i,j}$ are labeled such that the $j$ indicates that if a
  run sees this state, then it is the $j$th restart. The $i$ indicates that
  the run entered the restart state from the $i$th gadget of the current copy
  of the MDP\@. The black states are dummy states inserted in order to preserve
  path length throughout.
\vspace{-5mm}
}%
\label{restart}
\end{center}
\end{figure}
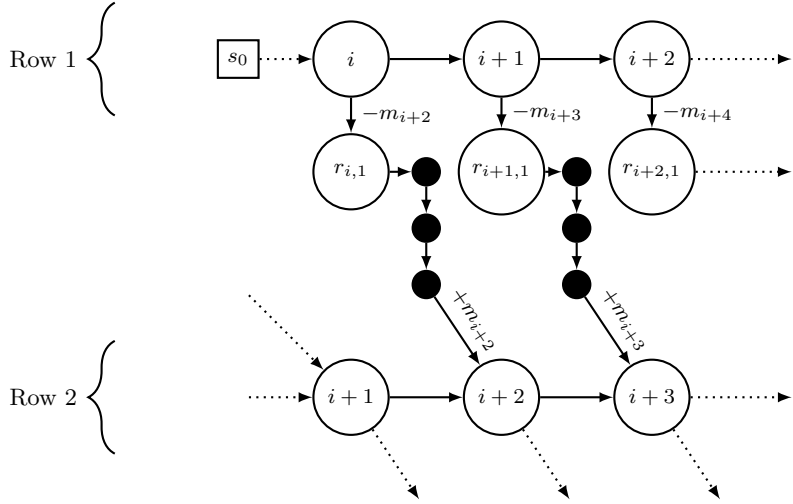


Even for acyclic MDPs with the step counter implicit in the state,
optimal (and even almost sure winning) strategies for $\liminfmpobj$
require infinite memory.
To prove this,
we consider a variant of the MDP from the previous section
which has been augmented to include restarts from the $\perp$ states.
For the rest of the section, $\mdp$ is the MDP constructed in \Cref{restart}. Initially the gadgets used are like \Cref{infinitegadget}, then we present similar results using \Cref{stepcounter} as the gadgets instead.

\begin{rem}
$\mdp$ is acyclic, finitely branching and the step counter is implicit
in the state. We now refer to the rows of \Cref{restart} as gadgets,
i.e., a gadget is a single instance of \Cref{chain} where the $\perp$ states lead to the next row.
\end{rem}

\begin{lem}\label{almostwin}
There exists a strategy $\sigma$ such that $\probm_{\mdp,\sigma,s_{0}}(\liminfmpobj)=1$.
\end{lem}
\begin{proof}[Outline of the proof.]
Recall the strategy $\zstrat_{1/2}$ defined in \Cref{infwin} which achieves at least $1/2$ in each gadget that it is played in.
We then construct the almost surely winning strategy $\zstrat$ by concatenating $\zstrat_{1/2}$ strategies in the sense that $\zstrat$ plays just like $\zstrat_{1/2}$ in each gadget from each gadget's start state.

Since $\zstrat$ achieves at least $1/2$ in every gadget that it sees, with probability $1$, runs generated by $\zstrat$ restart only finitely many times.
The intuition is then that a run restarting finitely many times must spend an infinite tail in some final gadget.
Since $\zstrat$ mimics in every controlled state, not restarting anymore directly implies that the total payoff is eventually always $\geq 0$. Hence all runs generated by $\zstrat$ and restarting only finitely many times satisfy $\liminfmpobj$.
Therefore all but a nullset of runs generated by $\zstrat$ are winning, i.e.\ $\probm_{\mdp,s_{0},\zstrat}(\liminfmpobj)=1$.
\end{proof}

\begin{proof}[Full Proof]

We will show that there exists a strategy $\zstrat$ that satisfies the mean payoff objective with probability 1 from $s_{0}$.
Towards this objective we recall the strategy $\zstrat_{1/2}$ defined in \Cref{infwin}.
In a given gadget of this MDP with restarts, playing $\sigma_{1/2}$ in said gadget, there is a probability of at most $1/2$ of restarting in that gadget.
We then construct strategy $\zstrat$ by concatenating $\zstrat_{1/2}$ strategies in the sense that $\zstrat$ plays just like $\zstrat_{1/2}$ in each gadget from each gadget's start state.

Let $\playset$ be the set of runs induced by $\zstrat$ from $s_{0}$. We partition $\playset$ into the sets $\playset_{i}$ and $\playset_{\infty}$ of runs such that
$\playset = \left( \bigcup_{i=0}^{\infty} \playset_{i} \right) \cup \playset_{\infty}$. We define for $i=0$
\[\playset_{0} \eqdef \{ \rho \in \playset \mid
    \forall \ell \in \N.\, \neg \eventually(r_{\ell, 1}) \},\]
for $i \geq 1$
\[\playset_{i} \eqdef \{ \rho \in \playset \mid
    \exists j \in \N.\, \eventually(r_{j,i}) \wedge \forall \ell \in \N.\, \neg \eventually(r_{\ell,i+1}) \}\]
and
\[\playset_{\infty} \eqdef \{ \rho \in \playset \mid
    \forall i \in \N\ \exists j \in \N.\,  \eventually(r_{j,i})\}.\]
That is to say for all $i \in \N$, $\playset_{i}$ is the set of runs in $\playset$ that restart exactly $i$ times and $\playset_{\infty}$ is the set of runs in $\playset$ that restart infinitely many times.

We go on to define the sets of runs $\playset_{\geq i} \eqdef \bigcup_{j=i}^{\infty} \playset_{j} $ which are those runs which restart at least $i$ times. In particular note that $\playset_{\infty}= \bigcap_{i=0}^{\infty}\playset_{\geq i}$ and $\playset_{\geq i+1} \subseteq \playset_{\geq i}$.

By construction, any run $\rho \in \playset_{\infty}$ is losing since the negative reward that is collected upon restarting instantly brings the mean payoff below $-1$ by definition of $m_{n}$. Thus restarting infinitely many times translates directly into the mean payoff dropping below $-1$ infinitely many times and thus a strictly negative $\liminf$ mean payoff. As a result it must be the case that $\playset_{\infty} \subseteq \neg\liminfmpobj$.

After every restart, the negative reward is reimbursed. Intuitively, going
through finitely many restarts does not damage the chances of winning.
We now show that, except for a nullset, the runs restarting only finitely many
times satisfy the objective.
Indeed, every run with only finitely many restarts must spend an infinite tail
in some final gadget in which it does not restart.
In this final gadget, the strategy plays just like $\sigma_{1/2}$, which means that it mimics the random choice in every controlled state.
Since, by assumption, there are no more restarts, we obtain
$\probm_{\mdp,s_{0},\zstrat}(\playset_{i}) = \probm_{\mdp,s_{0},\zstrat}(\playset_{i} \wedge \forall j \in \N, \always (\neg r_{j,i+1}))$.
We then apply \Cref{infwin} to obtain that
\begin{equation}\label{eq:alwayswin}
\probm_{\mdp,s_{0},\zstrat}(\playset_{i}) = \probm_{\mdp,s_{0},\zstrat}(\playset_{i} \wedge \forall j \in \N, \always (\neg r_{j,i+1})) = \probm_{\mdp,s_{0},\zstrat}(\playset_{i} \wedge \liminfmpobj).
\end{equation}
In other words, except for a nullset, the runs restarting finitely often (here
$i$ times) satisfy $\liminfmpobj$.
Furthermore, notice that from this observation, the sets $\playset_{i}$
partition the set of winning runs.

We show now that $\probm_{\mdp,s_{0},\zstrat}(\playset_{\infty}) = 0$. We do so firstly by showing by induction that $\probm_{\mdp,s_{0},\zstrat}(\playset_{\geq i}) \leq 2^{-i}$ for $i \geq 1$, then applying the continuity of measures from above to obtain that $\probm_{\mdp,s_{0},\zstrat}(\playset_{\infty}) = 0$.

Our base case is $i=1$. $\playset$, by definition of $\zstrat$, is the set of
runs induced by playing $\zstrat_{1/2}$ in every gadget. By \Cref{infwin}, $\zstrat$
attains $\geq 1/2$ in every gadget. Therefore in particular the probability of a run leaving the first gadget is no more than $1/2$, i.e.\ $\probm_{\mdp,s_{0},\zstrat}(\playset_{\geq 1}) \leq 1/2$.

Now suppose that $\probm_{\mdp,s_{0},\zstrat}(\playset_{\geq i}) \leq 2^{-i}$. After restarting at least $i$ times, the probability of a run restarting at least once more is still $\leq 1/2$ since the strategy being played in every gadget is $\zstrat_{1/2}$. Hence \[\probm_{\mdp,s_{0},\zstrat}(\playset_{\geq i+1}) \leq \probm_{\mdp,s_{0},\zstrat}(\playset_{\geq i}) \cdot \dfrac{1}{2} \leq 2^{-(i+1)}\] which is what we wanted.

Now we use the fact that $\playset_{\infty}= \bigcap_{i=0}^{\infty}\playset_{\geq i}$ and $\playset_{\geq i+1} \subseteq \playset_{\geq i}$ to apply continuity of measures from above and obtain:
\[
\probm_{\mdp,s_{0},\zstrat}(\playset_{\infty}) = \probm_{\mdp,s_{0},\zstrat} \left( \bigcap_{i=0}^{\infty} \playset_{\geq i} \right) =
\lim_{i \to \infty} \probm_{\mdp,s_{0},\zstrat}(\playset_{\geq i}) \leq \lim_{i \to \infty} 2^{-i} = 0.
\]
Hence $\playset_{\infty}$ is a null set.

We can now write down the following:
\begin{align*}
1 &= \probm_{\mdp,s_{0},\zstrat}(\playset) \\
& = \left( \sum_{i=0}^{\infty} \probm_{\mdp,s_{0},\zstrat}(\playset_{i}) \right) + \probm_{\mdp,s_{0},\zstrat}(\playset_{\infty}) &\text{by partition of }\playset \\
& = \left( \sum_{i=0}^{\infty} \probm_{\mdp,s_{0},\zstrat}(\playset_{i} \wedge \liminfmpobj) \right) + \probm_{\mdp,s_{0},\zstrat}(\playset_{\infty}) & \text{by \Cref{eq:alwayswin}} \\
& = \left( \sum_{i=0}^{\infty} \probm_{\mdp,s_{0},\zstrat}(\playset_{i} \wedge \liminfmpobj) \right) \\
& \hspace{5cm} + \probm_{\mdp,s_{0},\zstrat}(\playset_{\infty} \wedge \liminfmpobj) & \text{by } \probm_{\mdp,s_{0},\zstrat}(\playset_{\infty}) = 0 \\
& = \probm_{\mdp,s_{0},\zstrat}(\liminfmpobj) &\text{by partition of }\liminfmpobj
\end{align*}

Thus $\probm_{\mdp,s_{0},\zstrat}(\playset) = \probm_{\mdp,s_{0},\zstrat}(\liminfmpobj) = 1$, i.e.\ $\zstrat$ wins almost surely.
\end{proof}

\begin{restatable}{lem}{lemmaalmostlose}\label{almostlose}
For any FR strategy $\sigma$, $\probm_{\mathcal{M}, \sigma, s_{0}}(\liminfmpobj)=0$.
\end{restatable}

\begin{proof}[Outline of the proof.]
Let $\sigma$ be any FR strategy. We partition the runs generated by $\sigma$ into runs restarting infinitely often, and those restarting only finitely many times.
Any runs restarting infinitely often are losing by construction.
The runs restarting only finitely many times spend an infinite tail in a given gadget, letting the mean payoff dip below $-1$ infinitely many times with probability 1 by \Cref{inflose}.
Hence we have that $\probm_{\mathcal{M}, \sigma, s_{0}}(\liminfmpobj)=0$.
\end{proof}

\begin{proof}[Full proof]
There are two ways to lose when playing in this MDP:\@ either the mean payoff dips below $-1$ infinitely often because the run takes infinitely many restarts, or the run only takes finitely many restarts, but the mean payoff drops below $-1$ infinitely many times in the last copy of the gadget that the run stays in. Recall that in \Cref{inflose} we showed that any FR strategy with probability 1 either restarts or lets the mean payoff dip below $-1$ infinitely often.

Let $\sigma$ be any FR strategy and let $\playset$ to be the set of runs induced by $\sigma$ from $s_{0}$.
We partition $\playset$ into the sets $\playset_{i}$ and $\playset_{\infty}$ of runs such that
$\playset = \left( \bigcup_{i=0}^{\infty} \playset_{i} \right) \cup \playset_{\infty}$. Where we define for $i=0$
\[\playset_{0} \eqdef \{ \rho \in \playset \mid
    \forall \ell \in \N, \neg \eventually(r_{\ell,1}) \},\]
for $i \geq 1$
\[\playset_{i} \eqdef \{ \rho \in \playset \mid
    \exists j \in \N, \eventually(r_{j, i}) \wedge \forall \ell \in \N, \neg \eventually(r_{\ell,i+1}) \}\]
and
\[\playset_{\infty} \eqdef \{ \rho \in \playset \mid
    \forall i, \exists j  \text{ F}(r_{j,i})\}.\]
That is to say for all $i \in \N$, $\playset_{i}$ is the set of runs in $\playset$ that restart exactly $i$ times and $\playset_{\infty}$ is the set of runs in $\playset$ that restart infinitely many times.

We go on to define the sets of runs $\playset_{\geq i} \eqdef \bigcup_{j=i}^{\infty} \playset_{j} $ which are those runs which restart at least $i$ times. In particular note that $\playset_{\infty}= \bigcap_{i=0}^{\infty}\playset_{\geq i}$ and $\playset_{\geq i+1} \subseteq \playset_{\geq i}$.

Note that any run in $\playset_{\infty}$ is losing by construction. The negative reward that is collected upon restarting instantly brings the mean payoff below $-1$ by definition of $m_{n}$. Thus restarting infinitely many times translates directly into the mean payoff dropping below $-1$ infinitely many times. Thus $\playset_{\infty} \subseteq \neg\liminfmpobj$ and so it follows that  $\probm_{\mdp,s_{0},\sigma}(\playset_{\infty}) = \probm_{\mdp,s_{0},\sigma}(\playset_{\infty} \wedge \neg \liminfmpobj)$. Since the sets $\playset_{i}$ and $\playset_{\infty}$ partition $\playset$ we have that:
\[
\probm_{\mdp,s_{0},\sigma}(\playset) = \left( \sum_{i=0}^{\infty} \probm_{\mdp,s_{0},\sigma}(\playset_{i}) \right) + \probm_{\mdp,s_{0},\sigma}(\playset_{\infty}).
\]

It remains to show that every set $\playset_{i}$ is almost surely losing, i.e.\ $\probm_{\mdp,s_{0},\sigma}(\playset_{i}) = \probm_{\mdp,s_{0},\sigma}(\playset_{i} \wedge \neg \liminfmpobj)$.
Consider a run $\rho \in \playset_{i}$. By definition it restarts exactly $i$ times. As a result, it spends infinitely long in the $i+1$st gadget.
Because $\sigma$ is an FR strategy, it must be the case that any substrategy $\sigma^{*}$ induced by $\sigma$ that is played in a given gadget is also an FR strategy.
This allows us to apply \Cref{inflose} to obtain that
\begin{equation}\label{eq:alwayslose}
\probm_{\mdp,s_{0},\sigma}(\playset_{i}) =
\probm_{\mdp,s_{0},\sigma}\left(\playset_{i} \wedge (\neg \liminfmpobj \vee \exists j \in \N,
\eventually (r_{j,i+1}))\right).
\end{equation}
However, any run $\rho \in \playset_{i}$ never sees any state $r_{j,i+1}$ for any $j$ by definition. Therefore it follows that
\[
\probm_{\mdp,s_{0},\sigma}\left(\playset_{i} \wedge (\neg \liminfmpobj \vee \exists j \in \N, \eventually (r_{j,i+1}))\right)  =
\probm_{\mdp,s_{0},\sigma}\left(\playset_{i} \wedge (\neg \liminfmpobj )\right)
\]
Hence $\probm_{\mdp,s_{0},\sigma}(\playset_{i}) = \probm_{\mdp,s_{0},\sigma}(\playset_{i} \wedge \neg \liminfmpobj)$ as required.

As a result we have that
\begin{align*}
1 &= \probm_{\mdp,s_{0},\sigma}(\playset) \\
& = \left( \sum_{i=0}^{\infty} \probm_{\mdp,s_{0},\sigma}(\playset_{i}) \right) + \probm_{\mdp,s_{0},\sigma}(\playset_{\infty})    &\text{by partition of }\playset \\
& = \left( \sum_{i=0}^{\infty} \probm_{\mdp,s_{0},\sigma}(\playset_{i} \wedge \neg\liminfmpobj) \right) + \probm_{\mdp,s_{0},\sigma}(\playset_{\infty} \wedge \neg\liminfmpobj) &\text{by \Cref{eq:alwayslose}}\\
& = \probm_{\mdp,s_{0},\sigma}(\neg \liminfmpobj) &\text{by partition of }\playset
\end{align*}

That is to say that for any FR strategy $\sigma$, $\probm_{\mdp,s_{0},\sigma}(\liminfmpobj)=0$.
\end{proof}

From \Cref{almostwin} and \Cref{almostlose} we obtain the following theorem.

\begin{thm}\label{almostsummary}
There exists a countable, finitely branching and acyclic MDP $\mathcal{M}$ whose step counter is implicit in the state for which
$\state_0$ is almost surely winning $\liminfmpobj$, i.e.,
$\exists\hat{\sigma}\,\probm_{\mdp, s_{0}, \hat{\sigma}}(\liminfmpobj)=1$,
but every FR strategy $\sigma$ is such that
$\probm_{\mdp, s_{0}, \sigma}(\liminfmpobj)=0$.
In particular, almost sure winning strategies, when they exist, cannot be chosen
$k$-bit Markov for any $k \in \N$ for countable MDPs.
\end{thm}
\begin{proof}
Proved by \Cref{almostwin} and \Cref{almostlose}.
\end{proof}


\bigskip

Now we construct the MDP $\mdp_{\text{Restart}}$ by using \Cref{restart}, but we substitute the instances of \Cref{infinitegadget} gadgets with instances of \Cref{stepcounter} gadgets. This allows us to obtain the following results which state that optimal strategies for $\liminfmpobj$ requires infinite memory, even when the reward counter is implicit in the state.

\begin{lem}\label{liminfmpstepam1}
There exists an HD strategy $\sigma$ such that $\probm_{\mdp_{\text{\emph{Restart}}}, s_{0}, \sigma}(\liminfmpobj)=1$.
\end{lem}
\begin{proof}
The proof is identical to that of
\Cref{almostwin}.
\end{proof}

\begin{lem}\label{liminfmpstepam0}
For any FR strategy $\sigma$, $\probm_{\mdp_{\text{\emph{Restart}}}, s_{0}, \sigma}(\liminfmpobj)=0$.
\end{lem}
\begin{proof}
The proof is identical to that of
 \Cref{almostlose}.
\end{proof}

\begin{restatable}{thm}{thmmpstepoptlower}\label{mpstepoptlower}
There exists a countable, finitely branching and acyclic MDP $\mdp_{\text{\emph{Restart}}}$
whose total reward is implicit in the state where, for the initial state $s_0$,
\begin{itemize}
\item
there exists an HD strategy $\sigma$ s.t.\
$\probm_{\mdp_{\text{\emph{Restart}}}, s_{0}, \sigma}(\liminfmpobj)=1$.
\item
for every FR strategy $\sigma$,
$\probm_{\mdp_{\text{\emph{Restart}}}, s_{0}, \sigma}(\liminfmpobj)=0$.
\end{itemize}
\end{restatable}

\begin{proof}
This follows from \Cref{liminfmpstepam1} and \Cref{liminfmpstepam0}.
\end{proof}

\section{Total Payoff}\label{sec:totalpayoff}
\begin{table*}[hbtp]
\centering
\begin{tabular}{|ll||c|c|}
\hline
\multicolumn{2}{|l||}{Total Payoff}                                        & $\eps$-optimal & Optimal \\ \hline
\multicolumn{1}{|l|}{\multirow{2}{*}{Finitely branching}}   & Upper Bound & Det(RC)~\ref{fintpepsupper} &  Det(RC)~\ref{finoptupper}  \\ \cline{2-4}
\multicolumn{1}{|l|}{}                                      & Lower Bound & $\neg$Rand(F+SC)~\ref{infinitesummarytp} &    $\neg$Rand(F+SC)~\ref{almostsummarytp} \\ \hline
\multicolumn{1}{|l|}{\multirow{2}{*}{Infinitely branching}} & Upper Bound & Det(SC+RC)~\ref{inftpepsupper}  &  Det(SC+RC)~\ref{infoptuppertp} \\ \cline{2-4}
\multicolumn{1}{|l|}{}                                      & Lower Bond  &  \begin{tabular}[c]{@{}l@{}} $\neg$Rand(F+SC)~\ref{infinitesummarytp} and\\ $\neg$Rand(F+RC)~\ref{infbranchsteplowertp} \end{tabular}   &    \begin{tabular}[c]{@{}l@{}} $\neg$Rand(F+SC)~\ref{almostsummarytp} and\\ $\neg$Rand(F+RC)~\ref{infbranchsteplowertp} \end{tabular} \\ \hline
\end{tabular}
\caption{Strategy complexity of $\eps$-optimal/optimal
strategies for the total payoff objective in
infinitely/finitely branching MDPs.
}\label{table:totalpayoff}
\end{table*}

\subsection{Upper Bounds}

In order to tackle the upper bounds for the total payoff objective $\liminftpobj$, we
work with the derived MDPs $R(\mdp)$ and $S(\mdp)$ which encode the total reward and the
step counter into the state respectively. Once the total reward is encoded into the state,
the point payoff coincides with the total payoff. We use this observation to reduce
$\liminftpobj$ to $\liminfppobj$ and obtain our upper bounds from the corresponding point payoff results.

\begin{cor}\label{fintpepsupper}
Given a finitely branching MDP $\mdp$, there exist $\eps$-optimal strategies
for $\liminftpobj$ which use just a reward counter.
\end{cor}
\begin{proof}
We place ourselves in $R(\mdp)$ where
$\liminftpobj$ and $\liminfppobj$ coincide.
Thus we can apply \Cref{finpointpayoff} to obtain $\eps$-optimal MD
strategies for $\liminftpobj$ from every state of $R(\mdp)$.
By \Cref{totaltopoint} we can translate these MD strategies on $R(\mdp)$ back to
strategies on $\mdp$ with just a reward counter.
\end{proof}

\begin{restatable}{cor}{corinftpepsupper}\label{inftpepsupper}
Given an MDP $\mdp$ with initial state $s_{0}$,
\begin{itemize}
\item there exist $\eps$-optimal MD strategies for $\liminftpobj$ in $S(R(\mdp))$,
\item there exist $\eps$-optimal strategies for $\liminftpobj$ which use a step counter and a reward counter.
\end{itemize}
\end{restatable}
\begin{proof}
We consider the encoded system $R(\mdp)$ in which the reward counter is implicit in the state. Recall that total rewards in $\mdp$ correspond exactly to point rewards in $R(\mdp)$. We then apply \Cref{infpointpayoff} to $R(\mdp)$ to obtain  $\eps$-optimal MD strategies for $\liminfppobj$ in $S(R(\mdp))$.
\Cref{steptopoint} allows us to translate these MD strategies back to $R(\mdp)$ with a memory overhead of just a step counter. Then we apply
\Cref{totaltopoint} to translate these Markov strategies back to $\mdp$ with a memory overhead of just a reward counter. Hence $\eps$-optimal strategies for $\liminftpobj$ in $\mdp$ just use a step counter and a reward counter as required.
\end{proof}

\begin{rem}
While $\eps$-optimal strategies for mean payoff and total payoff
(in infinitely branching MDPs) have the same memory requirements, the step counter and the reward counter do not arise in the same way.
Both the step counter and reward counter used in $\eps$-optimal strategies for mean payoff  arise from the construction of $A(\mdp)$. However, in the case for total payoff, only the reward counter arises from the construction of $R(\mdp)$. The step counter on the other hand arises from the Markov strategy needed for point payoff in $R(\mdp)$.
\end{rem}

\begin{cor}\label{finoptupper}
Given a finitely branching MDP $\mdp$ and initial state $s_{0}$, optimal strategies, where they exist,
can be chosen with just a reward counter for $\liminftpobj$.
\end{cor}
\begin{proof}
We place ourselves in $R(\mdp)$ where
$\liminftpobj$ is shift invariant. Moreover, in $R(\mdp)$ the objectives
$\liminftpobj$ and $\liminfppobj$ coincide.
Thus we can apply \Cref{finpointpayoff} to obtain $\eps$-optimal MD
strategies for $\liminftpobj$ from every state of $R(\mdp)$.
From \Cref{epsilontooptimal} we obtain a single MD
strategy that is optimal from every state of $R(\mdp)$ that has an optimal
strategy. By \Cref{totaltopoint} we can translate this MD strategy on $R(\mdp)$ back to
a strategy on $\mdp$ with just a reward counter.
\end{proof}

\begin{cor}\label{infoptuppertp}
Given an MDP $\mdp$ and initial state $s_{0}$, optimal strategies, where they exist,
can be chosen with just a reward counter and a step counter for $\liminftpobj$.
\end{cor}

\begin{proof}
We place ourselves in
$S(R(\mdp))$ and apply \Cref{inftpepsupper} to obtain $\eps$-optimal MD
strategies for $\liminftpobj$ from every state of $S(R(\mdp))$.
While $\liminftpobj$ is not shift invariant in $\mdp$, it is shift invariant in $S(R(\mdp))$,
and thus we can apply \Cref{epsilontooptimal} to obtain a single MD
strategy that is optimal from every state of $S(R(\mdp))$ that has an optimal
strategy. The result then follows from \Cref{steptopoint} and \Cref{totaltopoint}.
\end{proof}

\subsection{Lower Bounds}

\subsubsection{$\eps$-optimal strategies}

\begin{thm}\label{infbranchsteplowertp}
There exists an infinitely branching MDP $\mdp$ as in \Cref{infinitebranchtp} with reward implicit in the state and initial state $s$ such that
\begin{itemize}
\item every FR strategy $\sigma$ is such that $\probm_{\mdp, s, \sigma} (\liminftpobj) = 0$
\item there exists an HD strategy $\sigma$ such that $\probm_{\mdp, s, \sigma} (\liminftpobj) = 1$.
\end{itemize}
Hence, optimal (and even almost-surely winning) strategies and $\eps$-optimal
strategies for $\liminftpobj$ require infinite memory
beyond a reward counter.
\end{thm}
\begin{proof}
This follows directly from~\cite[Theorem 4]{KMSW2017} and the observation that in \Cref{infinitebranchtp}, $\liminftpobj$, and co-B\"{u}chi objectives coincide.
\end{proof}

The statements and proofs of \Cref{infwin} and \Cref{inflose} also hold for $\liminftpobj$, giving us the following theorem.

\begin{thm}\label{infinitesummarytp}
There exists a countable, finitely branching and acyclic MDP $\mdp$ whose step counter is implicit in the state for which
$\valueof{\mdp,\liminftpobj}{s_{0}} = 1$ and any FR strategy $\zstrat$ is such that
$\probm_{\mdp, s_{0}, \zstrat}(\liminftpobj)=0$.
In particular, there are no $\eps$-optimal $k$-bit Markov strategies
for any $k \in \N$ and any $\eps < 1$ for
$\liminftpobj$ in
countable MDPs.
\end{thm}

\subsubsection{Optimal strategies}

The statements and proofs of \Cref{almostwin} and \Cref{almostlose} also hold for $\liminftpobj$, giving us the following theorem.

\begin{thm}\label{almostsummarytp}
There exists a countable, finitely branching and acyclic MDP $\mdp$ whose step counter is implicit in the state for which
$\state_0$ is almost surely winning $\liminftpobj$, i.e.,
$\exists\hat{\sigma}\,\probm_{\mdp, s_{0}, \hat{\sigma}}(\liminftpobj)=1$,
but every FR strategy $\sigma$ is such that
$\probm_{\mdp, s_{0}, \sigma}(\liminftpobj)=0$.
In particular, almost sure winning strategies, when they exist, cannot be chosen
$k$-bit Markov for any $k \in \N$ for countable MDPs.
\end{thm}
\begin{proof}
Proved by \Cref{almostwin} and \Cref{almostlose}.
\end{proof}

\section{Strengthening Results}\label{sec:strengthening}
The counterexamples we present in \Cref{sec:meanpayoff} feature finite but unbounded branching degree,
unbounded rewards and irrational transition probabilities.
In this section we show that the hardness does not depend on these aspects by
strengthening the counterexamples to have binary branching, bounded rewards
and rational transition probabilities.

Consider a new MDP $\mdp$ based on the MDP constructed in \Cref{infinitegadget} which now undergoes the following changes.
First we bound the branching degree by $2$.
We do so by replacing the outgoing transitions in states $s_{n}$ and $c_{n}$ of each gadget by binary trees with accordingly adjusted probabilities such that there is still a probability of $\delta_{i}(n)$ of receiving reward $i \cdot m_{n}$ in each gadget for $i \in \{0, 1, \ldots, k(n)\}$.

To adjust for the increased path lengths incurred by the modifications to each gadget, the construction in \Cref{chain} is accordingly modified by padding each vertical column of white states with extra transitions based on the number of transitions present in the matching gadget. As a result, path length is preserved even when skipping gadgets.
The construction in \Cref{restart} is similarly modified.

Second, we restrict the transition probabilities to rationals.
The transition probabilities $\delta_i(n)$ and $\eps_i(n)$
are replaced by close rationals in
$(\delta_i(n), \delta_i(n)+2^{-n})$
and in
$(\eps_i(n), \eps_i(n)+2^{-n})$, respectively.
These rationals are constructible, for example by approximation of $\delta_i(n)$ and $\eps_i(n)$ themselves.
Since these new rational probabilities are so close to the original ones,
all of the relevant convergence and divergence of series is preserved.

\begin{figure}
    \begin{tikzpicture}

    \node[draw,circle] (S1) at (0,0) {$s_{n}$};


    %
    %
    %
    %
    %
    \node[draw,circle] (MidBotL) at (2,-0.7){};
    \node[draw,circle] (MidTopL) at (2,0.7){};

    \node[draw, circle] (BotL) at (2,-1.5) {};
    \node[draw, circle] (TopL) at (2,1.5) {};

    \draw [dotted, thick] (MidBotL) -- (MidTopL);
    \draw [dotted, thick] (MidTopL) -- (TopL);

    %
    %
    %
    %
    \coordinate[shift={(0mm,2.5mm)}] (Mshift) at (M);


    \draw[->,>=latex] (S1) edge[bend left] node[above, midway, xshift=-7pt]{$\delta_{k(n)}(n)$} (TopL)
                    (S1) edge[bend left=20] node[below, midway]{$\delta_{j}(n)$} (MidTopL)
                    (S1) edge[bend right=20] node[above, midway]{$\delta_{1}(n) $} (MidBotL)
                    (S1) edge[bend right] node[below, midway, xshift=-5pt]{$\delta_{0}(n)$} (BotL);


    \node (T) at (4,0) {is replaced by};


	\node[draw,circle] (S2) at (6,0) {$s_{n}$};

	\node[draw,circle] (B1) at (7,1.5) {};
	\node[draw,circle] (B2) at (7,-1.5) {};

	\draw[->,>=latex] (S2) edge node[left, near end]{\small $p_{2}(n)$} (B1)
	(S2) edge node[left, near end]{\small $p_{1}(n)$} (B2);

	\node[draw,circle] (C1) at (8,2) {};
	\node[draw,circle] (C2) at (8,1) {};
	\node[draw,circle] (C3) at (8,-1) {};
	\node[draw,circle] (C4) at (8,-2) {};

	\draw[->,>=latex] (B1) edge node[above, near start, xshift=-5pt]{\footnotesize $p_{2,2}(n)$} (C1)
	(B1) edge node[below, near start, yshift=-3pt]{\footnotesize $p_{2,1}(n)$} (C2)
	(B2) edge node[above, near start, yshift=3pt]{\footnotesize $p_{1,2}(n)$} (C3)
	(B2) edge node[below, near start, xshift=-5pt]{\footnotesize $p_{1,1}(n)$} (C4);

	\node[draw,circle] (D1) at (9,2.25) {};
	\node[draw,circle] (D2) at (9,0.75) {};
	\node[draw,circle] (D3) at (9,-0.75) {};
	\node[draw,circle] (D4) at (9,-2.25) {};

	\draw [dotted, thick] (C1) edge (D1)
	(C2) edge (D2)
	(C3) edge (D3)
	(C4) edge (D4);

	\node (V1) at (9,1.5) {\vdots};
	\node (V2) at (9,-1.5) {\vdots};

	\node (E1) at (10, 2.25) {$\delta_{k(n)}(n)$};
	\node (E2) at (10.1, 0.75) {$\delta_{\left\lceil\tfrac{k(n)}{2}\right\rceil}(n)$};
    \node (E3) at (10.1, -0.75) {$\delta_{\left\lfloor\tfrac{k(n)}{2}\right\rfloor}(n)$};
    \node (E4) at (9.8, -2.25) {$\delta_{0}(n)$};

    \end{tikzpicture}
    \caption{Schema for replacing arbitrary finite branching with binary branching in \Cref{infinitegadget}.}%
    \label{kbranchtobinarybranch}
    \end{figure}
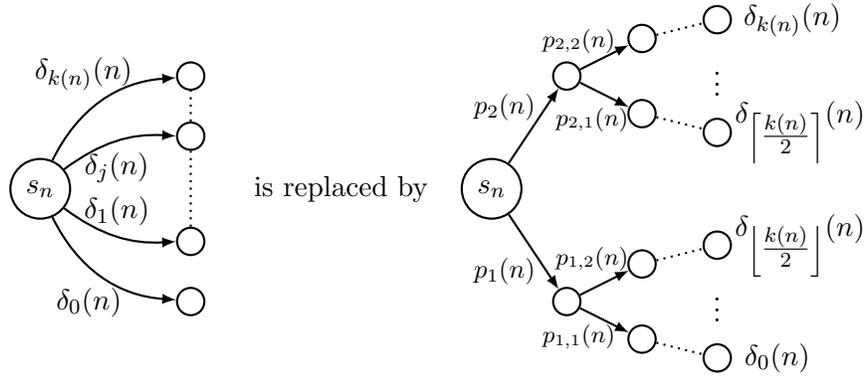

\begin{defi}[Binary branching]\label{def:binarybranching}
We formally define how to modify the MDPs in \Cref{sec:mplowerbounds} such that they have a branching degree of no more than $2$.

In each gadget in \Cref{infinitegadget} and \Cref{stepcounter}, we do as is shown in \Cref{kbranchtobinarybranch}. I.e.\ the outgoing transitions from $s_n$ and $c_n$ are replaced by a binary tree of depth at most $\left\lceil \lg(k(n)+1) \right\rceil$. Because $c_n$ is player controlled, we do not need to define any new transition probabilities. For the outgoing transitions from $s_n$, new probabilities must be defined such that the probability of receiving reward $im(n)$ is still $\delta_{i}(n)$.
For illustrative purposes, the transition probabilities for the initial branching are as follows.
\[
p_{1}(n) \eqdef \sum _{i=0}^{\left\lfloor \tfrac{k(n)}{2} \right\rfloor} \delta_{i}(n)
 \text{ and } p_{2}(n) \eqdef \sum_{i=\left\lceil \tfrac{k(n)}{2} \right\rceil}^{k(n)} \delta_{i}(n)
\]
All other transition probabilities are obtained inductively.

The extended path lengths are mirrored in a modified \Cref{chain} by padding
the path lengths by $4 + 2\lceil\lg(k(n))\rceil$ steps instead of $4$
steps. Similarly, in \Cref{restart} we pad the length
of the chains of black states by an extra $2 \lceil\lg(k(n)) \rceil$ steps.
So the step counter is still implicit in the state.
\end{defi}


\begin{defi}[Rational probabilities]\label{def:rationalprobs}
We define the probabilities $\gamma_{i}(n)$ and $\theta_{i}(n)$ as follows. We set $\gamma_{i}(n) \in (\delta_{i}(n), \delta_{i}(n) + 2^{-n}) \cap \mathbb{Q}$ and similarly set $\theta_{i}(n) \in (\eps_{i}(n), \eps_{i}(n) + 2^{-n}) \cap \mathbb{Q}$.
Furthermore, we define new functions $g^*$, $h^*$ and $k^*$ as follows.
We set \[
g^*(i) \eqdef \min \left\{ N : \left( \sum_{n>N} \gamma_{i-1}(n) \theta_{i-1}(n) \right)  \le 2^{-i}  \right\}, \quad h^*(1) \eqdef 2
\]
\[
h^*(i+1) \eqdef \left\lceil \max \left\{ g^*(i+1), \text{Tower}(i+2),  \min \left\{ m+1 \in \N :\sum^{m}_{n=h^{*}(i)} \theta_{i-1}(n) \geq 1 \right\} \right\} \right\rceil.
\]
which yield $k^*(i) \eqdef h^{*^{-1}}(n)$.

Note that $g^*(i)$ is well defined since
\begin{align*}
\sum_{n>N} \gamma_{i-1}(n) \theta_{i-1}(n) & < \sum_{n>N} \delta_{i-1}(n) \eps_{i-1}(n) + 2^{-n}(\delta_{i-1}(n) + \eps_{i-1}(n) + 2^{-n})\\
& <\sum_{n>N} \delta_{i-1}(n) \eps_{i-1}(n) + 2^{-n} \cdot 3
\end{align*}
is convergent for all $i$.
Similarly, $h^*(i)$ is also well defined since
\[
\sum^{\infty}_{n=h^{*}(i)} \theta_{i-1}(n) > \sum^{\infty}_{n=h^{*}(i)} \eps_{i-1}(n)
\]
which diverges for all $i$.
\end{defi}

\begin{rem}\label{rem:rationalprobs}
We now make sure that the results from \Cref{sec:meanpayoff} still hold when we change \Cref{infinitegadget} by replacing the transition probabilities $\delta_{i}(n)$ and $\eps_{i}(n)$ with $\gamma_{i}(n)$ and $\theta_{i}(n)$, respectively.
To this end we must check that some crucial results still hold. Namely, that \Cref{lem:welldefined} and \Cref{alglose} still hold given the modified transition probabilities.

In order to show that \Cref{lem:welldefined} still holds, we must show that $\sum^{k(n)-1}_{i=0} \gamma_i (n) < 1$. I.e.\ we get
\begin{align*}
\sum^{k(n)-1}_{i=0} \gamma_i (n) & \leq \sum^{k(n)-1}_{i=0} \delta_i (n) + \sum^{k(n)-1}_{i=0} 2^{-n} \\
& = \sum^{k(n)-1}_{i=0} \delta_i (n) + (k(n)-1) 2^{-n} \\
& < \sum^{k(n)}_{i=1} 2^{-i} + 2^{k(n)-1} \cdot 2^{-n} \\
& \leq \sum^{k(n)}_{i=1} 2^{-i} + 2^{-\tfrac{n}{2}} < 1 & \text{since } k(n) < \dfrac{n}{2}
\end{align*}
That is to say that the transition probabilities are indeed well defined using rational probabilities $\gamma_i (n)$ in lieu of $\delta_{i} (n)$.

Similarly, we must now show that the following sum diverges.
\[
\sum_{n=k^{-1}(2)}^{\infty} \Big( \gamma_{j(n)}(n)(\alpha_{n} \theta_{j(n)}(n) + (1- \alpha_{n}) \theta_{i(n)}(n)) + \gamma_{i(n)}(n)(\alpha_{n} + (1- \alpha_{n}) \theta_{i(n)}(n)) \Big)
\]
We do so by noticing that
\begin{align*}
\sum_{n=k^{-1}(2)}^{\infty} \Big( \gamma_{j(n)}(n)(\alpha_{n} \theta_{j(n)}(n) + (1- \alpha_{n}) & \theta_{i(n)}(n)) + \gamma_{i(n)}(n)(\alpha_{n} + (1- \alpha_{n}) \theta_{i(n)}(n)) \Big) \\
\geq \sum_{n=k^{-1}(2)}^{\infty} \Big( \delta_{j(n)}(n)(\alpha_{n} \eps_{j(n)}(n) + (1- \alpha_{n}) & \eps_{i(n)}(n)) + \delta_{i(n)}(n)(\alpha_{n} + (1- \alpha_{n}) \eps_{i(n)}(n)) \Big) \\
& \text{since  }\gamma_{i}(n) \geq \delta_i (n) \text{ and } \theta_i (n) \geq \eps_i (n) \text{ for all } i.
\end{align*}
Hence \Cref{alglose} yields the desired divergence result.

Putting both of the above results together, we can obtain rational probability versions of \Cref{infwin} and \Cref{inflose}.
\end{rem}

Combining these constructions allows us to obtain the following properties.

\begin{thm}
There exists a countable, acyclic MDP $\mdp$, whose step counter is implicit in the state, whose transition probabilities are rational and whose branching degree is bounded by $2$ for which
$\valueof{\mdp,\liminfmpobj}{s_{0}} = 1$ and any FR strategy $\sigma$ is such that
$\probm_{\mdp, s_{0}, \sigma}(\liminfmpobj)=0$.
In particular, there are no $\eps$-optimal step counter plus finite memory strategies
for any $\varepsilon < 1$ for the
$\liminfmpobj$ objective for countable MDPs.
\end{thm}

\begin{proof}
This follows from \Cref{infwin}, \Cref{inflose} by modifying the constructions in \Cref{infinitegadget} and \Cref{chain} as detailed in \Cref{def:binarybranching} and \Cref{def:rationalprobs}.
\end{proof}

\begin{thm}
There exists a countable, acyclic MDP $\mdp$, whose step counter is implicit in the state, whose transition probabilities are rational and whose branching degree is bounded by $2$ for which
$s_{0}$ is almost surely winning for $\liminfmpobj$ and any FR strategy $\sigma$ is such that
$\probm_{\mdp, s_{0}, \sigma}(\liminfmpobj)=0$.
In particular, almost sure winning strategies, when they exist, cannot be chosen
with a step counter plus finite memory for countable MDPs.
\end{thm}

\begin{proof}
This follows from \Cref{almostwin}, \Cref{almostlose} by modifying the constructions in \Cref{infinitegadget} and \Cref{restart} as detailed in \Cref{def:binarybranching} and \Cref{def:rationalprobs}.
\end{proof}

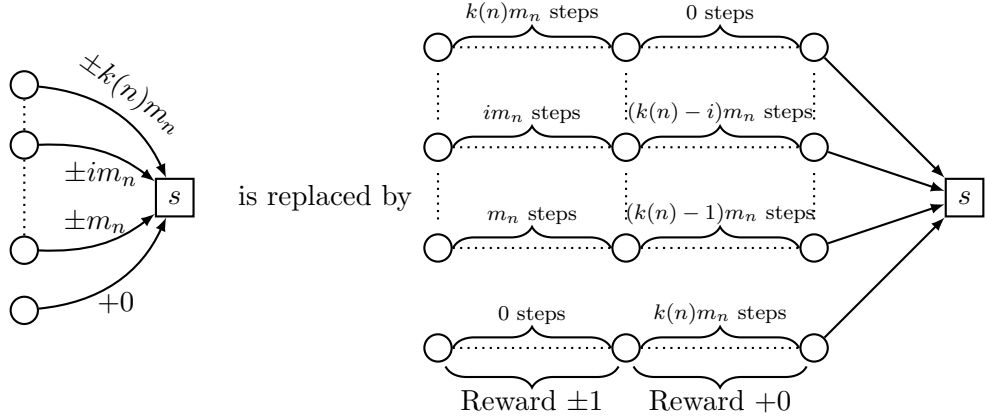
\begin{figure*}
\begin{center}
\begin{tikzpicture}

\node[draw, minimum height=0.5cm, minimum width=0.5cm] (C1) at (2,0) {$s$};

\node[draw, circle] (Q1) at (0,1.5) {};
\node[draw,circle] (Q2) at (0,0.7){};
\node[draw,circle] (Q3) at (0,-0.7){};
\node[draw, circle] (Q4) at (0,-1.5) {};

\draw [dotted, thick] (Q3) -- (Q2);
\draw [dotted, thick] (Q2) -- (Q1);

\draw[->,>=latex]
                    (Q1) edge[bend left] node[above, sloped, midway]{$\pm k(n)m_{n}$} (C1)
                    (Q2) edge[bend left=20] node[below, midway]{$\pm im_{n}$} (C1)
                    (Q3) edge[bend right=20] node[above, midway, xshift=-2pt]{$\pm m_{n}$} (C1)
                    (Q4) edge [bend right] node[below, midway]{$+0$} (C1);

\node (T) at (4,0) {is replaced by};


\node[draw, minimum width=0.5cm, minimum height=0.5cm] (C2) at (12.5,0) {$s$};

\node[draw, circle] (P1) at (5.5,2) {};
\node[draw, circle] (P2) at (5.5,0.67) {};
\node[draw, circle] (P3) at (5.5,-0.67) {};
\node[draw, circle] (P4) at (5.5,-2) {};

\node[draw, circle] (P5) at (8,2) {};
\node[draw, circle] (P6) at (8,0.67) {};
\node[draw, circle] (P7) at (8,-0.67) {};
\node[draw, circle] (P8) at (8,-2) {};

\node[draw, circle] (P9) at (10.5,2) {};
\node[draw, circle] (P10) at (10.5,0.67) {};
\node[draw, circle] (P11) at (10.5,-0.67) {};
\node[draw, circle] (P12) at (10.5,-2) {};

\node[] (I1) at (5.5,1.8) {};
\node[] (I2) at (5.5,0.87) {};
\node[] (I3) at (5.5,0.47) {};
\node[] (I4) at (5.5,-0.47) {};

\node[] (I5) at (8,1.8) {};
\node[] (I6) at (8,0.87) {};
\node[] (I7) at (8,0.47) {};
\node[] (I8) at (8,-0.47) {};

\node[] (I9) at (10.5,1.8) {};
\node[] (I10) at (10.5,0.87) {};
\node[] (I11) at (10.5,0.47) {};
\node[] (I12) at (10.5,-0.47) {};

\draw[dotted, thick] (P1) -- (P5);
\draw[dotted, thick] (P2) -- (P6);
\draw[dotted, thick] (P3) -- (P7);
\draw[dotted, thick] (P4) -- (P8);

\draw[dotted, thick] (P5) -- (P9);
\draw[dotted, thick] (P6) -- (P10);
\draw[dotted, thick] (P7) -- (P11);
\draw[dotted, thick] (P8) -- (P12);

\draw[dotted, thick] (I1) -- (I2);
\draw[dotted, thick] (I3) -- (I4);
\draw[dotted, thick] (I5) -- (I6);
\draw[dotted, thick] (I7) -- (I8);

\draw[dotted, thick] (I9) -- (I10);
\draw[dotted, thick] (I11) -- (I12);

\draw[->,>=latex]
(P9) edge (C2)
(P10) edge (C2)
(P11) edge (C2)
(P12) edge (C2);

\draw [decorate, decoration={brace,amplitude=8pt},xshift=0pt,yshift=0pt]
(5.7,2) -- (7.8,2) node [black,midway,above, yshift=5pt] {\scriptsize $k(n) m_{n}$ steps};
\draw [decorate, decoration={brace,amplitude=8pt},xshift=0pt,yshift=0pt]
(5.7,0.67) -- (7.8,0.67) node [black,midway,above, yshift=5pt] {\scriptsize $im_{n}$ steps};
\draw [decorate, decoration={brace,amplitude=8pt},xshift=0pt,yshift=0pt]
(5.7,-0.67) -- (7.8,-0.67) node [black,midway,above, yshift=5pt] {\scriptsize $m_{n}$ steps};
\draw [decorate, decoration={brace,amplitude=8pt},xshift=0pt,yshift=0pt]
(5.7,-2) -- (7.8,-2) node [black,midway,above, yshift=5pt] {\scriptsize $0$ steps};

\draw [decorate, decoration={brace,amplitude=8pt},xshift=0pt,yshift=0pt]
(8.2,2) -- (10.3,2) node [black,midway,above, yshift=5pt] {\scriptsize $0$ steps};
\draw [decorate, decoration={brace,amplitude=8pt},xshift=0pt,yshift=0pt]
(8.2,0.67) -- (10.3,0.67) node [black,midway,above, yshift=5pt] {\scriptsize $(k(n)-i)m_{n}$ steps};
\draw [decorate, decoration={brace,amplitude=8pt},xshift=0pt,yshift=0pt]
(8.2,-0.67) -- (10.3,-0.67) node [black,midway,above, yshift=5pt] {\scriptsize $(k(n)-1)m_{n}$ steps};
\draw [decorate, decoration={brace,amplitude=8pt},xshift=0pt,yshift=0pt]
(8.2,-2) -- (10.3,-2) node [black,midway,above, yshift=5pt] {\scriptsize $k(n)m_n$ steps};

\draw [decorate, decoration={brace,mirror,amplitude=8pt},xshift=0pt,yshift=-7pt]
(5.6,-2) -- (7.9,-2) node [black,midway,below, yshift=-5pt] {Reward $\pm 1$};
\draw [decorate, decoration={brace,mirror,amplitude=8pt},xshift=0pt,yshift=-7pt]
(8.1,-2) -- (10.4,-2) node [black,midway,below, yshift=-5pt] {Reward $+0$};

\end{tikzpicture}
\caption{Schema for replacing large rewards with bounded rewards in \Cref{infinitegadget}.
}%
\label{boundedreward}
\end{center}
\end{figure*}
We now further alter \Cref{infinitegadget} by bounding the rewards.
The rewards on transitions are now limited to $-1$, $0$ or $1$.
To compensate for the smaller rewards, in the $n$-th gadget, each transition bearing a reward is replaced by $k(n) \cdot m_{n}$ transitions as follows. If the original transition had reward $j \cdot m_{n}$ then that transition is replaced with $j \cdot m_{n}$ transitions with reward $1$, and $(k(n) - j) \cdot m_{n}$ transitions with reward $0$. Symmetrically all negatively weighted transitions are similarly replaced by transitions with rewards $-1$ and $0$.
The extra padding with the transitions with reward $0$ is done in order to
preserve the path length, i.e., such that the step counter is still implicit
in the state.

\begin{defi}[Bounded rewards]\label{def:boundedrewards}
We formally define how to modify the MDPs in \Cref{sec:mplowerbounds} such that their rewards are either $-1$, $+1$ or $0$.
In each \Cref{infinitegadget} gadget, we do as illustrated in \Cref{kbranchtobinarybranch}. I.e.\ the incoming transitions to $c_n$ and $s_n$ carry rewards $\pm im_n$ for some $i$ with $0 \leq i \leq k(n)$. We then split this transition carrying reward $\pm im_n$ into a chain of $k(n)m_n$ transitions. The first $i m_n$ of which carry reward $\pm 1$, and the last $(k(n)-i)m_n$ of which carry reward $0$.

The extended path lengths are mirrored in a modified \Cref{chain} by padding the path lengths by an extra $2k(n)m_n$
steps.
Because skipping ahead in \Cref{chain} reimburses reward $+i$ upon entering state $s_{N^{*}+i+1}$, we replace these transitions with $i$ transitions bearing reward $+1$ and reflect this increased path length by padding the incoming transitions to $s_{N^{*}+i+1}$ with an extra $i$ transitions bearing reward $0$.

The extended path lengths must also be reflected in \Cref{restart}. This is done by replacing transitions carrying reward $\pm m_i$ by $m_i$ transitions carrying reward $\pm 1$. We also increase the number of black states from $3$ to $2k(n)m_n + 3$ to match the number of steps taken inside the $n$th gadget.
\end{defi}

\begin{thm}
There exists a countable, acyclic MDP $\mdp$, whose step counter is implicit in the state, whose transition probabilities are rational, whose rewards on transitions are in $\{-1, 0, 1\}$ and whose branching degree is bounded by $2$ for which
$\valueof{\mdp,\liminftpobj}{s_{0}} = 1$ and any FR strategy $\sigma$ is such that
$\probm_{\mdp, s_{0}, \sigma}(\liminftpobj)=0$.
In particular, there are no $\eps$-optimal step counter plus finite memory strategies
for any $\varepsilon < 1$ for the
$\liminftpobj$ objective for countable MDPs.
\end{thm}

\begin{proof}
This follows from \Cref{infwin}, \Cref{inflose} by modifying the constructions in \Cref{infinitegadget} and \Cref{chain} as detailed in \Cref{def:binarybranching}, \Cref{def:boundedrewards} and \Cref{def:rationalprobs}.
\end{proof}

\begin{thm}
There exists a countable, acyclic MDP $\mdp$, whose step counter is implicit in the state, whose transition probabilities are rational, whose rewards on transitions are in $\{-1, 0, 1\}$ and whose branching degree is bounded by $2$ for which
$s_{0}$ is almost surely winning for $\liminftpobj$ and any FR strategy $\sigma$ is such that
$\probm_{\mdp, s_{0}, \sigma}(\liminftpobj)=0$.
In particular, almost sure winning strategies, when they exist, cannot be chosen
with a step counter plus finite memory for countable MDPs.
\end{thm}

\begin{proof}
This follows from \Cref{almostwin}, \Cref{almostlose} by modifying the constructions in \Cref{infinitegadget} and \Cref{restart} as detailed in \Cref{def:binarybranching}, \Cref{def:boundedrewards} and \Cref{def:rationalprobs}.
\end{proof}

\begin{rem} 
We draw attention to the fact that we could state the total payoff theorems using bounded rewards, but we did not do so for the equivalent mean payoff results.
In the case of mean payoff, with the step counter implicit in the state,
having bounded transition rewards, e.g.\ bounded by $\pm b$,
means that the average reward in any given state will always be bounded by $\pm b$.
In the context of our example in \Cref{infinitegadget}, this means that if we used the construction in \Cref{def:boundedrewards}, the absolute worst the average reward can be is $\sim -1$, this can only happen going from the $0$th random transition to the $k(n)$th choice.
But even worse, using only one bit of memory to remember whether the random transition was $> \tfrac{k(n)}{2}$ or not, the mean payoff is suddenly at worst $\sim -\tfrac{1}{k(n)}$ which converges to $0$.

In general it would be interesting to consider the mean payoff objective with step counter encoded and bounded rewards since our results do not obviously carry over to this case.
\end{rem}

\bigskip

Some extra care is needed to convince ourselves that \Cref{mpstepepslower} and \Cref{mpstepoptlower} can also be strengthened. Consider the construction in \Cref{stepcounter}. In the random choice, the transition rewards are already all $0$, so only the branching degree needs to be adjusted by padding the choice with a binary tree as above. In the controlled choice, the transitions carrying reward $ \pm m_{n}^{i}$ are replaced by $m_{n}^{i}$ transitions each bearing reward $\pm 1$ respectively.
Therefore, the path lengths increase in the following way in the $n$-th
gadget. In $s_{n}$ and $c_{n}$, the binary trees increase path length by up to
$\lceil\lg(k(n)+1)\rceil$ (where $\lg$ is the logarithm to base $2$)
and after $c_{n}$ the path length increases by up to $m_{n}^{k(n)}$ twice.

Consider the scenario where the play took the $i$-th random choice and the
player makes the `best' mistake where they choose transition $i+1$.
We show that, even in this best error case (and thus in all other error
cases), the newly added path lengths do
still not help to prevent seeing a mean payoff $\le -1/2$ in the $n$-th gadget.
In this case, in the state between $c_{n}$ and $s_{n+1}$, the total payoff is $-m_{n}^{i+1}$ and the total number of steps taken by the play so far is upper bounded by
\[
\beta_n \eqdef \left( \sum_{i=N^{*}}^{n-1} 2\lceil\lg(k(i)+1)\rceil + 2m_{i}^{k(i)} \right)
+ 2\lceil\lg(k(n)+1)\rceil + m_{n}^{i} + m_{n}^{i+1}.
\]
Recall that $m_{n} \defeq \sum_{i = N^{*}}^{n-1} m_{i}^{k(n)}$ with $m_{N^{*}} \defeq 1$, and this is the definition of $m_{n}$ from \Cref{stepcounter} which is different from the definition of $m_{n}$ in \Cref{infinitegadget}. Note that $k(n)$ is very slowly growing, so it follows that
\[
\beta_n \leq 3m_{n} + m_{n}^{i} + m_{n}^{i+1} \leq 2m_{n}^{i+1}.
\]
That is to say that the mean payoff is $\le \dfrac{-m_{n}^{i+1}}{2m_{n}^{i+1}}
= -1/2$. As a result, in the case of a bad aggressive decision,
the mean payoff will still drop below $-1/2$
in this modified MDP (instead of dropping below $-1$ in the original MDP).
This is just as good to falsify $\liminfmpobj$.

Thus we obtain the following two results.

\begin{thm}
There exists a countable, acyclic MDP $\mdp$, whose reward counter is implicit in the state, whose transition probabilities are rational, whose rewards on transitions are in $\{-1, 0, 1\}$ and whose branching degree is bounded by $2$ for which
$\valueof{\mdp,\liminfmpobj}{s_{0}} = 1$ and any FR strategy $\sigma$ is such that
$\probm_{\mdp, s_{0}, \sigma}(\liminfmpobj)=0$.
In particular, there are no $\eps$-optimal step counter plus finite memory strategies
for any $\eps < 1$ for the
$\liminfmpobj$ objective for countable MDPs.
\end{thm}

\begin{proof}
This follows from \Cref{liminfmpstepval1}, \Cref{liminfmpstepval0}, \Cref{def:binarybranching}, \Cref{def:boundedrewards} and \Cref{def:rationalprobs}.
\end{proof}

\begin{thm}
There exists a countable, acyclic MDP $\mdp$, whose reward counter is implicit in the state, whose transition probabilities are rational, whose rewards on transitions are in $\{-1, 0, 1\}$ and whose branching degree is bounded by $2$ for which
$s_{0}$ is almost surely winning for $\liminfmpobj$ and any FR strategy $\sigma$ is such that
$\probm_{\mdp, s_{0}, \sigma}(\liminfmpobj)=0$.
In particular, almost sure winning strategies, when they exist, cannot be chosen
with a step counter plus finite memory for countable MDPs.
\end{thm}

\begin{proof}
This follows from \Cref{liminfmpstepam1}, \Cref{liminfmpstepam0}, \Cref{def:binarybranching}, \Cref{def:boundedrewards} and \Cref{def:rationalprobs}.
\end{proof}

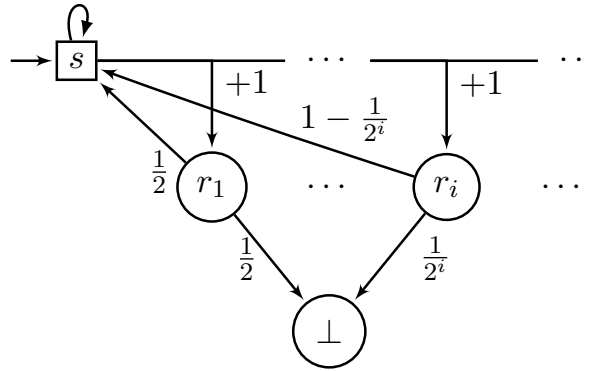
\begin{figure}
\begin{center}
\scalebox{1.2}{
\begin{tikzpicture}[>=latex',shorten >=1pt,node distance=1.9cm,on grid,auto,
roundnode/.style={circle, draw,minimum size=1.5mm},
squarenode/.style={rectangle, draw,minimum size=2mm},
diamonddnode/.style={diamond, draw,minimum size=2mm}]

\node [squarenode,initial,initial text={}] (s) at(0,0) [draw]{$s$};

\node[roundnode] (r1)  [below right=1.4cm and 1.5cm of s] {$r_1$};
\node[roundnode] (r3)  [draw=none,right=1.3 of r1] {$\cdots$};
\node[roundnode] (r4)  [right=1.3cm of r3] {$r_i$};
\node[roundnode] (r5)  [draw=none,right=1.3cm of r4] {$\cdots$};

\node [roundnode,inner sep = 4pt] (t)  [below=1.6cm of r3] {$ \perp $};

\draw [->] (s) -- ++(1.5,0)  -- node[near start, right]{$+1$} (r1);
\node[roundnode] (d)  [draw=none,right=2.8 of s] {$\cdots$};
\draw[-] (s) -- (d);
\draw [->] (d) -- ++(1.3,0) -- node[near start, right]{$+1$} (r4);
\node[roundnode] (dd)  [draw=none,right=2.8 of d] {$\cdots$};
\draw[-] (d) -- (dd);

\path[->] (r1) edge node [midway,left] {$\frac{1}{2}$} (t);
\path[->] (r4) edge node [midway,right=.2cm] {$\frac{1}{2^{i}}$} (t);
\path[->] (r1) edge  node[pos=0.3,below] {$\frac{1}{2}$} (s);
\path[->] (r4) edge [bend left=1] node [pos=0.2,above] {$1-\frac{1}{2^{i}}$} (s);

\path (s) edge[->,>=latex, loop above] (s);

\end{tikzpicture}
}
\caption{
An MDP where $\eps$-optimal strategies for $\liminfmpobj$ require only memory that grows unboundedly with the number of steps taken so far.
}%
\label{growingmemory}
\end{center}
\end{figure}

\begin{rem}
The result from \Cref{inflose} holds even for strategies $\sigma$ whose memory
grows unboundedly, but slower than $k(n)-1$. That is to say that there exists a countable, acyclic MDP $\mdp$, whose step counter is implicit in the state such that $\valueof{\mdp,\liminfmpobj}{s_{0}} = 1$ and any strategy $\sigma$ with number of memory modes $< k(n) - 1$ in the $n$th gadget is such that
$\probm_{\mdp, \sigma, s_{0}}(\liminfmpobj)=0$.
This follows from a slightly modified version of \Cref{alglose} which considers the situation where states $i(n)$ and $k(n)$ are confused in the player's memory. Then the argument used in \Cref{inflose} can be modified to include $i(n), j(n): \mathbb{N} \to \{ 0, 1, \ldots, k(n)\}$.
The result then follows since in every gadget at least one memory mode will confuse at least two states $i(n), j(n): \N \to \{ 0, 1, \ldots, k(n)-1\}$, which as we have shown is enough to falsify $\liminfmpobj$.

This is in contrast to examples such as \Cref{growingmemory} where the only requirement on the memory is that it grow unboundedly.
\end{rem}

\section{Conclusion and Outlook}\label{sec:conclusion}
We have established matching lower and upper bounds on the strategy complexity
of $\liminf$ threshold objectives for point, total and mean payoff on countably infinite
MDPs; cf.~\Cref{table:allresults}.

The upper bounds hold not only for integer transition rewards, but also
for rationals or reals, provided that the reward counter (in those cases where one
is required) is of the same type.
The lower bounds hold even for integer transition rewards, since all our
counterexamples are of this form.

Directions for future work include the corresponding questions for $\limsup$
threshold objectives. While the $\liminf$ point payoff objective generalizes co-B\"uchi
(see \Cref{sec:prelim}), the $\limsup$ point payoff objective generalizes
B\"uchi. Thus the lower bounds for $\limsup$ point payoff are at least as high as
the lower bounds for B\"uchi objectives~\cite{KMST:ICALP2019,KMST2020c}.

\bibliographystyle{alphaurl}
\bibliography{conferences,refs}
\end{document}